\documentclass[a4paper,twocolumn,10pt,accepted=2025-03-05]{quantumarticle}
\pdfoutput=1
\usepackage[utf8]{inputenc}
\usepackage[english]{babel}
\usepackage[T1]{fontenc}

\usepackage{amsmath}
\usepackage[hidelinks]{hyperref}

\usepackage{tikz}
\usepackage{lipsum}

\usepackage{qcircuit}
\usepackage{enumitem}
\usepackage{amsmath}
\usepackage{amssymb}
\usepackage{amsthm}
\usepackage{amsfonts}
\usepackage{braket}
\usepackage{dsfont}
\usepackage{textcomp,gensymb}
\usepackage{bm}
\usepackage{color}
\usetikzlibrary{decorations.pathmorphing}
\usepackage{booktabs}
\usepackage{multirow}
\usepackage{nicefrac}
\usepackage{oubraces}

\usepackage{rotating}

\usepackage{graphicx, import}
\usepackage{soul}
\usepackage[utf8]{inputenc}
\usepackage[export]{adjustbox}
\usepackage[caption=false]{subfig}
\usepackage{parskip} 
\usepackage{hyperref}
\usepackage{url}
\usepackage{algorithm}
\usepackage{algpseudocode}
\usepackage{qcircuit}
\usepackage[numbers, sort&compress]{natbib}

\newtheorem{theorem}{Theorem}
\newtheorem{definition}[theorem]{Definition}
\newtheorem{corollary}[theorem]{Corollary}
\newtheorem{lemma}[theorem]{Lemma}
\newtheorem{proposition}[theorem]{Proposition}
\newtheorem{remark}[theorem]{Remark}

\newcommand{\I}{\mathds{1}}

\DeclareMathOperator*{\bin}{bin}

\DeclareMathOperator*{\CNOT}{CNOT}
\DeclareMathOperator*{\PSWAP}{PSWAP}
\DeclareMathOperator*{\CPSWAP}{CPSWAP}
\DeclareMathOperator{\var}{Var}

\DeclareMathOperator*{\sgn}{sgn}

\def\CNOT{\textnormal{CNOT}}
\def\CnNOT#1{\textnormal{C}^{#1}\textnormal{NOT}}
\def\CnPSWAP#1{\textnormal{C}^{#1}\textnormal{PSWAP}}
\def\CnZ#1{\textnormal{C}^{#1}Z}

\def\fib#1#2{\frac{\phi^{#1}-(-\phi)^{#2}}{\sqrt{5}}}

\def\ketbra#1#2{{\vert#1\rangle\!\langle#2\vert}}
\def\hc{\text{h.c.}}
\def\hil{\mathcal{H}}
\newcommand{\bH}{\mathbf{H}}
\newcommand{\bS}{\mathbf{S}}

\begin{document}

\title{Solving lattice gauge theories using the quantum Krylov algorithm and qubitization}

\author{Lewis~W.~Anderson} 
\email{lewis.anderson@ibm.com}
\affiliation{Clarendon Laboratory, University of Oxford, Parks Road, Oxford OX1 3PU, UK}
\affiliation{IBM Research Europe, Hursley, Winchester SO21 2JN, UK}
\orcid{0000-0003-0269-3237}

\author{Martin~Kiffner}
\affiliation{Clarendon Laboratory, University of Oxford, Parks Road, Oxford OX1 3PU, UK}
\affiliation{PlanQC GmbH, Lichtenbergstr. 8, 85748 Garching, Germany}
\orcid{0000-0002-8321-6768}

\author{Tom O'Leary}
\affiliation{Clarendon Laboratory, University of Oxford, Parks Road, Oxford OX1 3PU, UK}
\orcid{0009-0003-2065-1695}

\author{Jason~Crain}
\affiliation{IBM Research Europe, The Hartree Centre STFC Laboratory, Sci-Tech Daresbury, Warrington WA4 4AD, UK}
\affiliation{Clarendon Laboratory, University of Oxford, Parks Road, Oxford OX1 3PU, UK}
\orcid{0000-0001-8672-9158}

\author{Dieter~Jaksch}
\affiliation{Institut f\"ur Quantenphysik, Universit\"at Hamburg, 22761 Hamburg, Germany}
\affiliation{Clarendon Laboratory, University of Oxford, Parks Road, Oxford OX1 3PU, UK}
\orcid{0000-0002-9704-3941}

\maketitle

\begin{abstract}
Computing vacuum states of lattice gauge theories (LGTs) containing fermionic degrees of freedom can present significant challenges for classical computation using Monte-Carlo methods. Quantum algorithms may offer a pathway towards more scalable computation of groundstate properties of LGTs.
However, a comprehensive understanding of the quantum computational resources required for such a problem is thus far lacking.
In this work, we investigate using the quantum subspace expansion (QSE) algorithm to compute the groundstate of the Schwinger model, an archetypal LGT describing quantum electrodynamics in one spatial dimension.
We perform numerical simulations, including the effect of measurement noise, to extrapolate the resources required for the QSE algorithm to achieve a desired accuracy for a range of system sizes.
Using this, we present a full analysis of the resources required to compute LGT vacuum states using a quantum algorithm using qubitization within a fault tolerant framework.
We develop of a novel method for performing qubitization of a LGT Hamiltonian based on a ``linear combination of unitaries'' (LCU) approach. The cost of the corresponding block encoding operation scales as $\tilde{\mathcal{O}}(N)$ with system size $N$.  Including the corresponding prefactors, our method reduces the gate cost by multiple orders of magnitude when compared to previous LCU methods for the QSE algorithm, which scales as $\tilde{\mathcal{O}}(N^2)$ when applied to the Schwinger model.
While the qubit and single circuit T-gate cost resulting from our resource analysis is appealing to early fault-tolerant implementation, we find that the number of shots required to avoid numerical instability within the QSE procedure must be significantly reduced in order to improve the feasibility of the methodology we consider and discuss how this might be achieved.
\end{abstract}

\section{Introduction}

Gauge theories are quantum field theories (QFTs) in which a gauge field imposes the fundamental symmetry on the Lagrangian and mediates interactions.
They play foundational roles within various fields of theoretical physics. Within particle physics, gauge theories describe fundamental particles and their interaction through the exchange of force carriers. The gauge symmetry determines the nature of the theory; with quantum electrodynamics (QED) arising from an abelian $\textrm{U}(1)$ symmetry and quantum chromodynamics (QCD) from a non-abelian $\textrm{SU}(3)$ symmetry~\cite{schwartz2014quantum}. Gauge field theories also play important roles in cosmology~\cite{bailin2004cosmology} and condensed matter physics~\cite{fradkin2013field}.

Many gauge theories are too complex to study analytically and, because many phenomena of interest arise non-perturbatively, efficient numerical methods have become increasingly important for studying gauge theories. Originally formulated in seminal work by Wilson~\cite{wilson1974confinement}, numerical study of such phenomena is now routinely performed using lattice gauge theories (LGTs) in which spacetime is discretized onto a lattice of points on which the values of matter and gauge fields are defined. Wilson's path integral formulation of LGTs saw extensive computational study using Monte-Carlo methods~\cite{creutz1983monte}. Applying Monte-Carlo algorithms to some Hamiltonians however, leads to the so called \textit{sign problem} or \textit{complex action problem}~\cite{troyer2005computational} in which the integrand contains negative terms or the action becomes complex (and so the Boltzmann weight cannot be directly interpreted as a probability density function). In dealing with these issues, for example by using a Taylor expansion~\cite{sandvik1991quantum}, computing expectation values can involve summing many highly oscillatory terms of which many cancel due to alternating signs. This behaviour leads to statistical error that grows exponentially with system size~\cite{troyer2005computational}. Due to the minus sign in the fermionic commutation relation the sign problem poses a significant challenge to simulations involving interacting electrons, as well as lattice QCD with non-zero Baryon density.

One of the most widely studied gauge field theories is the Schwinger model~\cite{schwinger1962gauge} which describes charged particles and their interactions with electromagnetic fields within QED in 1+1D\footnote{one space and one time dimension}. Being the simplest gauge theory, the Schwinger model plays an important role as a benchmark for numerical methods for solving gauge theories as well as a pedagogical tool that can be used to study a range of complex phenomena that arise in more complex quantum field theories. Such phenomena include: chiral symmetry breaking, in which an effective fermionic mass arises from  pairing of fermion and anti-fermions~\cite{chen1990chiral}; confinement, where fermions only exist alongside their corresponding anti-particle and separating such pairs requires increasing energy as they are drawn further apart (analogous to quark confinement in 3+1D QCD)~\cite{melnikov2000lattice}.

Simulating the real-time dynamics of the Schwinger model as well as the ground state in the presence of a chemical potential imbalance between different fermion flavours both give rise to the sign problem. Consequently, the lattice Schwinger model has been used to benchmark novel computational approaches which do not suffer from the sign problem\footnote{this is not to say that such computational approaches do not lead to their own computational bottlenecks}. A promising classical approach is to use tensor network methods~\cite{buyens2015tensor, banuls2018tensor} which use a Hamiltonian, rather than path-integral, representation and thus don't suffer from the sign problem. Matrix product states (MPS), a popular tensor network ansatz, are particularly well suited to gapped locally interacting systems in one spatial dimension and have seen significant application to studying lattice gauge theories~\cite{buyens2014matrix, zapp2017tensor}. MPS have been used to study mass spectra in cases with zero and various values of non-zero bare fermion mass~\cite{banuls2013mass}, CP-violating phase transitions~\cite{ercolessi2018phase, funcke2023exploring}, chiral condensation~\cite{banuls2016chiral}, ground and excited state energies in the thermodynamic limit~\cite{zache2022toward}, ground state entanglement structure~\cite{banuls2017efficient} and real time dynamics~\cite{pichler2016real, buyens2017real, magnifico2020real, rigobello2021entanglement, halimeh2022achieving} in the Schwinger model and related models with abelian $\mathbb{Z}_n$ or $\mathrm{SU}(2)$ symmetries. Despite their successes in 1+1D however, extending tensor network methods to QFTs in two or more spatial dimensions poses significant challenges. MPS methods in more than one dimension often require bond-dimensions that grow rapidly with system size and tensor networks embedded in higher dimensions, such as projected entangled pairs states (PEPS)~\cite{zapp2017tensor}, cannot be contracted efficiently. 

Quantum computing may offer a pathway towards simulation of lattice gauge theories when the sign problem arises and tensor networks methods fail~\cite{bauer2023quantum, di2024quantum}. Multiple works have investigated algorithmic approaches to lattice gauge theories~\cite{banuls2020simulating}, and the Schwinger model in particular. These quantum computing methods typically aim to solve one of two problems. The first is to simulate time evolution of the model. Notably Shaw et al.~\cite{shaw2020quantum} provide rigorous bounds for Trotterised time evolution of the model in both NISQ and fault tolerant settings which scale as $\tilde{\mathcal{O}}(N^{3/2}T^{3/2}\sqrt{x}\Lambda)$, where $N$ is system size, $T$ is simulation time, $x$ is interaction strength and $\Lambda$ is the local dimension of the gauge field (see Sec.~\ref{subsec: schwinger hamiltonian} for further details of these model parameters)\footnote{They note that this scaling with $\Lambda$ is more favourable than other methods based on qubitization~\cite{low2019hamiltonian} or random compilation for Hamiltonian simulation (QDRIFT)~\cite{campbell2019random}.}. Recently, Kan and Nam reported an algorithm based on a second order product-formula which recovers the scaling reported by Shaw et at al. and can be applied to higher dimensions~\cite{kan2022simulating}. Other authors have focused on presenting experimental implementations using quantum hardware or classical emulation methods. Notable examples of such work have demonstrated real time dynamics of the Schwinger model using Trotterisation~\cite{klco2018quantum, nguyen2022digital}; the largest simulation of which was performed using 112 qubits on an IBM superconducting device~\cite{farrell2024quantum}; and open quantum system dynamics~\cite{de2022quantum}. The second problem that has been investigated using quantum computers is finding groundstate, or vacuum states, of LGTs.

Research studying the groundstate computation for lattice gauge theories, and in particular the Schwinger model, on quantum computers has primarily focused on the use of VQE methods~\cite{klco2018quantum, kokail2019self, avkhadiev2020accelerating, ferguson2021measurement, thompson2022quantum, funcke2022exploring, mazzola2021gauge, paulson2021simulating, lumia2022two, meth2023simulating, farrell2023scalable}, including as part of initial state preparation for studying dynamics~\cite{farrell2024quantum}. The scalability of VQE is significantly impacted by barren plateaus in which gradients of the cost functions landscape become exponentially small with system size~\cite{mcclean2018barren, wang2021noise, cerezo2021cost, marrero2021entanglement, uvarov2021barren, holmes2022connecting, fontana2023adjoint, ragone2023unified}. Moreover, the highly stochastic and time-intensive optimisation procedure of VQE means that systematic studies of the cost and scaling of the optimisation routine are challenging using both simulated and real quantum processors. Other quantum algorithms for computing ground states of lattice field theories include methods based on building the state up site-by-site~\cite{moosavian2019site} and using a quantum circuit for an MPS ansatz~\cite{hamed2018faster}, which have thus far only been applied to purely Fermionic (i.e., that do no contains explicit gauge fields) Hamiltonians; as well as adiabatic methods~\cite{jordan2014quantum}, which require a iteratively performing a costly phase estimation step to ensure the state does not leave the adiabtic manifold.

In order to avoid these issues, we investigate the use of the recently proposed quantum subspace expansion (QSE)~\cite{mcclean2017hybrid, mcclean2020decoding, yoshioka2022generalized, yoshioka2022variational, kirby2024analysis} algorithm and apply it to solving the lattice Schwinger model. In contrast to VQE, since the QSE algorithm does not involve a stochastic optimisation procedure it avoids the issue of barren plateaus (although obstacles to the scaling of the algorithm may exist) and analysis of the performance and scaling of the algorithm with numerical methods may be more fruitful. Recent work has shown that qubitization~\cite{low2019hamiltonian} can be used to implement QSE~\cite{kirby2023exact}, as well as the related a algorithm for computing Green's functions using Lanczos recursion~\cite{baker2021lanczos}, on fault tolerant processors. Moreover, the theory underpinning QSE in such a setting is closely related that of the quantum eigenvalue transformation for unitary matrices (QETU) algorithm~\cite{dong2022ground}, which has been studied for $\textrm{U}(1)$ gauge theories in 2+1D via both Trotterisation and qubitization~\cite{kane2024nearly}.

\subsection{Summary of results}

In this article, we investigate the use of QSE for finding groundstates of a LGT. We consider all-to-all qubit connectivity and fault tolerant operation throughout\footnote{later, we compare the required two-qubit gate depth to decoherence times for a range of quantum processors to confirm the need for fault tolerance in systems of hundreds of lattice sites.}. We lay out a framework in which methods based on qubitization and quantum singular value transformation can be used to generate the necessary expectation values required for the lattice Schwinger model. We provide a full description and resource estimation for the qubitization procedure which is based on the linear combination of unitaries (LCU) method. As part of this, we propose a novel way of applying the LCU procedure to lattice models in which we exploit translational symmetry of the model. Compared to existing proposals for performing QSE with qubitization~\cite{kirby2023exact}, which can be applied to the Schwinger model with a gate cost of $\tilde{\mathcal{O}}(N^2)$, the gate cost of our methods is $\tilde{\mathcal{O}}(N)$ and leads to a reduction in gate cost by multiple orders of magnitude for system sizes of hundreds of lattice sites.

To estimate the total resource cost of the QSE algorithm, we perform extensive numerical simulation of the algorithm, including the effect of measurement noise for systems of up to $N=26$ lattice sites. We consider the effect of shot noise on expectation values and use a partitioned quantum subspace~\cite{oleary2024partitioned} to mitigate issues due to ill-conditioning within the QSE algorithm resulting from this noise. Extrapolating from the numerical results we estimate the qubit and T-gate cost of using QSE to solve the lattice Schwinger model for lattice sizes comparable to state of the art classical results. For systems of hundreds of lattice sites, we find that the number of T-gates required for the qubitization procedure are on the order of $10^6$ (more precise values are given later in the paper). While such single circuit gate costs may be feasible for future fault tolerant devices, the number of shots required to control measurement noise leads to total number of circuit evaluations (and therefore runtime) beyond what we expect to be feasible. This is despite the use of partitioned quantum subspace approaches~\cite{oleary2024partitioned}, which reduces the cost significantly compared to previous thresholded approaches~\cite{epperly2022theory}. We discuss how further algorithmic improvements may offer a pathway towards reducing the measurement cost.

Although we consider the single flavour Schwinger model, for which the groundstate problem does not lead to the sign problem, the analysis and numerics we present here can be extended to more complex models for which Monte-Carlo methods fail. The work we present here serves as a first step towards understanding and using a new quantum algorithm to study LGTs.

The remainder of this article is structured as follows. In Sec.~\ref{sec: theory} we describe the theoretical background. This includes Sec.~\ref{subsec: qse} describing the QSE algorithm and Sec.~\ref{subsec: schwinger hamiltonian} describing the Hamiltonian formulation of the lattice Schwinger model. In Sec.~\ref{sec: qsvt for qse} we introduce the concept of quantum singular values transforms (QSVT) and qubitization. This includes Sec.~\ref{subsec: block encoding} defining a block encoding and QSVT and then Sec.~\ref{subsec: krylov basis procedure} in which we show how this can be used to generate a Krylov basis for QSE. In Sec.~\ref{sec: numerical experiments} we present our numerical simulations. This includes a description of how we model measurement noise in Sec.~\ref{subsec: measurement noise}, how we use partitioned quantum subspace expansion to deal with this noise in Sec.~\ref{subsec: dealing with noise}, and a set of results in Sec.~\ref{subsec: qse results} showing the number of measurements and basis size needed to achieve a desired energy error.  In Sec.~\ref{sec: total resources} we calculate the resources required to run the algorithm, the mathematical details of which are deferred to Appendix~\ref{app: resource cost schwinger}. Finally, in Sec.~\ref{sec: qse disc} we discuss our results.

\section{Theoretical background}\label{sec: theory}

\subsection{Quantum subspace expansion}\label{subsec: qse}

For a Hamiltonian $H$, the QSE algorithm aims to calculate an approximate groundstate within a subspace spanned by a $D$ dimensional non-orthogonal basis $\{\ket{\psi_0}, \ket{\psi_1}, \dots,\ket{\psi_{D-1}}\}$ which can be generated by a quantum processor. Within this subspace, a trial state has the form $\ket{\psi(\vec{c})} = \sum_{i=0}^{D-1}c_i\ket{\psi_i}$. Minimising the functional $\bra{\psi(\vec{c})}H\ket{\psi(\vec{c})}$ to find optimal coefficients $\vec{c}$ is equivalent to solving the generalised eigenvalue problem 

\begin{equation}\label{eqn: qse gev}
    \bH\vec{c} = E\bS\vec{c},
\end{equation}
for minimum eigenvalue $E$, where matrices $\bH$ and $\bS$ are given by
\begin{equation}\label{eqn: qse matrices}
    \bH_{i,j} = \braket{\psi_i|H|\psi_j},\quad \bS_{i,j} = \braket{\psi_i|\psi_j}.
\end{equation}

The power of the QSE algorithm relies on making a good choice of basis and being able to efficiently measure expectation values for $\bH_{i,j}$ and $\bS_{i,j}$. The basis $\{\ket{\psi_k}|k=0,1,\dots,(D-1)\}$ is typically generated by applying (often repeatedly) operators to an input state $\ket{\psi_0}$, which is chosen to be close to the true groundstate. Many types of operators have been used to generate this basis including real and imaginary time evolution~\cite{parrish2019quantum, stair2020multireference, cortes2022quantum, stair2023stochastic, motta2020determining, yeter2020practical}, fermionic excitation operators~\cite{takeshita2020increasing}, Pauli operators~\cite{colless2018computation} and using parameterised quantum circuits in variational manner~\cite{huggins2020non}. Depending on the basis used, the expectation values can be calculated in different ways such as a SWAP test~\cite{barenco1997stabilization, buhrman2001quantum} or by measuring constituent Pauli expectation values and recombining as done in VQE.

A particularly powerful basis, called a Krylov basis, is formed by repeated application of the Hamiltonian to the reference state: $\{\ket{\psi_k} = H^k\ket{\psi_0}|k=0,1,\dots,D-1\}$. This forms the basis for the well studied classical Lanczos algorithm for computing groundstates~\cite{lanczos1950iteration, park1986unitary, cullum2002lanczos}. In the infinite $D$ limit, the Krylov basis includes infinite imaginary time evolution and, assuming non-zero overlap between the initial and ground states and infinite precision arithmetic, will lead to the exact groundstate. In the case of finite sized bases for a given problem instance, the Lanczos algorithm converges exponentially with basis size, and so very good solutions can often be found using reasonable basis sizes~\cite{paige1971computation, saad1980rates}. QSE using a Krylov basis is often refered to as the \textit{quantum Krylov algorithm}. Throughout the remainder of this article we will only consider QSE using a Krylov basis and refer to it as `QSE'. A schematic of this algorithm is shown in Fig.~\ref{fig: qse schematic}~(a). 

For quantum systems in a Hilbert space, the cost of exactly evaluating and storing the Hamiltonian acting on the reference states grows exponentially with system size, making the use of a quantum processor to generate the basis within a QSE procedure very appealing. Moreover, when using a Krylov basis, $\bH$ and $\bS$ are Hankel matrices, specified by $\mathcal{O}(D)$ unique values given by $\bH_{i,j} = \braket{\psi_0|H^{i+j+1}|\psi_0}$, $\bS_{i,j} = \braket{\psi_0|H^{i+j}|\psi_0}$, rather than the $\mathcal{O}(D^2)$ unique values required in general when using operators to generate the basis that do not commute with the Hamiltonian. Of course the Hamiltonian is not a unitary operator in general, and so methods to compute expectation values must use some additional classical computation to generate the Krylov basis. Strategies to compute the non-unitary expectation values in $\bH,\bS$ include classical post processing to combine expectation values from time evolutions~\cite{seki2021quantum} or methods based on qubitization~\cite{kirby2023exact}. We will use the latter in this work.

\subsection{Hamiltonian formulation of the single flavour lattice Schwinger model}\label{subsec: schwinger hamiltonian}

When working in a Hamiltonian formulation, the Schwinger model is often represented on a lattice by using a ``staggered fermion'' representation~\cite{banks1976strong, susskind1977lattice}. As we will show in the following, the staggered fermion representation uses odd lattice sites to encode fermions and even sites to encode anti-fermions. The single flavour model is constructed as follows. We define a single component fermion field at lattice site labelled by integer $n$ with creation and annihilation operators $\phi^\dagger(n)$ and $\phi(n)$ respectively. These obey the usual commutation relations
\begin{equation}
\{\phi^\dagger(m),\phi(n)\} = \delta_{m,n},\; \{\phi(m),\phi(n)\} = 0.
\end{equation}
The gauge field is defined on links connecting sites $(n, n+1)$ with operator
\begin{equation}
    U(n,n+1) := e^{i\theta(n)},
\end{equation}
which mediates interactions between neighbouring fermionic sites. Within the ``compact'' formalism, the onsite gauge-field becomes an angular variable on the lattice with conjugate spin variable $L(n)$ obeying
\begin{equation}
    [\theta(m),L(n)] = i\delta_{m,n}
\end{equation}
such that $L(n)$ has eigenvalues $L(n)=0,\pm1,\pm2,\dots$ and $e^{i\theta(n)}$ and $e^{-i\theta(n)}$ act as raising and lowering operators respectively. That is $e^{i\theta(n)}$ ($e^{-i\theta(n)}$) has the effect of increasing (decreasing) the value of $L(n)$ linking lattice sites $(n,n+1)$ by one, while having no effect on $L(m\neq n)$ for other lattice sites. 

The dimensionless Hamiltonian for the single flavour Schwinger model is given by
\begin{equation}
    H = H_0 + xV
\end{equation}
with
\begin{align}\label{eqn: schwinger ham ungauged}
    H_0 &= \sum_{n=1}^{N-1} L^2(n) + \mu\sum_{n=1}^N (-1)^n\phi^\dagger(n)\phi(n)\\
    V & = -i\sum_{n=1}^{N-1} \left(\phi^\dagger(n)e^{i\theta(n)}\phi(n+1) - \text{h.c}\right),
\end{align}
where $\mu$ and $x$ correspond to re-scaled mass and coupling parameters~\cite{hamer1997series}. As we will see shortly, the alternating sign on the mass term allows us to incorporate both matter and anti-matter particles on a bipartite lattice.

Physical states must obey Gauss's law which ensures that the electric flux into and out of a lattice site is balanced by the onsite charge. On the lattice, Gauss's law is given by~\cite{hamer1997series}
\begin{equation}\label{eqn: gauss}
\begin{split}
    L(n) - L(n-1) &= \phi^\dagger(n)\phi(n) - \frac{1}{2}\left(1-(-1)^n\right) \\
    &\qquad \forall \; 0 < n \leq N.
\end{split}
\end{equation}
This means that excitations $\phi^\dagger(n)\phi(n)=1$ to an even $n$ site generates $-1$ unit of flux between the gauge fields into and out of the lattice site. Thus, on an even lattice site $\phi^\dagger(n)\phi(n)=1$ is interpreted as the existence of a fermion with negative unit charge, while $\phi^\dagger(n)\phi(n)=0$ is an absence. Conversely, for an odd $n$ lattice site, $\phi^\dagger(n)\phi(n)=0$ generates $+1$ units of flux. Thus on an odd lattice site $\phi^\dagger(n)\phi(n)=0$ can be interpreted as the existence of an anti-fermion with positive unit charge and $\phi^\dagger(n)\phi(n)=1$ as the absence~\cite{hauke2013quantum}. We can see from the mass term in $\eqref{eqn: schwinger ham ungauged}$ that this interpretation gives the expected mass addition $\mu$ for both fermion and anti-fermion.

It is common to impose zero field to the boundaries, namely $L(0) = L(N) = 0$.  In the $x\rightarrow 0$ (no coupling) limit, the ground state can be found by inspection to be
\begin{equation}\label{eqn: uncoupled ground state}
    L(n) = 0 \; \forall n,\quad \phi^\dagger(n)\phi(n) = \begin{cases}
    0 \; \text{even } n\\
    1 \; \text{odd } n.
    \end{cases}
\end{equation}
As one might expect for a non-interacting theory, this vacuum state corresponds to a complete absence of particles and anti-particles as well as zero gauge field.

\begin{remark}\label{rem: m log N}
    Assuming open boundaries $L(0) = L(N) = 0$, since the occupation of neighbouring gauge fields can vary by at most $\pm1$ due to the occupancy of a single fermionic site, the maximum value of $|L(n)|$ is $N/4$ for states obeying Gauss's law.
\end{remark}

This is particularly useful when studying such models with quantum computers or tensor network methods in which it is necessary to truncate the dimensionality of the gauge fields to map them onto finite dimensional systems. When using qubit or spin system with a binary encoding, the number of qubits needed to encode the gauge field subspace is $m=\lceil \log_2(N/2 + 1)\rceil$, which gives maximum and minimum gauge field occupancies of $\pm\Lambda$ with $\Lambda\geq N/4$. We will use such a value for $m$ throughout this work. By using an initial state with zero field, and generating the Krylov basis with repeated action of the Hamiltonian which conserves Gauss's law, we can be sure that no states generated require field values beyond the corresponding maximum and minimum field values $\pm N/4$ and there is no loss of accuracy due to using a truncated finite dimensional gauge field space. It is likely that further truncation can be applied with minimal impact on accuracy of solution. It has been shown~\cite{tong2022provably} that for time evolution and computation of  spectrally isolated eigenstates it suffices to approximate the gauge field with a truncated Fock space of dimension $\Lambda$ that scales polylogarithmically with inverse error. The ability to significantly truncate the dimensionality of the gauge fields has been noted in numerical tests~\cite{kuhn2014quantum, buyens2017finite, zache2022toward}, and is standard practice in tensor network methods for LGTs.

The fermionic operators within the Hamiltonian can be mapped to spin operators using the Jordan-Wigner transformation
\begin{equation}
\begin{split}
    \phi(n) &= \prod_{j<n}\left[i\sigma_3(l)\right]\sigma^-(n),\\ \phi^\dagger(n) &= \prod_{j<n}\left[-i\sigma_3(l)\right]\sigma^+(n),
\end{split}
\end{equation}
where $\sigma_1,\sigma_2,\sigma_3$ are the spin-half Pauli operators and $\sigma^{\pm} = (\sigma_1\pm i \sigma_2)/2$. Under this transformation the Hamiltonian becomes
\begin{equation}\label{eqn: schwinger ham pauli}
\begin{split}
H_0 & = \sum_{n=1}^{N-1}L^2(n) + \sum_{n=1}^N(-1)^n\left(\frac{\mu}{2} + \frac{\mu}{2}\sigma_3(n)\right),\\
V & = \sum_{n=1}^{N-1}\left(\sigma^+(n)e^{i\theta(n)}\sigma^-(n+1) + \text{h.c}\right)
\end{split}
\end{equation}
and Gauss' law can be reformulated as
\begin{equation} \label{eqn: schwinger interaction}
\begin{split}
    L(n)-L(n-1) &=  \frac{1}{2}\left(\sigma_3(n)+(-1)^n\right)\\
    &\qquad \forall \; 0 < n \leq N.
\end{split}
\end{equation}

If we assume free boundaries $L(0)=L(N)=0$ we can eliminate the gauge field by using Gauss' law combined with the gauge transformation
\begin{equation}
    \sigma^-(n)\rightarrow\prod_{j<n}\left[e^{-i\theta(j)}\right]\sigma^-(n),
\end{equation}
in which case the Hamiltonian becomes
\begin{equation}\label{eqn: gauged transformed schwinger}
\begin{split}
    H_0=&\sum_{n=1}^N(-1)^n\left(\frac{\mu}{2} + \frac{\mu}{2}\sigma_3(n)\right) \\
    &+ \sum_{n=1}^{N-1}\left[\frac{1}{2}\sum_{m=1}^n\left(\sigma_3(m)+(-1)^m\right)\right]^2,\\
    V =& \sum_{n=1}^{N-1}\left(\sigma^+(n)\sigma^-(n+1) + \text{h.c}\right).
\end{split}
\end{equation}
Within this spin-only representation, the $x\rightarrow0$ ground state corresponds to anti-ferromagnetic state with $\sigma_3(n)=-(-1)^n\; \forall \; n$~\cite{nguyen2022digital}.

Fig.~\ref{fig: model schematic} shows a schematic diagram of a state written in the original staggered fermion representation (with Gauss's law) along with the corresponding state after undergoing Jordan Wigner and gauge transformations to remove the gauge field. Similar transformations can be applied for arbitrary open boundary conditions choosing $L(0)$ and $L(N)$ to be any fixed value, this corresponds to a topological $\theta$-term~\cite{buyens2017finite, funcke2020topological}.

\begin{figure*}
    \includegraphics[width=0.49\textwidth]{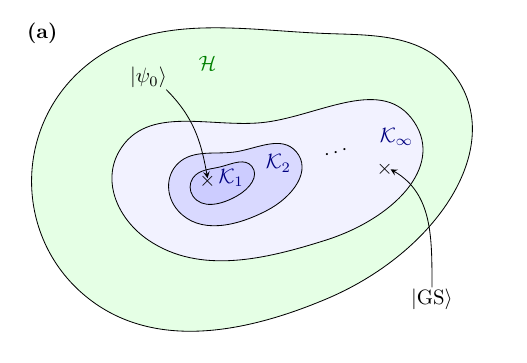}
    \includegraphics[width=0.49\textwidth]{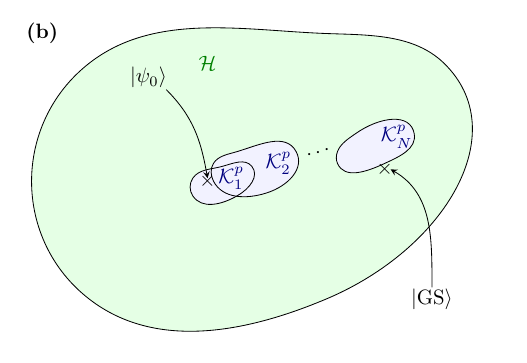}
    \caption{Schematic of the quantum subspace expansion algorithms considered in this work. \textbf{(a)} QSE algorithm using a Krylov basis. Starting from initial state $\ket{\psi_0}$ within the total Hilbert space $\mathcal{H}$, Krylov bases $\mathcal{K}_D$ with increasing basis sizes $D$ allow for access to more of the Hilbert space. For $D=\infty$, the corresponding Krylov space $\mathcal{K}_\infty$ contains infinite imaginary time evolution and therefore the true ground state $\ket{\textrm{GS}}$. \textbf{(b)} PQSE algorithm using a Krylov basis. In contrast to QSE, PQSE builds up a series of Krylov bases $\mathcal{K}^p_1, \mathcal{K}^p_2\dots\mathcal{K}^p_n$ in a non-Markovian way. A selection criteria based on the variance of corresponding state (see algorithm description in Ref.~\cite{oleary2024partitioned}) at each step to create a Krylov basis $\mathcal{K}^p_i$ with automatically chosen dimension using a reference state corresponding to the optimal state in the preceding Krylov subspace $\mathcal{K}_{i-1}^p$. In practice, this improves the numerical stability of the generalised eigenvalue problem being solved and aims to find a reasonable approximation to the ground state.}\label{fig: qse schematic}
\end{figure*}

\begin{figure*}
    \centering
    \includegraphics[width=\textwidth]{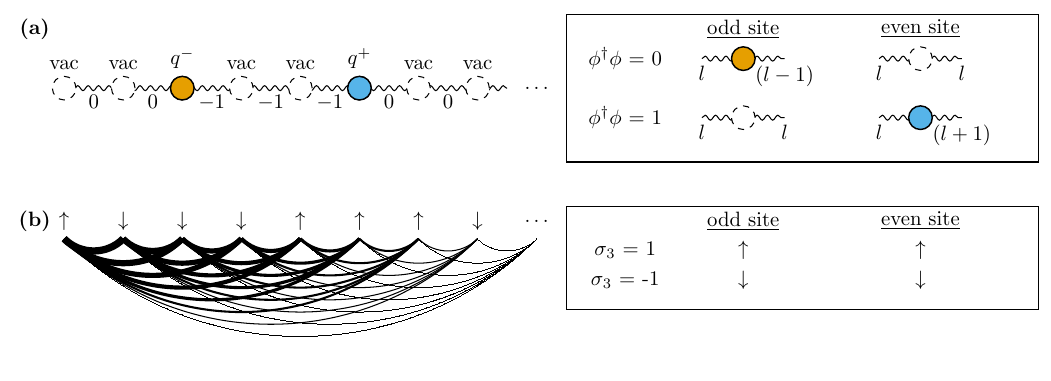}
    \caption{\textbf{(a)} State satisfying Gauss's law for Lattice Schwinger model in staggered fermion formalism. Fermionic lattice sites are indicated by (occupied or unoccupied circles) with connecting lines corresponding to gauge field. Integers below gauge field lines correspond to value of gauge field $L$. The box to right indicates how the value of the fermionic number operator is interpreted as existence of quark or anti-fermion, along with the gauge field flux arising from charge balance imposed by Gauss's law. \textbf{(b)} The same state after applying the Jordan-Wigner transformation and gauge-transformation to arrive at a spin model. Curves connecting spins indicate non-local interactions arising out of gauge transformation, with line thickness of each one corresponding to strength of corresponding term in Hamiltonian. The box to right indicates how existence of a quark or anti-quark corresponds to up or down spin states depending on the parity of the lattice site.}
    \label{fig: model schematic}
\end{figure*}

\section{Quantum subspace expansion via QSVT}\label{sec: qsvt for qse}

\subsection{Block encoding and QSVT}\label{subsec: block encoding}

Qubitization and quantum singular value transformations (QSVT) have emerged as powerful methods for implementing arbitrary functions of non-unitary operators on quantum computers. It has been shown how many paradigmatic quantum algorithms can be implemented using such methods which in some cases offers significant improvements over resources requirements and scalings of the original proposals~\cite{martyn2021grand}. QSVT has become one of the most popular methods for implementing quantum phase estimation and time evolution algorithms. Although such methods are expected to require some level of fault tolerance, the latter application has already been demonstrated on a trapped-ion hardware platform~\cite{kikuchi2023realization}.

In this work, we use QSVT via qubitization to generate the polynomial expectation values required to generate the Krylov basis within QSE. In the following, we describe the key idea behind qubitization and block encoding and then quantum singular value transforms.

\begin{definition}[Block encoding, modified from \cite{martyn2021grand}]\label{def: block encoding}
    Let $H$ be a Hamiltonian acting on Hilbert space $\hil_s$ with spectral norm $||H||\leq1$. A block-encoding of $H$ is three unitary operators $(U,G,\tilde{G})$, where $U$ acts on $\hil_a\otimes\hil_s$, for some auxiliary Hilbert space $\hil_a$, and $\ket{G} := G\ket{0}\in \hil_a$ and $\ket{\tilde{G}} := \tilde{G}\ket{0}\in \hil_a$, are states such that
    \begin{equation}
        \left(\bra{\tilde{G}}\otimes \I_s\right)U\left(\ket{G}\otimes\I_s\right) = H.
    \end{equation}
\end{definition}

From this definition, we see that block encoding allows us to encode the operation of $H$ within the subspace labelled by auxiliary states $\ket{G}$, $\ket{\tilde{G}}~$\footnote{In the case where $G=\tilde{G}$ it is this state that gives rise to the name qubitization. $\hil_a$ acts as a qubit subspace spanned by $\ket{G}$ and its orthogonal complement. Functions of the block-encoded operator are implemented by repeated rotations around the Block sphere of this effective qubit.}. It is often the case, such as in the original proposal~\cite{low2019hamiltonian}, that the auxiliary states are taken to be equal, i.e., $\ket{G}=\ket{\tilde{G}}$. This is also  the case in previous works using qubitization for QSE. We will see later that we can reduce the cost of QSE via QSVT by using an asymmetric encoding $\ket{G}\neq\ket{\tilde{G}}$.

Provided we can find unitary operators $U$, $G$ and $\tilde{G}$ (if not equal to $G$), qubitization allows us to measure operators and, as we will see later, functions of non-unitary operators using quantum circuits.

The linear combination of unitaries (LCU) method~\cite{low2019hamiltonian} is a commonly used method for block encoding local Hamiltonians that we will use in this work. LCU works as follows: We express an $n$ qubit Hamiltonian as a linear combination of $T$ Pauli operators:
\begin{equation}\label{eqn: gen hamiltonian}
    H = \sum_{i=0}^{T-1} \alpha_i\hat{P}_i,
\end{equation}
where $\alpha_i$ are real coefficients and $\hat P_i \in \{\I,X,Y,Z\}^{\otimes n}$ is an $n$ Pauli operator. In most situations we can assume that $\sum_{i=0}^{T-1} |\alpha_i| = 1$, such that $||H|| < 1$, by renormalising the Hamiltonian. The block encoding unitary is 
\begin{equation}\label{eqn: U}
    U=\sum_{i=1}^{T-1}\ketbra{i}{i}\otimes\hat P_i,
\end{equation}
which corresponds to the application of each $\hat P_i$ controlled on the state of the auxiliary register being in $\ket{i}$. 

As for the corresponding block encoding states, we divert from previous implementations of QSE with qubitization~\cite{kirby2023exact} by using an asymmetric block encoding $G\neq \tilde{G}$. For our work we will use
\begin{equation}
    \ket{G}=\sum_{i=0}^{T-1}\sqrt{|\alpha_i|}\ket{i}
\end{equation}
and 
\begin{equation}
    \ket{\tilde{G}}=\sum_{i=0}^{T-1}\sgn(\alpha_i)\sqrt{|\alpha_i|}\ket{i}.
\end{equation}
It is straightforward to confirm that these definitions of $U$, $\ket{G}$ and $\ket{\tilde{G}}$ satisfy Definition~\ref{def: block encoding} for a block encoding. We will see later that since $G$ and $\tilde G$ differ only by the possible signs of $\alpha_i$ that they can be prepared using almost identical circuits.

The challenge in performing a block encoding in this way is to choose a labelling procedure $\{\ket{i}\}$ that makes efficient use of auxiliary qubits leading to an efficient method of generating the block encoded state. We use a labelling procedure in which $\ket{i}$ are computational basis states corresponding to a pair of $n$ length binary numbers $(\vec{x},\vec{z})_i$ that give a binary encoding of $n$ qubit Pauli operator $\hat{P}_i$. If $\vec{x}$ ($\vec{z}$) contains a one in bit $j$, then $\hat{P}_i$ includes an $X$ ($Z$) operator on qubit $j$. Constituent $Y$ operators arise when both vectors contain a one in the same position as $ZX = iY$. The operation of $U$ in~\eqref{eqn: U} can then be implemented with the following three sets of gates: (1) A set of CX gates, each controlled on a qubit within the $\vec{x}$ register and targeting the corresponding qubit within the system register. (2) A set of CZ gates doing the same thing, except each controlled on one of the qubits within the $\vec{z}$ register. (3) A controlled $S$ ($i$ phase gate) acting between a qubit in the $\vec{x}$ register and the qubit in the same position within the $\vec{z}$ register.

To prepare $\ket{G}$ we must prepare a real superposition with amplitudes $\{\sqrt{|\alpha_i|}\}$ over computational basis states $\ket{(\vec{x},\vec{z})_i}$. To do so, a series of multi-controlled partial CNOT and SWAP gates must be applied with control and target qubits as well as rotation angles chosen in such a way as to rotate into different states with increasing Hamming weights with the correct amplitudes. The process and order by which this is done is quite complex, a full description is given in Appendix~\ref{app: preparing G}. As described by Kirby~\cite{kirby2023exact} this method requires us to apply a different multi-controlled rotation gate for each Pauli operator within the Hamiltonian. We present in Sec.~\ref{sec: total resources}, a method for reducing the number of these gates required by exploiting translational symmetry within the Hamiltonian such that the most costly of these gates need only be applied to single gauge link within the lattice and then effectively copied over to the remaining lattice sites. This reduces the cost from $\tilde{\mathcal{O}}(N^2)$ (if applying the method given by Kirby~\cite{kirby2023exact} to our Hamiltonian) to $\tilde{\mathcal{O}}(N)$. The other block encoding state $\ket{\tilde{G}}$ can be generated using the same set of rotation gates except flipping the phase of the some of the rotation angles in order to include an additional minus sign phase where needed.

In previous work~\cite{kirby2023exact} in which $\tilde{G}=G$, to include the negative signs on values of $\alpha_i$ a potentially large number of multicontrolled Z gates must be applied within $U$. In using our asymmetric encoding to include the required signs within $\ket{\tilde{G}}$ we remove this cost which, depending on the number of negative $\alpha_i$, could otherwise dominate the cost of the algorithm.

Given access to the block encoding, functions of $H$ can be implemented through the quantum singular value transforms (QSVT)~\cite{gilyen2019quantum}. For our purpose of computing a Krylov basis, the following theorem to generate polynomials of $H$ will suffice.

\begin{theorem}[Polynomials through QSVT, modified from~\cite{gilyen2019quantum, martyn2021grand}]\label{thrm: qsvt}
    Given Hamiltonian $H$ with eigenvalues on the range $[-1,1]$ and block encoding $(U,G,\tilde{G})$, we define the projectors
    \begin{equation}
        \Pi := \ketbra{G}{G},\quad \tilde{\Pi} :=\ketbra{\tilde{G}}{\tilde{G}}
    \end{equation}
    onto the subspaces spanned by $\ket{G}$ and $\ket{\tilde{G}}$ and corresponding rotations by angle $\varphi\in\mathbb{R}$, 
    \begin{equation}
        \Pi_\varphi:=e^{2i\varphi\Pi}, \quad \tilde{\Pi}_\varphi:=e^{2i\varphi\tilde{\Pi}}.
    \end{equation}
    Then, for any polynomial $f:[-1,1]\rightarrow[-1,1]$ of degree $k$, 
    \begin{itemize}
    \item if $k$ is odd, there exists angles $\vec{\varphi}\in \mathbb{R}^k$ defining the operator
        \begin{equation}
        U_{\vec{\varphi}}:=\tilde{\Pi}_{\varphi_1}U\prod_{j=1}^{\frac{k-1}{2}}\left[\Pi_{\varphi_{2j}}U^\dagger\tilde{\Pi}_{\varphi_{2j+1}}U\right]
    \end{equation}
    such that 
    \begin{equation}
        \bra{\tilde{G}}U_{\vec{\varphi}}\ket{G} = f(H),
    \end{equation}
    \item if $k$ is even, there exists angles $\vec{\varphi}\in \mathbb{R}^k$ defining the operator
        \begin{equation}
        U_{\vec{\varphi}}:=\prod_{j=1}^{\frac{k}{2}}\left[\Pi_{\varphi_{2j-1}}U^\dagger\tilde{\Pi}_{\varphi_{2j}}U\right],
    \end{equation}
    such that 
    \begin{equation}
        \bra{{G}}U_{\vec{\varphi}}\ket{G} = f(H).
    \end{equation}
    \end{itemize}
    Moreover, there exists an efficient classical algorithm for computing $\vec{\varphi}$.
\end{theorem}

The required quantum circuits therefore have the form
\begin{widetext}
\begin{equation}
\Qcircuit @C=1em @R=.7em {
&&&& \mbox{\hspace{4.5em} Repeat for $j=\frac{k}{2},\frac{k}{2}-1,\dots,1$} &\\
&&&&&\\
\lstick{\ket{\psi_0}\in \mathbf{C}^{N_q}} & \qw {/} & \qw & \multigate{2}{U} & \qw & \multigate{2}{U^\dagger} & \qw & \qw & {/} \qw & \meter \\
\lstick{\ket{0}^{\otimes N_\text{q}}} & \qw {/} &  \multigate{1}{G} & \ghost{U} & \multigate{1}{\tilde{\Pi}_{2j}} & \ghost{U^\dagger} & \multigate{1}{\Pi_{2j-1}} & \multigate{1}{G^\dagger} & {/} \qw & \meter\\
\lstick{\ket{0}^{\otimes N_\text{q}}} & \qw {/} & \ghost{G} & \ghost{U} & \ghost{\tilde{\Pi}_{2j}} & \ghost{U^\dagger} & \ghost{\Pi_{2j-1}} & \ghost{G^\dagger} & {/} \qw & \meter
\gategroup{3}{4}{5}{7}{.7em}{--}
}
\end{equation}
for $k$ even and
\begin{equation}
\Qcircuit @C=1em @R=.7em {
&&&& \mbox{\hspace{4.5em} Repeat for $j=\frac{k-1}{2},\frac{k-1}{2}-1,\dots,1$} &\\
&&&&&\\
\lstick{\ket{\psi_0}\in \mathbf{C}^{N_q}} & \qw {/} & \qw & \multigate{2}{U} & \qw & \multigate{2}{U^\dagger} & \qw & \multigate{2}{U} & \qw & \qw & {/} \qw & \meter \\
\lstick{\ket{0}^{\otimes N_\text{q}}} & \qw {/} &  \multigate{1}{G} & \ghost{U} & \multigate{1}{\tilde{\Pi}_{2j+1}} & \ghost{U^\dagger} & \multigate{1}{\Pi_{2j}} & \ghost{U} & \multigate{1}{\tilde{\Pi}_{\varphi_1}} & \multigate{1}{\tilde{G}^\dagger} & {/} \qw & \meter\\
\lstick{\ket{0}^{\otimes N_\text{q}}} & \qw {/} & \ghost{G} & \ghost{U} & \ghost{\tilde{\Pi}_{2j+1}} & \ghost{U^\dagger} & \ghost{\Pi_{2j}} & \ghost{U} & \ghost{\tilde{\Pi}_{\varphi_1}} & \ghost{\tilde{G}^\dagger} & {/} \qw & \meter
\gategroup{3}{4}{5}{7}{.7em}{--}
}
\end{equation}
for $k$ odd, where $N_q$ are the number of qubits required to represent the system, measurements at the end of the circuit represent measuring all $3N_q$ qubits in the computational basis.
\end{widetext}

\begin{corollary}\label{cor: qsvt for poly}
    There exists $\vec{\varphi}\in \mathbb{R}^k$ such that if $k$ is odd
    \begin{equation}
        \left(\bra{\tilde{G}}\otimes\bra{\psi_0}\right)U_{\vec{\varphi}}\left(\ket{G}\otimes\ket{\psi_0}\right) = \bra{\psi_0}H^k\ket{\psi_0}
    \end{equation}
    and if $k$ is even
        \begin{equation}
        \left(\bra{G}\otimes\bra{\psi_0}\right)U_{\vec{\varphi}}\left(\ket{G}\otimes\ket{\psi_0}\right) = \bra{\psi_0}H^k\ket{\psi_0}.
    \end{equation}
\end{corollary}

We can use this to measure the expectation values required for $\mathbf{H}$ and $\mathbf{S}$ within a Krylov basis. The rotation gates required can be implemented as follows.
\begin{remark}\label{rem: pi g}
The rotation about $\ket{G}$ required for the QSVT procedure in Theorem~\ref{thrm: qsvt} by angle $\varphi_j$ can be implemented with the following circuit
\begin{equation}\begin{split}
\Qcircuit @C=0.5em @R=0.em {
& \push{\rule{0em}{3.5em}\Pi_{\varphi_j}=\rule{1.5em}{0em}}&
\lstick{\ket{0}}& \qw & \qw & \targ & \gate{e^{2i{\varphi_j} Z}} & \targ& \qw & \qw & \qw&\rstick{\ket{0}}\\
 &&&\qw /& \gate{G} & \ctrlo{-1} & \qw & \ctrlo{-1}& \gate{G^\dagger} & /\qw &\qw&\rstick{}\\
}
\qquad\\
\\
\end{split}
\end{equation}
A similar construction can be used for $\tilde{\Pi}_{\varphi_j}$ if $\tilde{G}\neq G$.
\end{remark}

\subsection{Generating the Krylov basis using QSVT}\label{subsec: krylov basis procedure}
Recalling the definition of $\mathbf{H}$ and $\mathbf{S}$ for a Krylov basis given in~\eqref{eqn: qse matrices}, the matrix elements for a $D$ dimensional Krylov space $\bra{\psi_0} H^k \ket{\psi_0}$ for $k=0,1,\dots(2D-1)$ must be estimated on a quantum computer. Given access to $G$, $\tilde{G}$ and $U$ this can be done according to Corollary~\ref{cor: qsvt for poly} as follows.
\begin{enumerate}
    \item Generate a reference state $\ket{\psi_0}$ in the system register and apply $G$ to the auxiliary register to prepare $\ket{G}\otimes\ket{\psi_0}$.
    \item For each $\varphi_j$ in $\vec{\varphi}$, apply the block encoding operator $U$ to both registers as well as the rotation operator $\Pi_{\varphi_j}$ or $\tilde{\Pi}_{\varphi_j}$ (depending on if $j$ is odd or even) to the auxiliary register.
    \item Apply the reverse operation for initial state preparation to the system registers, as well as $G^\dagger$ or $\tilde{G}^\dagger$ depending on the parity of $k$ to the auxiliary register.
    \item Measure all system and auxiliary qubits in the computational basis and count the result if the all zero state is measured.
    \item Return to step 1. and repeat until enough measurements are performed to give desired precision. The expectation value of $\bra{\psi_0} H^k \ket{\psi_0}$ will be the fraction of times that the zero state is measured.
\end{enumerate}

We obtain accurate gate and qubit counts for the qubitization and QSVT procedures by applying the LCU procedure along with the method for generating $\ket{G}$ described by Kirby et al. (Appendix C in Ref.~\cite{kirby2023exact}) to the single flavour Schwinger model to give accurate gate and qubit counts for $U$ and $\ket{G}$ required for the qubitization and QSVT procedure. We exploit translational symmetry within our Hamiltonian to improve the gate scaling of $G$ by a factor of system size $N$. By combining this method with the use of an assymetrical block encoding such that $\tilde{G}\neq{G}$, we are able to reduce the scaling of $U$ by a factor of $N$. The cost of $\tilde{G}$ is the same as $G$. Since the procedures for generating $(U,G,\tilde{G})$ along with analysis resource requirements are quite involved, we defer detailed descriptions and calculation to Appendix~\ref{app: resource cost schwinger}. 

In Appendix~\ref{app: resource cost schwinger} we show that the total number of qubits required to represent the system and implement the block encoding procedure is $\mathcal{O}(\log(N)N)$. Furthermore, the number of gate operations required scale as $\mathcal{O}(N)$ gates for $U$ and $\mathcal{O}(\log(N)N)$ for $G$ and $\tilde{G}$. Both of these gate costs are reduced from $\mathcal{O}(\log(N)N^2)$ by using our method to exploit translational symmetry. Exact gate costs and prefactors are given in Theorems~\ref{thm: cost U} and \ref{thm: cost G} in Appendix~\ref{app: resource cost schwinger}. The $\mathcal{O}(\log(N)N^2)$ gate cost resulting from the LCU procedure described by Kirby at al.~\cite{kirby2023exact}, including prefactors, is shown in Theorem~\ref{thm: cost U with phases} of Appendix~\ref{subsec: thm: cost of U with phases}. By exploiting translational symmetry, we reduce the gate cost of the algorithm by multiple orders of magnitude for relevant system sizes, as compared to the qubitization based on the work of Kirby~\cite{kirby2023exact}. We note that the work by Kirby et al. considers two methods of implementing the LCU procedure, we compare to the second LCU method therein which requires fewer gates at the expense of more circuit operations than the first method they describe. 
Further discussion of the gate costs of the QSVT procedure, as well as the total algorithmic costs, is given in Sec.~\ref{sec: total resources}.

\section{Numerical experiments}\label{sec: numerical experiments}

To analyse the performance of the QSE algorithm for the Schwinger Hamiltonian, we perform classical simulations based on sparse matrix multiplication. We consider fault tolerant operation and therefore neglect the effect of device noise. To begin with, we will also neglect the effect of finite measurement noise on the expectation values of $\langle H\rangle$ before introducing it later on. For the quantum computing algorithm, we propose working with the single flavour 1+1D Schwinger Hamiltonian of~\eqref{eqn: schwinger ham pauli} in which we explicitly include the gauge fields. This is in contrast to using the gauge-transformed model of~\eqref{eqn: gauged transformed schwinger} in which the transformed interactions contains long range terms. Although the gauge transformed representation will require fewer qubits as avoid explicitly storing the value of gauge fields, we expect the gate cost to scale more poorly. Furthermore, for more general models it is necessary to include the gauge fields as it is not possible to remove gauge fields through such a transformation for models in two or more spatial dimensions.

When performing sparse matrix multiplication with the Hamiltonian, the cost is not limited by the locality of interactions, but rather the memory usage required to store states. Thus, to calculate expectation values in $\mathbf{H}, \mathbf{S}$ for our numerical experiments, we work with the gauge transformed Hamiltonian of~\eqref{eqn: gauged transformed schwinger}. This allows us to use dense vectors of length $2^N$ to represent the corresponding spin states, rather than length $2^{(N-1)m + N}$ which would be required if we included the Hilbert space of the gauge fields. Using this representation, we are able to generate Krylov bases for up to $N=30$ lattice sites through sparse matrix multiplication and perform direct diagonalisation of the ground state for up to $N=26$ using double precision with up to 3TB of RAM. In all of our experiments we fix $\mu=1.5$ and $x=0.5$. 

As reference state $\ket{\psi_0}$ we use the ground state of the corresponding $x=0$ system given by~\eqref{eqn: uncoupled ground state}. Since this state satisfies Gauss's law, and the Hamiltonian used to generate the Krylov basis commutes with each side of Gauss's law in~\eqref{eqn: gauss}, the solutions generated by the QSE algorithm will lie in the correct symmetry sector of Gauss's law. This input state is simply a computational basis state generated by applying X operators to the odd lattice sites. 

We begin by simulating the algorithm with no measurement noise, running the QSE algorithm for system sizes ranging from $N=4$ to $N=26$. For each system size, we perform QSE with increasing basis size $D$ until the generalised eigenvalue problem can no longer be solved due to ill-conditioning. Although we do not introduce measurement noise at this point, finite machine precision still means that the generalised eigenvalue problem becomes unstable for sufficiently large Krylov basis. As a metric for the success of the algorithm we consider the \textit{fractional energy error} between the error in the ground state energy produced by the QSE algorithm and the interaction energy of the system (the difference between the $x=0$ non-interacting ground state energy and the $x\neq0$ ground state energy). This is denoted by $\Delta E/E_\textrm{int}$. As is to be expected with a Krylov basis, we find that for each value of $N$, the fractional energy error decreases exponentially with the basis size. For each value of $N$, by performing an exponential fit between $\Delta E/E_\textrm{int}$ compared to Krylov basis size we are able to estimate the basis size required to achieve a chosen energy error (assuming the the algorithm does not fail due to finite precision). For further details of the experiments and fitting, see Appendix~\ref{app: noiseless fitting}.

\begin{figure}
    \includegraphics[width=\linewidth]{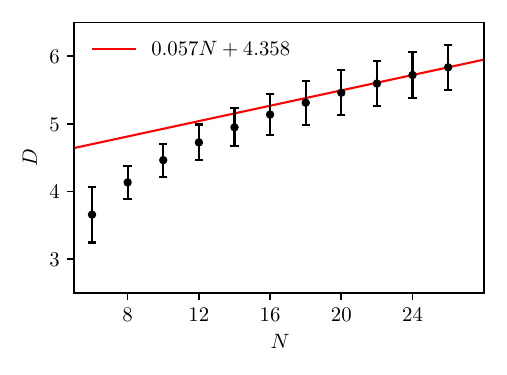}
    \caption[Estimated Krylov order $D$ needed for QSE to achieve fractional energy error of $\Delta E/E_\text{int} = 10^{-4}$ in the case of no measurement noise]{Estimated Krylov order $D$ needed for QSE to achieve fractional energy error of $\Delta E/E_\text{int} = 10^{-4}$ in the case of no measurement noise calculated using data from Appendix~\ref{app: noiseless fitting} for a system with $\mu=1.5$, $x=0.5$. Error bars correspond to standard errors on the fits to generate these points (see Appendix~\ref{app: noiseless fitting} for details). Red line corresponds to linear fit to final two data points.}
    \label{fig: noiseless results vs calls extrapolation}
\end{figure}

The result of this fitting is given in Fig.~\ref{fig: noiseless results vs calls extrapolation} in which we present an estimate for the basis size required to achieve a desired energy error of $\Delta E/E_\textrm{int}=10^{-4}$. Considering the gradient of the curve drawn out by the data points, the rate at which the required Krylov order $D$ increases with $N$, decreases as we go from $N=4$ to $N=26$, with the value of $D$ becoming approximately linear with $N$ for the later data points. Assuming the required order continues to grow linearly, or the gradient with respect to $N$ further decreases, by fitting a straight line to the final two data points and extrapolating this to larger $N$, we can estimate an upper bound on the required basis for a given $N$. Such a line is shown in Fig.~\ref{fig: noiseless results vs calls extrapolation} and corresponds to an upper bound of $D\approx0.057N + 4.358$. This linear relationship, characterised by a shallow slope, suggests that the basis size cost does not exhibit a substantial increase as the system size grows assuming zero measurement error and infinite precision. However, to gain a full understanding of the resource cost of the algorithm we must analyse the problem in the presence of finite measurement noise. As we will see from the following sections, the behaviour and cost of the QSE algorithm in such a scenario differs significantly from the zero-noise case.

\subsection{Modelling measurement noise}\label{subsec: measurement noise}

We now consider the effect of shot noise on the expectations values of $\mathbf{S}$ and $\mathbf{H}$ resulting from finite number of measurement. The effect of this is to perturb these matrices as $\mathbf{S}\rightarrow \mathbf{S+\Delta_S}$ and $\mathbf{H}\rightarrow \mathbf{H+\Delta_H}$. Like the original overlap matrices, the perturbation matrices $\mathbf{\Delta_S}, \mathbf{\Delta_H}$ are Hankel matrices\footnote{in which each ascending skew-diagonal from left to right is constant}. Since expectation values are measured as Bernoulli trials, elements of the perturbation matrices can be modelled as random variables drawn from a normal distribution with width given by the expected variance of the expectation value. 

Let $\overline{\bra{\psi_0}H^k\ket{\psi_0}}$ denote the best estimate for the expectation value of $\bra{\psi_0}H^k\ket{\psi_0}$ as measured according to the procedure in Sec.~\ref{subsec: krylov basis procedure}. If $M_k$ measurements are used to calculate this estimate then the variance is given by
\begin{equation}\label{eqn: var matrix elements}
\begin{split}
    \var{\overline{\bra{\psi_0}H^k\ket{\psi_0}}} &= \frac{\var \bra{\psi_0}H^k\ket{\psi_0}}{M_k} \\
    &= \frac{\bra{\psi_0}H^{2k}\ket{\psi_0} - \bra{\psi_0}H^k\ket{\psi_0}^2}{M_k}.
\end{split}
\end{equation}
From the above, we are able to calculate the variance of each matrix element used in the generalised eigenvalue problem by calculating values of $\bra{\psi_0}H^k\ket{\psi_0}$ for $k$ twice as large as the maximum value needed for the expectation values. We choose values of $M_k$ as follows.

We use the total calls to the block encoding operator $\Pi_\varphi U$ as a measure of the cost of performing the measurements required (counting $\Pi_\varphi$ and $\tilde{\Pi}_\varphi$ in the same way since they have the same cost). By controlling this, we effectively control the number of gates required. Using the QSVT procedure of Theorem~\ref{thrm: qsvt}, the circuit required to generate $H^k$ requires $k$ applications of $\Pi_\varphi U$. Thus, given a distribution of shots $\{M_k|k=1,\dots,2D+1\}$, the total number of calls to $\Pi_\varphi U$ is
\begin{equation}\label{eqn: tot U calls}
    \# \;\Pi_\varphi U \text{ calls} = \sum_{k=1}^{2D+1}k M_k.
\end{equation}
We choose $\{M_k\}$ such that the measurement error as a fraction of expectation value is the same for all $k$, and as low as possible given a total number of $\Pi_\varphi U$ applications. Thus, using \eqref{eqn: var matrix elements} we assign $M_k$ as
\begin{equation}
    M_k= \mathcal{M}\frac{\var\bra{\psi_0}H^k\ket{\psi_0}}{\bra{\psi_0}H^k\ket{\psi_0}^2},
\end{equation}
where proportionality constant $\mathcal{M}$ is calculated to normalise the total number of block encoding calls according to~\eqref{eqn: tot U calls}.

To simulate the effect of the resulting measurement noise, we can then sample elements of the perturbation matrices as ${\Delta_S}_{i,j} = \delta_{i+j}$, ${\Delta_H}_{i,j} = \delta_{i+j+1}$, where $\delta_k$ are drawn from a normal distribution of width $\sqrt{\var\bra{\psi_0}H^k\ket{\psi_0}/M_k}$.

\subsection{Dealing with measurement noise}\label{subsec: dealing with noise}

For the noise levels we consider in this work, standard QSE often fails for moderately sized Krylov bases due to ill-conditioning. Therefore, for simulations with shot noise we use partitioned quantum subspace expansion (PQSE) which has been found to perform better than normal QSE and thresholded approaches in avoiding numerical instability in the presence of finite noise~\cite{oleary2024partitioned}.

PQSE decomposes the full $D$-dimensional Krylov space into a series of smaller subspaces of dimension $\{d_i\}$ such that $\sum_i(d_i-1)=D-1$.
Each subspace of dimensions $d_i$ is a Krylov subspace in which the input state is the optimal state from the previous subspace of dimension $d_{i-1}$. Compared to the full $D$ dimensional Krylov space, within the smaller subspaces the generalised eigenvalue problem contains smaller matrices which are less likely to be ill-conditioned. The size of each partitioned subspace is selected systematically throughout the PQSE algorithm by choosing $d_i$ such that the solution to the corresponding generalised eigenvalue problem gives the smallest energy variance. Although the partitioned subspaces correspond to different input states, the overlap matrices $\mathbf{H},\mathbf{S}$ required for each partition can be efficiently calculated from the same set of overlaps. As a result, PQSE requires similar overlaps as standard QSE, namely $\bra{\psi_0}H^k\ket{\psi_0}$ for $k=0,1,\dots,2D$, where the single additional overlap (for $k=2D$) is required for variance estimation. For a full description of the PQSE algorithm, we refer the reader to the original publication~\cite{oleary2024partitioned}. 

\subsection{Results}\label{subsec: qse results}

To test the performance of the PQSE algorithm in the presence of measurement noise we simulate noise as described in Sec.~\ref{subsec: measurement noise} for differing numbers of calls to block encoding operators $\Pi_\varphi U$. For a given value of $\Pi_\varphi U$ calls, we split these calls among varying numbers of expectation values to apply the PQSE procedure with different total Krylov basis size $D$.

We considered lattice sizes varying from $N=4$ to $N=26$. For the extrapolation for the total resources we use data resulting from simulations using $10^2, 10^3, \dots,10^{15}$ calls to $\Pi_\varphi U$, for clarity, we show a representative subset of these results in Fig.~\ref{fig: results vs calls}. For each choice of $\Pi_\varphi U$ calls and $N$, 100 noise instances are simulated with the median and upper and lower quartiles plotted. From this we notice three things. Firstly, for a given number of calls to $\Pi_\varphi U$, as we increase the maximum Krylov order the fractional energy error starts off decreasing, until reaching a minimum at which point it begins to increase. We can interpret the initial decreasing behaviour as the error reduces due to the increased basis size - of course the error is worse than the noiseless equivalent due to error on the expectation values. Once the minimum is reached however, the matrices become sufficiently ill-conditioned such that any finite measurement noise is enough for the generalised eigenvalue problem to become unstable and so increasing order doesn't improve solution quality. In fact, when this happens increasing the maximum value of $D$ increases the measurement error on expectation values as we have fixed the number of shots, and so error increases rapidly. Secondly, the value of $D$ giving the lowest error increases with $N$. This is to be expected as noiseless results for this, as well as other Hamiltonians, typically require larger Krylov bases for larger systems to achieve the same error. Finally, in all cases increasing the number of calls to $\Pi_\varphi U$, and therefore reducing measurement error decreases overall solution error.

\begin{figure*}
    \centering
    \includegraphics[width=\linewidth]{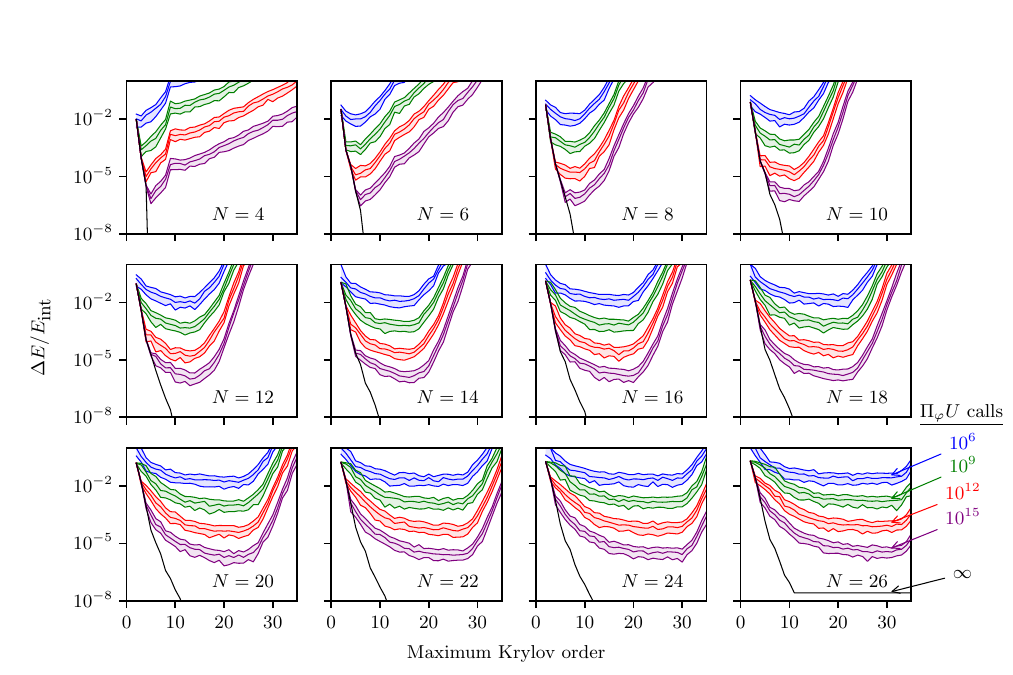}
    \caption[Fractional energy error with PQSE procedure with finite number of measurements for range of different lattice sizes.]{Fractional energy error with PQSE procedure with finite number of measurements for range of different lattice sizes $N$ for $\mu=1.5$, $x=5$. Values are shown as a function of the maximum allowed Krylov basis for the PQSE procedure and as a function of number of calls to the block encoding operator $\Pi_\varphi U$ (which affect the number of measurements as defined in the main text). Lines and shaded area corresponds to the median as well as upper and lower quartiles of energy error. The black line (labelled as infinite calls to $\Pi_\varphi U$) corresponds to no measurement noise.}
    \label{fig: results vs calls}
\end{figure*}

As in the noiseless case, we want to predict the number of resources required, this time in terms of total calls to $\Pi_\varphi U$ rather than basis size, to achieve a desired solution accuracy for a given system size. To estimate this, we begin by taking the smallest median energy error for each number of calls to $\Pi_\varphi U$. These correspond to the lowest points on each curve in Fig.~\ref{fig: results vs calls}. This assumes that we have a reasonably good guess for the maximum order required to give the best results through PQSE. In practice, this could be inferred through some method of scaling up system size through real or simulated experiments. We see that for larger systems in Fig.~\ref{fig: results vs calls}, there are a number of choices of Maximum Krylov order that lead to near-optimal output error. 

\begin{figure}
    \includegraphics[width=\linewidth]{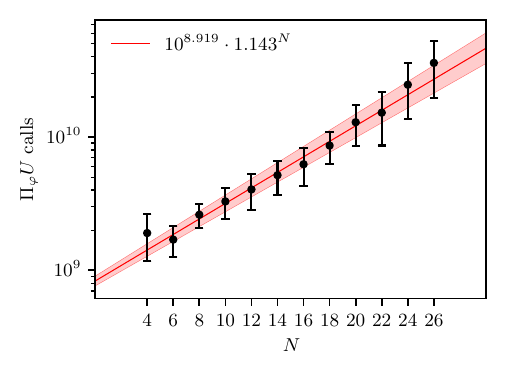}
    \caption[Estimated number of call to qubitization procedure using PQSE to achieve fractional energy error of $10^{-4}$]{Estimated number of call to qubitization procedure using PQSE to achieve fractional energy error of $\Delta E/E_\textrm{int} = 10^{-4}$ for $\mu=1.5$, $x=0.5$, calculated using data from Fig.~\ref{fig: results vs calls}. Error bars correspond to the standard error resulting from fits used to generate the data points. Red line corresponds to exponential fit (straight line on log-linear axes) with filled area indicating uncertainty in fitting and extrapolation.}
    \label{fig: noisy results vs calls extrapolation}
\end{figure}

For each value of $N$, we find that the best median energy error achieved by PQSE decreases with a power law dependence on the number of calls to the block encoding operator $\Pi_\varphi U$. This manifests itself, rather than the log-linear relationship demonstrated in the noiseless case, as a log-log dependence between measurement error and resource cost. See Appendix~\ref{app: noise fitting}. Fitting to this dependence we are able to predict the number of calls to the $\Pi_\varphi U$ to achieve a desired fractional energy error. For example, the number of calls required to achieve a fractional error of $10^{-4}$ is shown in Fig.~\ref{fig: noisy results vs calls extrapolation}. From this we see that the number of calls to $\Pi_\varphi U$ grows approximately exponentially with $N$ as $\# \Pi_\varphi U \;\mathrm{calls} \approx 10^{8.919}\cdot1.143^N$ for the values of $N$ considered.

Applying thresholded QSE as described by Epperly et al.~\cite{epperly2022theory} for the same lattice sizes gives energy errors that are approximately an order of magnitude worse than PQSE with same number of resources. See Appendix~\ref{app: more qse simulations} for further details.

\section{Estimating resource requirements}\label{sec: total resources}

After making an estimate of the basis size; and therefore number of calls to the qubitization procedure required, to understand the total resource cost of solving the Schwinger model, we must now compute the number of gates required to implement each part of the quantum circuit. A full description of the different parts of the circuits, along with the mathematics needed to estimate the gate cost is long-winded. Therefore, we defer a full description of the algorithmic implementation along with proofs for upper bounds of the required gate costs to Appendix~\ref{app: resource cost schwinger}. We point the reader in particular to Appendix~\ref{subsec: sketch proof costs} which outlines a sketch of the required proofs. Final results of the costing are presented in Theorems~\ref{thm: cost G}, \ref{thm: cost U}, \ref{thm: cost R} and \ref{thm: total scaling}.

In Fig.~\ref{fig: k=1 gates}~(a) we plot the number of gates as a function of system size $N$ for a single step of the block encoding procedure. This is given by the cost of $U$ plus $\Pi_\varphi$ (which is the same as $\tilde{\Pi}_\varphi$). Fig.~\ref{fig: k=1 gates}~(b) gives the qubit cost of the algorithm which, for systems of $N$ in the hundreds or thousands, is not far off the thousands of physical qubits targeted by, for example IBM, in the next few years. Figs.~\ref{fig: k=1 gates}~(c)~and~(d) show a direct comparison between our improved implementations of the block encoding procedure described in the previous section and the LCU method described by Kirby et al.~\cite{kirby2023exact}. We see that for both T-gate count and CNOT count, our optimisations lead to significant reductions in the gate cost.

\begin{figure*}
    \centering
    \includegraphics[width=\linewidth]{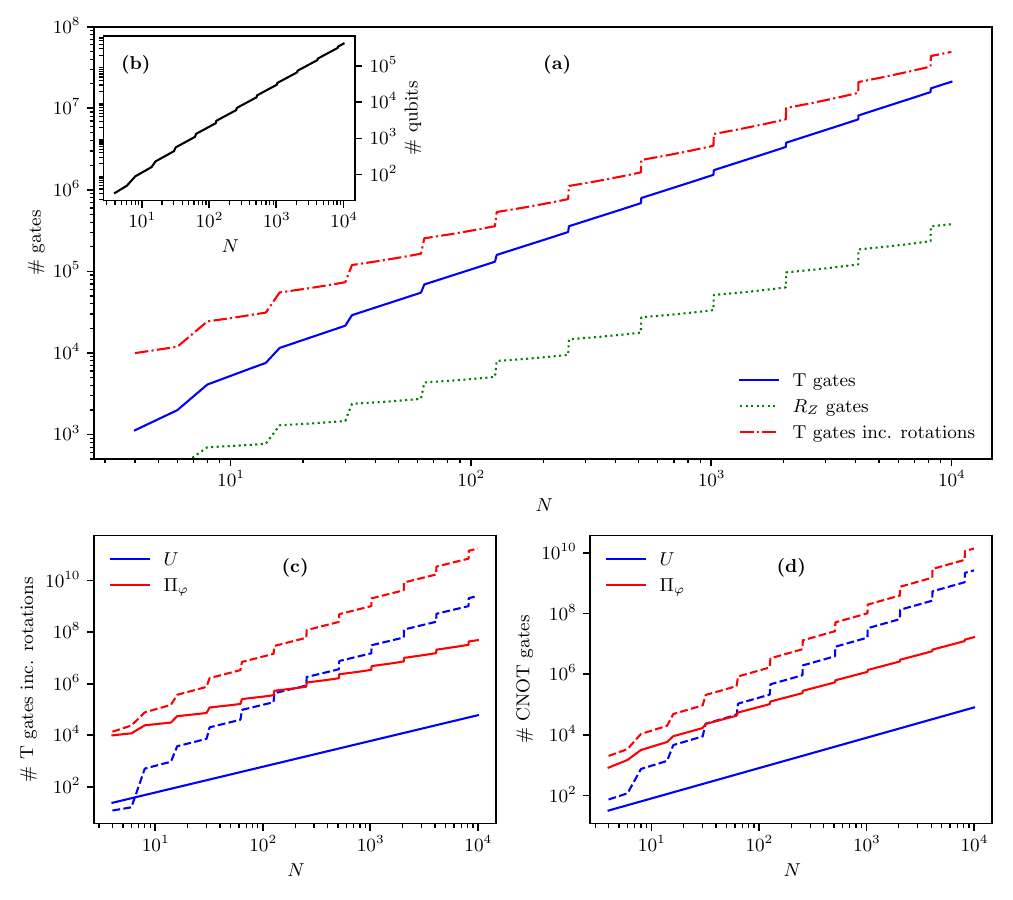}
    \caption[Upper bound gate and qubit costs for block encoding procedure]{Upper bounds for resource costs in a single step of QSVT procedure (single application of $U$ and $\Pi_\varphi$). \textbf{(a)} Total gate costs as a function of lattice size. The line labelled ``T-gates'' corresponds to T-gates arising from controlled operations excluding rotation gates. The line labelled ``T-gates inc. rotations'' includes the additional cost of converting single qubit rotation gates (line labelled ``$R_z$ gates'') using Ref.~\cite{ross2014optimal}. Although not shown, the CNOT gate cost is indistinguishable from the T-gate line on this scale. \textbf{(b)} Total logical qubit count. \textbf{(c)} T-gate cost for individual operators $U$ and $\Pi_\varphi$. Dashed lines give costs for $U$ and $\Pi_\varphi$ using LCU procedure as described by Kirby et al. (see Appendix~\ref{app: preparing G}). Solid lines give cost using our improved methods (see Appendix~\ref{sec: copy trans inv terms} for details). \textbf{(d)} Same as (c) except for CNOT gates.}
    \label{fig: k=1 gates}
\end{figure*}

Throughout this work, we have assumed fault tolerant computation. To quickly assess whether this is necessary, or whether non error-corrected devices may be sufficient, we compare the expected runtime of the algorithmic gates to the coherence times for existing processors. Assuming that two qubit gates are run in series and give the largest contribution to runtime we find that the time needed to execute the CNOTs required for a single step of the QSVT procedure (the application of $\Pi_\varphi U$) is larger than the coherence times for systems of $N=100$ for a range of current generation quantum processors. For systems of $N=1000$ and $N=10000$, the runtime of the gates is many orders of magnitude larger than the coherence times of the devices. See Appendix~\ref{app: hardware specs} for further details on what hardware specifications were used for these comparisons. Even assuming that two-qubit gates can be maximally parallelised (running $N_\text{qubits}/2$ two-qubit gates at the same time), the time to execute ${\Pi}_\varphi U$ can be of the order of $\sim$0.1\% to $\sim$1000\% of the coherence times. Such estimates indicate the circuit lengths required to solve the Schwinger model, at least using the implementation presented here, are beyond what is possible on currently available devices and those that we can expect in the foreseeable future. Therefore, for system sizes of hundreds of lattice sites the algorithm we present must be run in an error-corrected setting.

Having confirmed that fault-tolerance is indeed necessary for systems of hundreds of lattice sites, we now assess the feasibility of our algorithm on an error corrected device. Rather than CNOT gates dominating the execution time of the algorithm, the number of non-Clifford gates play the largest part in determining how long a computation will take to run. We estimate the total number of T-gates required to calculate the ground state energy by combining the circuit costs for a single block encoding step with the total number of calls to the block encoding estimated in extrapolating the numerical results of Section~\ref{sec: numerical experiments} as in Fig.~\ref{fig: noisy results vs calls extrapolation}. This cost is shown for a range of desired errors in Fig.~\ref{fig: total t gates}. We see that, for systems of hundreds of lattice sites, the exponential number of measurements we require in order to control the measurement noise leads to T-gate costs prohibitive within even the most optimistic predictions of fault tolerant computations~\cite{litinski2019magic}. However, given that we have only been able to simulate systems of up to $N=26$ and extrapolated to systems orders of magnitude larger, it is possible that the exponential scaling we predict does not hold for larger systems. Further numerical tests, using tensor network methods to study larger system sizes and using Monte-Carlo methods to calculate the true ground state energy (which does not suffer from the sign problem for the single flavour system we consider) may shed more light on this. Whether the resource scaling would be more or less promising is unknown.

\begin{figure}
    \centering
    \includegraphics[width=\linewidth]{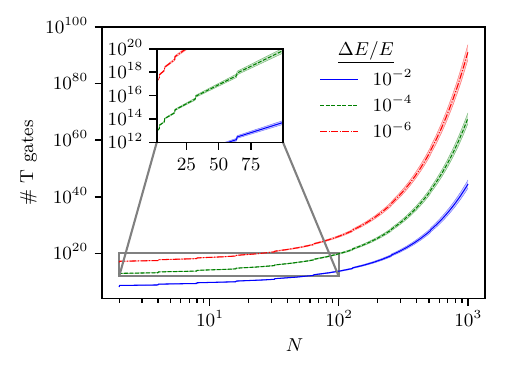}
    \caption{Total number of T-gates required for the whole PQSE procedure for the single flavour Schwinger model with $\mu=1.5$, $x=0.5$ as a function of lattice size $N$ for different fractional energy errors $\Delta E/E_\textrm{int}$.}
    \label{fig: total t gates}
\end{figure}

Furthermore, we repeated our analysis within a regime of strong coupling with $\mu=1.5$ and $x=5$ (as opposed to $x=0.5$ considered up to this point). Details of these results, and in particular the associated plots can be found in App.~\ref{app: strong coupling} On the whole, we find a similar behaviour as in the $x=0.5$ case. Within the noiseless regime for fixed $N$, $\Delta E / E_\text{int}$ decrees exponentially with basis size $D$ and for fixed $\Delta E / E_\text{int}$, the required basis $D$ grows approximately linearly with $N$. The gradient of this linear dependence is larger than in the strong coupling case, reflecting the fact that the overlap between the N\'eel state and the true ground state decreases more quickly with $N$ and therefore a larger Krylov basis is needed to achieve the same error. We find a similar behaviour when considering the effect of measurement noise, giving an estimate of gate costs that grow exponentially with $N$ and with a stronger exponential dependence that the weaker coupling regime.

\section{Discussion and conclusions}\label{sec: qse disc}

In this work we presented a study of the application of the QSE algorithm to a LGT. In particular, we estimated the resources required to compute the ground state of the single flavour lattice Schwinger model for different system sizes and assess how feasible such a computation would be on early error-corrected quantum processors.

By developing a novel LCU procedure for block encoding the Schwinger model Hamiltonian and analysing the gate cost required, in combination with numerical simulation, we successfully give a upper bound estimate for the full gate cost of the algorithm.

For systems of a few hundred lattice sites, we see that the gate depth for a single circuit is beyond what we we expect to be feasible on non-error corrected devices in the near term. However, single circuit T-gate costs on the order of $10^4$ to $10^6$, along with hundreds of logical qubits required make such circuits appealing targets for early fault tolerant devices.

While such gate depths for individual circuits can be seen as a positive step towards useful application of QSE to LGTs, the exponential measurement cost revealed by our numerical simulations indicates that further algorithmic improvements will be required in order to use such a workflow to compute accurate ground states of LGTs - even for reasonably large error corrected devices. While the exponential measurement cost may occur for other systems and QSE methods, by performing a rigorous analysis with respect to system size scaling, we have been able to quantify and draw conclusions about the resource cost for different algorithmic approaches.

Since ground state problems for fermionic systems tend to be QMA-complete~\cite{oliveira2005complexity, liu2007quantum, schuch2009computational, bookatz2012qma}, such an exponential scaling is not particularly surprising\footnote{Although the single flavour Schwinger model can be solved with Monte-Carlo methods that do not suffer from the sign problem. We anticipate quantum algorithms, like the one we present here, being applied to multi-flavour models which are known to suffer from the sign problem and therefore seem unlikely to have efficient computational solutions.}. Nor should it deter research into using QSE algorithms for LGTs but rather motivate further work aimed at reducing the algorithmic cost of such methods. Importantly, our work presents the first systematic study of the scaling of two QSE algorithms, namely thresholded and partitioned, with system size. We show that PQSE can significantly reduce the exponential scaling of the algorithm compared to the current state of the art (TQSE), serving as an important benchmark for future improvements and demonstrating the necessity to further develop noise-robust subspace methods.

A number of works have considered QSE using a basis consisting of real-time evolution~\cite{parrish2019quantum, stair2020multireference, bespalova2021hamiltonian, cortes2022quantum, shen2023real,kirby2024analysis}, imaginary time evolution~\cite{motta2020determining, yeter2020practical}, Chebyshev polynomials~\cite{kirby2023exact} and time evolved states that have been Fourier transformed with respect to a set of filter energies~\cite{cortes2022quantum} as opposed to powers of the Hamiltonian. One particularly promising method for improving the noise resilience of QSE has been the proposal of the Gaussian power basis~\cite{zhang2023measurement}. Generating such basis states may still be infeasible with non-fault-tolerant devices, particularly if one tries to use Trotterisation for large systems in multiple spatial dimensions. In the fault-tolerant regime however, these bases could be implemented on a quantum processor using QSVT techniques and the block encoding we present in this work. While investigating other basis methods is beyond the scope of the work, and would require more sophisticated methods for classical simulation of the overlap matrices, it would be interesting to perform the resource estimation we present here for such methods to see if this leads to more feasible gate costs. This would play an important part in developing a thorough understanding of the resources required to implement such algorithms using block-encoding or Trotterization based approaches as well as related algorithms such as QETU~\cite{kane2024nearly}. This may reveal regimes of both problem size and hardware specification in which the lower overhead of Trotterisation outweighs the favourable asymptotic scaling of block-encoding methods and vice-versa. As one considers scaling up system size, it seems likely that approaches based on real- or imaginary-time evolution via Trottersation give lower gate costs for smaller problem instances within the near-term non-error corrected or even early-fault tolerant era and more complex qubitization based subspace methods preferred for larger systems with the onset of large scale fault tolerant devices.

Finally, as part of our comprehensive resource analysis for ground state computation using QSE, we developed novel improvements to the LCU block encoding procedure with particular application to lattice gauge theory Hamiltonians. By exploiting symmetry within the Hamiltonian, something not present in molecular Hamiltonians for example, we can reduce the cost of the LCU procedure when applied to the Schwinger model from $\tilde{\mathcal{O}}(N^2)$ to $\tilde{\mathcal{O}}(N)$. Our procedure could be used to reduce the resources required for other quantum algorithms to study ground states or dynamics of LGTs such as phase estimation and Hamiltonian simulation, both of which have seen significant cost reductions by using qubitization for quantum chemistry in an error corrected setting~\cite{low2019hamiltonian,berry2019qubitization}.

\section*{Code availabilty}
Code for reproducing the figures presented in this paper can be found at \url{https://github.com/lw-anderson/quantum_krylov_for_lattice_schwinger}.

\begin{acknowledgments}
LWA thanks Nick Bultinck and Mari Carmen Bañuls for helpful discussions. LWA, MK and DJ acknowledge support from the EPSRC National Quantum Technology Hub in Networked Quantum Information Technology (EP/M013243/1) and the EPSRC Hub in Quantum Computing and Simulation (EP/T001062/1). TO acknowledges support by the EPSRC through an EPSRC iCASE studentship award in collaboration with IBM Research. DJ acknowledges support from the Hamburg Quantum Computing Initiative (HQIC) project EFRE. The project is co-financed by ERDF of the European Union and by `Fonds of the Hamburg Ministry of Science, Research, Equalities and Districts (BWFGB)'. DJ acknowledges funding by the Cluster of Excellence `Advanced Imaging of Matter' of the Deutsche Forschungsgemeinschaft (DFG) - EXC 2056 - project ID 390715994. The authors would like to acknowledge the use of the University of Oxford Advanced Research Computing (ARC) facility in carrying out this work~\cite{richards2015university}.
\end{acknowledgments}

\bibliographystyle{unsrtnat}
\bibliography{main.bib} 

\begin{thebibliography}{120}
\providecommand{\natexlab}[1]{#1}
\providecommand{\url}[1]{\texttt{#1}}
\expandafter\ifx\csname urlstyle\endcsname\relax
  \providecommand{\doi}[1]{doi: #1}\else
  \providecommand{\doi}{doi: \begingroup \urlstyle{rm}\Url}\fi

\bibitem[Schwartz(2014)]{schwartz2014quantum}
Matthew~D Schwartz.
\newblock \emph{Quantum field theory and the standard model}.
\newblock Cambridge University Press, 2014.
\newblock \doi{10.1017/9781139540940}.

\bibitem[Bailin and Love(2004)]{bailin2004cosmology}
David Bailin and Alexander Love.
\newblock \emph{Cosmology in gauge field theory and string theory}.
\newblock Taylor \& Francis, 2004.
\newblock \doi{10.1201/9780367806637}.

\bibitem[Fradkin(2013)]{fradkin2013field}
Eduardo Fradkin.
\newblock \emph{Field theories of condensed matter physics}.
\newblock Cambridge University Press, 2013.
\newblock \doi{10.1017/CBO9781139015509}.

\bibitem[Wilson(1974)]{wilson1974confinement}
Kenneth~G Wilson.
\newblock Confinement of quarks.
\newblock \emph{Physical review D}, 10\penalty0 (8):\penalty0 2445, 1974.
\newblock \doi{10.1103/PhysRevD.10.2445}.

\bibitem[Creutz et~al.(1983)Creutz, Jacobs, and Rebbi]{creutz1983monte}
Michael Creutz, Laurence Jacobs, and Claudio Rebbi.
\newblock Monte carlo computations in lattice gauge theories.
\newblock \emph{Physics Reports}, 95\penalty0 (4):\penalty0 201--282, 1983.
\newblock \doi{10.1016/0370-1573(83)90016-9}.

\bibitem[Troyer and Wiese(2005)]{troyer2005computational}
Matthias Troyer and Uwe-Jens Wiese.
\newblock Computational complexity and fundamental limitations to fermionic quantum monte carlo simulations.
\newblock \emph{Physical review letters}, 94\penalty0 (17):\penalty0 170201, 2005.
\newblock \doi{10.1103/PhysRevLett.94.170201}.

\bibitem[Sandvik and Kurkij{\"a}rvi(1991)]{sandvik1991quantum}
Anders~W Sandvik and Juhani Kurkij{\"a}rvi.
\newblock Quantum monte carlo simulation method for spin systems.
\newblock \emph{Physical Review B}, 43\penalty0 (7):\penalty0 5950, 1991.
\newblock \doi{10.1103/PhysRevB.43.5950}.

\bibitem[Schwinger(1962)]{schwinger1962gauge}
Julian Schwinger.
\newblock Gauge invariance and mass. {II}.
\newblock \emph{Physical Review}, 128\penalty0 (5):\penalty0 2425, 1962.
\newblock \doi{10.1103/PhysRev.125.397}.

\bibitem[Chen and Luo(1990)]{chen1990chiral}
Qi-Zhou Chen and Xiang-Qian Luo.
\newblock Chiral-symmetry breaking in the model with wilson fermions.
\newblock \emph{Physical Review D}, 42\penalty0 (4):\penalty0 1293, 1990.
\newblock \doi{10.1103/PhysRevD.42.1293}.

\bibitem[Melnikov and Weinstein(2000)]{melnikov2000lattice}
Kirill Melnikov and Marvin Weinstein.
\newblock Lattice model: Confinement, anomalies, chiral fermions, and all that.
\newblock \emph{Physical Review D}, 62\penalty0 (9):\penalty0 094504, 2000.
\newblock \doi{10.1103/PhysRevD.62.094504}.

\bibitem[Buyens et~al.(2015)Buyens, Haegeman, Verstraete, and Van~Acoleyen]{buyens2015tensor}
Boye Buyens, Jutho Haegeman, Frank Verstraete, and Karel Van~Acoleyen.
\newblock Tensor networks for gauge field theories.
\newblock \emph{arXiv preprint arXiv:1511.04288}, 2015.
\newblock \doi{10.48550/arXiv.1511.04288}.

\bibitem[Banuls et~al.(2019)Banuls, Cichy, Cirac, Jansen, and Kühn]{banuls2018tensor}
Mari~Carmen Banuls, Krzysztof Cichy, J.~Ignacio Cirac, Karl Jansen, and Stefan Kühn.
\newblock Tensor networks and their use for lattice gauge theories.
\newblock \emph{PoS}, LATTICE2018:\penalty0 022, 2019.
\newblock \doi{10.22323/1.334.0022}.

\bibitem[Buyens et~al.(2014)Buyens, Haegeman, Van~Acoleyen, Verschelde, and Verstraete]{buyens2014matrix}
Boye Buyens, Jutho Haegeman, Karel Van~Acoleyen, Henri Verschelde, and Frank Verstraete.
\newblock Matrix product states for gauge field theories.
\newblock \emph{Physical review letters}, 113\penalty0 (9):\penalty0 091601, 2014.
\newblock \doi{10.1103/PhysRevLett.113.091601}.

\bibitem[Zapp and Or{\'u}s(2017)]{zapp2017tensor}
Kai Zapp and Rom{\'a}n Or{\'u}s.
\newblock Tensor network simulation of {QED} on infinite lattices: Learning from (1+1)d, and prospects for (2+1)d.
\newblock \emph{Physical Review D}, 95\penalty0 (11):\penalty0 114508, 2017.
\newblock \doi{10.1103/PhysRevD.95.114508}.

\bibitem[Ba{\~n}uls et~al.(2013)Ba{\~n}uls, Cichy, Cirac, and Jansen]{banuls2013mass}
Mari~Carmen Ba{\~n}uls, K~Cichy, J~Ignacio Cirac, and Karl Jansen.
\newblock The mass spectrum of the {S}chwinger model with matrix product states.
\newblock \emph{Journal of High Energy Physics}, 2013\penalty0 (11):\penalty0 1--21, 2013.
\newblock \doi{10.1007/JHEP11(2013)158}.

\bibitem[Ercolessi et~al.(2018)Ercolessi, Facchi, Magnifico, Pascazio, and Pepe]{ercolessi2018phase}
Elisa Ercolessi, Paolo Facchi, Giuseppe Magnifico, Saverio Pascazio, and Francesco~V Pepe.
\newblock Phase transitions in $\mathbb{Z}_n$ gauge models: Towards quantum simulations of the {S}chwinger-{W}eyl {QED}.
\newblock \emph{Physical Review D}, 98\penalty0 (7):\penalty0 074503, 2018.
\newblock \doi{10.1103/PhysRevD.98.074503}.

\bibitem[Funcke et~al.(2023)Funcke, Jansen, and K{\"u}hn]{funcke2023exploring}
Lena Funcke, Karl Jansen, and Stefan K{\"u}hn.
\newblock Exploring the {CP}-violating dashen phase in the {S}chwinger model with tensor networks.
\newblock \emph{Physical Review D}, 108\penalty0 (1):\penalty0 014504, 2023.
\newblock \doi{10.1103/PhysRevD.108.014504}.

\bibitem[Ba{\~n}uls et~al.(2016)Ba{\~n}uls, Cichy, Jansen, and Saito]{banuls2016chiral}
Mari~Carmen Ba{\~n}uls, Krzysztof Cichy, Karl Jansen, and Hana Saito.
\newblock Chiral condensate in the {S}chwinger model with matrix product operators.
\newblock \emph{Physical Review D}, 93\penalty0 (9):\penalty0 094512, 2016.
\newblock \doi{10.1103/PhysRevD.93.094512}.

\bibitem[Zache et~al.(2022)Zache, Van~Damme, Halimeh, Hauke, and Banerjee]{zache2022toward}
Torsten~V Zache, Maarten Van~Damme, Jad~C Halimeh, Philipp Hauke, and Debasish Banerjee.
\newblock Toward the continuum limit of a (1+1){D} quantum link {S}chwinger model.
\newblock \emph{Physical Review D}, 106\penalty0 (9):\penalty0 L091502, 2022.
\newblock \doi{10.1103/PhysRevD.106.L091502}.

\bibitem[Ba{\~n}uls et~al.(2017)Ba{\~n}uls, Cichy, Cirac, Jansen, and K{\"u}hn]{banuls2017efficient}
Mari~Carmen Ba{\~n}uls, Krzysztof Cichy, J~Ignacio Cirac, Karl Jansen, and Stefan K{\"u}hn.
\newblock Efficient basis formulation for (1+1)-dimensional {SU}(2) lattice gauge theory: Spectral calculations with matrix product states.
\newblock \emph{Physical Review X}, 7\penalty0 (4):\penalty0 041046, 2017.
\newblock \doi{10.1103/PhysRevX.7.041046}.

\bibitem[Pichler et~al.(2016)Pichler, Dalmonte, Rico, Zoller, and Montangero]{pichler2016real}
Thomas Pichler, Marcello Dalmonte, Enrique Rico, Peter Zoller, and Simone Montangero.
\newblock Real-time dynamics in {U}(1) lattice gauge theories with tensor networks.
\newblock \emph{Physical Review X}, 6\penalty0 (1):\penalty0 011023, 2016.
\newblock \doi{10.1103/PhysRevX.6.011023}.

\bibitem[Buyens et~al.(2017{\natexlab{a}})Buyens, Haegeman, Hebenstreit, Verstraete, and Van~Acoleyen]{buyens2017real}
Boye Buyens, Jutho Haegeman, Florian Hebenstreit, Frank Verstraete, and Karel Van~Acoleyen.
\newblock Real-time simulation of the {S}chwinger effect with matrix product states.
\newblock \emph{Physical Review D}, 96\penalty0 (11):\penalty0 114501, 2017{\natexlab{a}}.
\newblock \doi{10.1103/PhysRevD.96.114501}.

\bibitem[Magnifico et~al.(2020)Magnifico, Dalmonte, Facchi, Pascazio, Pepe, and Ercolessi]{magnifico2020real}
Giuseppe Magnifico, Marcello Dalmonte, Paolo Facchi, Saverio Pascazio, Francesco~V Pepe, and Elisa Ercolessi.
\newblock Real time dynamics and confinement in the $\mathbb{Z}_n$ {S}chwinger-{W}eyl lattice model for 1+1 {QED}.
\newblock \emph{Quantum}, 4:\penalty0 281, 2020.
\newblock \doi{10.22331/q-2020-06-15-281}.

\bibitem[Rigobello et~al.(2021)Rigobello, Notarnicola, Magnifico, and Montangero]{rigobello2021entanglement}
Marco Rigobello, Simone Notarnicola, Giuseppe Magnifico, and Simone Montangero.
\newblock Entanglement generation in (1+ 1){D} {QED} scattering processes.
\newblock \emph{Physical Review D}, 104\penalty0 (11):\penalty0 114501, 2021.
\newblock \doi{10.1103/PhysRevD.104.114501}.

\bibitem[Halimeh et~al.(2022)Halimeh, Van~Damme, Zache, Banerjee, and Hauke]{halimeh2022achieving}
Jad~C Halimeh, Maarten Van~Damme, Torsten~V Zache, Debasish Banerjee, and Philipp Hauke.
\newblock Achieving the quantum field theory limit in far-from-equilibrium quantum link models.
\newblock \emph{Quantum}, 6:\penalty0 878, 2022.
\newblock \doi{10.22331/q-2022-12-19-878}.

\bibitem[Bauer et~al.(2023)Bauer, Davoudi, Balantekin, Bhattacharya, Carena, de~Jong, Draper, El-Khadra, Gemelke, Hanada, et~al.]{bauer2023quantum}
Christian~W Bauer, Zohreh Davoudi, A~Baha Balantekin, Tanmoy Bhattacharya, Marcela Carena, Wibe~A de~Jong, Patrick Draper, Aida El-Khadra, Nate Gemelke, Masanori Hanada, et~al.
\newblock Quantum simulation for high-energy physics.
\newblock \emph{PRX Quantum}, 4\penalty0 (2):\penalty0 027001, 2023.
\newblock \doi{10.1103/PRXQuantum.4.027001}.

\bibitem[Di~Meglio et~al.(2024)Di~Meglio, Jansen, Tavernelli, Alexandrou, Arunachalam, Bauer, Borras, Carrazza, Crippa, Croft, et~al.]{di2024quantum}
Alberto Di~Meglio, Karl Jansen, Ivano Tavernelli, Constantia Alexandrou, Srinivasan Arunachalam, Christian~W Bauer, Kerstin Borras, Stefano Carrazza, Arianna Crippa, Vincent Croft, et~al.
\newblock Quantum computing for high-energy physics: state of the art and challenges.
\newblock \emph{PRX Quantum}, 5\penalty0 (3):\penalty0 037001, 2024.
\newblock \doi{10.1103/PRXQuantum.5.037001}.

\bibitem[Banuls et~al.(2020)Banuls, Blatt, Catani, Celi, Cirac, Dalmonte, Fallani, Jansen, Lewenstein, Montangero, et~al.]{banuls2020simulating}
Mari~Carmen Banuls, Rainer Blatt, Jacopo Catani, Alessio Celi, Juan~Ignacio Cirac, Marcello Dalmonte, Leonardo Fallani, Karl Jansen, Maciej Lewenstein, Simone Montangero, et~al.
\newblock Simulating lattice gauge theories within quantum technologies.
\newblock \emph{The European physical journal D}, 74:\penalty0 1--42, 2020.
\newblock \doi{10.1140/epjd/e2020-100571-8}.

\bibitem[Shaw et~al.(2020)Shaw, Lougovski, Stryker, and Wiebe]{shaw2020quantum}
Alexander~F Shaw, Pavel Lougovski, Jesse~R Stryker, and Nathan Wiebe.
\newblock Quantum algorithms for simulating the lattice {S}chwinger model.
\newblock \emph{Quantum}, 4:\penalty0 306, 2020.
\newblock \doi{10.22331/q-2020-08-10-306}.

\bibitem[Low and Chuang(2019)]{low2019hamiltonian}
Guang~Hao Low and Isaac~L Chuang.
\newblock Hamiltonian simulation by qubitization.
\newblock \emph{Quantum}, 3:\penalty0 163, 2019.
\newblock \doi{10.22331/q-2019-07-12-163}.

\bibitem[Campbell(2019)]{campbell2019random}
Earl Campbell.
\newblock Random compiler for fast hamiltonian simulation.
\newblock \emph{Physical review letters}, 123\penalty0 (7):\penalty0 070503, 2019.
\newblock \doi{10.1103/PhysRevLett.123.070503}.

\bibitem[Kan and Nam(2022)]{kan2022simulating}
Angus Kan and Yunseong Nam.
\newblock Simulating lattice quantum electrodynamics on a quantum computer.
\newblock \emph{Quantum Science and Technology}, 8\penalty0 (1):\penalty0 015008, 2022.
\newblock \doi{10.1088/2058-9565/aca0b8}.

\bibitem[Klco et~al.(2018)Klco, Dumitrescu, McCaskey, Morris, Pooser, Sanz, Solano, Lougovski, and Savage]{klco2018quantum}
Natalie Klco, Eugene~F Dumitrescu, Alex~J McCaskey, Titus~D Morris, Raphael~C Pooser, Mikel Sanz, Enrique Solano, Pavel Lougovski, and Martin~J Savage.
\newblock Quantum-classical computation of {S}chwinger model dynamics using quantum computers.
\newblock \emph{Physical Review A}, 98\penalty0 (3):\penalty0 032331, 2018.
\newblock \doi{10.1103/PhysRevA.98.032331}.

\bibitem[Nguyen et~al.(2022)Nguyen, Tran, Zhu, Green, Alderete, Davoudi, and Linke]{nguyen2022digital}
Nhung~H Nguyen, Minh~C Tran, Yingyue Zhu, Alaina~M Green, C~Huerta Alderete, Zohreh Davoudi, and Norbert~M Linke.
\newblock Digital quantum simulation of the {S}chwinger model and symmetry protection with trapped ions.
\newblock \emph{PRX Quantum}, 3\penalty0 (2):\penalty0 020324, 2022.
\newblock \doi{10.1103/PRXQuantum.3.020324}.

\bibitem[Farrell et~al.(2024{\natexlab{a}})Farrell, Illa, Ciavarella, and Savage]{farrell2024quantum}
Roland~C Farrell, Marc Illa, Anthony~N Ciavarella, and Martin~J Savage.
\newblock Quantum simulations of hadron dynamics in the schwinger model using 112 qubits.
\newblock \emph{Physical Review D}, 109\penalty0 (11):\penalty0 114510, 2024{\natexlab{a}}.
\newblock \doi{10.1103/PhysRevD.109.114510}.

\bibitem[de~Jong et~al.(2022)de~Jong, Lee, Mulligan, P{\l}osko{\'n}, Ringer, and Yao]{de2022quantum}
Wibe~A de~Jong, Kyle Lee, James Mulligan, Mateusz P{\l}osko{\'n}, Felix Ringer, and Xiaojun Yao.
\newblock Quantum simulation of nonequilibrium dynamics and thermalization in the {S}chwinger model.
\newblock \emph{Physical Review D}, 106\penalty0 (5):\penalty0 054508, 2022.
\newblock \doi{10.1103/PhysRevD.106.054508}.

\bibitem[Kokail et~al.(2019)Kokail, Maier, van Bijnen, Brydges, Joshi, Jurcevic, Muschik, Silvi, Blatt, Roos, et~al.]{kokail2019self}
Christian Kokail, Christine Maier, Rick van Bijnen, Tiff Brydges, Manoj~K Joshi, Petar Jurcevic, Christine~A Muschik, Pietro Silvi, Rainer Blatt, Christian~F Roos, et~al.
\newblock Self-verifying variational quantum simulation of lattice models.
\newblock \emph{Nature}, 569\penalty0 (7756):\penalty0 355--360, 2019.
\newblock \doi{10.1038/s41586-019-1177-4}.

\bibitem[Avkhadiev et~al.(2020)Avkhadiev, Shanahan, and Young]{avkhadiev2020accelerating}
A~Avkhadiev, PE~Shanahan, and RD~Young.
\newblock Accelerating lattice quantum field theory calculations via interpolator optimization using noisy intermediate-scale quantum computing.
\newblock \emph{Physical review letters}, 124\penalty0 (8):\penalty0 080501, 2020.
\newblock \doi{10.1103/PhysRevLett.124.080501}.

\bibitem[Ferguson et~al.(2021)Ferguson, Dellantonio, Al~Balushi, Jansen, D{\"u}r, and Muschik]{ferguson2021measurement}
Ryan~R Ferguson, Luca Dellantonio, Abdulrahim Al~Balushi, Karl Jansen, Wolfgang D{\"u}r, and Christine~A Muschik.
\newblock Measurement-based variational quantum eigensolver.
\newblock \emph{Physical review letters}, 126\penalty0 (22):\penalty0 220501, 2021.
\newblock \doi{10.1103/PhysRevLett.126.220501}.

\bibitem[Thompson and Siopsis(2022)]{thompson2022quantum}
Shane Thompson and George Siopsis.
\newblock Quantum computation of phase transition in the massive {S}chwinger model.
\newblock \emph{Quantum Science and Technology}, 7\penalty0 (3):\penalty0 035001, 2022.
\newblock \doi{10.1088/2058-9565/ac5f5a}.

\bibitem[Funcke et~al.(2022)Funcke, Hartung, Jansen, K{\"u}hn, Pleinert, Schuster, and von Zanthier]{funcke2022exploring}
Lena Funcke, Tobias Hartung, Karl Jansen, Stefan K{\"u}hn, Marc-Oliver Pleinert, Stephan Schuster, and Joachim von Zanthier.
\newblock Exploring the phase structure of the multi-flavor {S}chwinger model with quantum computing.
\newblock \emph{arXiv preprint arXiv:2211.13020}, 2022.
\newblock \doi{10.48550/arXiv.2211.13020}.

\bibitem[Mazzola et~al.(2021)Mazzola, Mathis, Mazzola, and Tavernelli]{mazzola2021gauge}
Giulia Mazzola, Simon~V Mathis, Guglielmo Mazzola, and Ivano Tavernelli.
\newblock Gauge-invariant quantum circuits for {U}(1) and yang-mills lattice gauge theories.
\newblock \emph{Physical Review Research}, 3\penalty0 (4):\penalty0 043209, 2021.
\newblock \doi{10.1103/PhysRevResearch.3.043209}.

\bibitem[Paulson et~al.(2021)Paulson, Dellantonio, Haase, Celi, Kan, Jena, Kokail, Van~Bijnen, Jansen, Zoller, et~al.]{paulson2021simulating}
Danny Paulson, Luca Dellantonio, Jan~F Haase, Alessio Celi, Angus Kan, Andrew Jena, Christian Kokail, Rick Van~Bijnen, Karl Jansen, Peter Zoller, et~al.
\newblock Simulating 2{D} effects in lattice gauge theories on a quantum computer.
\newblock \emph{PRX Quantum}, 2\penalty0 (3):\penalty0 030334, 2021.
\newblock \doi{10.1103/PRXQuantum.2.030334}.

\bibitem[Lumia et~al.(2022)Lumia, Torta, Mbeng, Santoro, Ercolessi, Burrello, and Wauters]{lumia2022two}
Luca Lumia, Pietro Torta, Glen~B Mbeng, Giuseppe~E Santoro, Elisa Ercolessi, Michele Burrello, and Matteo~M Wauters.
\newblock Two-dimensional $\mathbb{Z}_2$ lattice gauge theory on a near-term quantum simulator: Variational quantum optimization, confinement, and topological order.
\newblock \emph{PRX Quantum}, 3\penalty0 (2):\penalty0 020320, 2022.
\newblock \doi{10.1103/PRXQuantum.3.020320}.

\bibitem[Meth et~al.(2023)Meth, Haase, Zhang, Edmunds, Postler, Steiner, Jena, Dellantonio, Blatt, Zoller, et~al.]{meth2023simulating}
Michael Meth, Jan~F Haase, Jinglei Zhang, Claire Edmunds, Lukas Postler, Alex Steiner, Andrew~J Jena, Luca Dellantonio, Rainer Blatt, Peter Zoller, et~al.
\newblock Simulating 2{D} lattice gauge theories on a qudit quantum computer.
\newblock \emph{arXiv preprint arXiv:2310.12110}, 2023.
\newblock \doi{10.48550/arXiv.2310.12110}.

\bibitem[Farrell et~al.(2024{\natexlab{b}})Farrell, Illa, Ciavarella, and Savage]{farrell2023scalable}
Roland~C Farrell, Marc Illa, Anthony~N Ciavarella, and Martin~J Savage.
\newblock Scalable circuits for preparing ground states on digital quantum computers: The schwinger model vacuum on 100 qubits.
\newblock \emph{PRX Quantum}, 5\penalty0 (2):\penalty0 020315, 2024{\natexlab{b}}.
\newblock \doi{10.1103/PRXQuantum.5.020315}.

\bibitem[McClean et~al.(2018)McClean, Boixo, Smelyanskiy, Babbush, and Neven]{mcclean2018barren}
Jarrod~R McClean, Sergio Boixo, Vadim~N Smelyanskiy, Ryan Babbush, and Hartmut Neven.
\newblock Barren plateaus in quantum neural network training landscapes.
\newblock \emph{Nature communications}, 9:\penalty0 4812, 2018.
\newblock \doi{10.1038/s41467-018-07090-4}.

\bibitem[Wang et~al.(2021)Wang, Fontana, Cerezo, Sharma, Sone, Cincio, and Coles]{wang2021noise}
Samson Wang, Enrico Fontana, Marco Cerezo, Kunal Sharma, Akira Sone, Lukasz Cincio, and Patrick~J Coles.
\newblock Noise-induced barren plateaus in variational quantum algorithms.
\newblock \emph{Nature Communications}, 12:\penalty0 6961, 2021.
\newblock \doi{10.1038/s41467-021-27045-6}.

\bibitem[Cerezo et~al.(2021)Cerezo, Sone, Volkoff, Cincio, and Coles]{cerezo2021cost}
Marco Cerezo, Akira Sone, Tyler Volkoff, Lukasz Cincio, and Patrick~J Coles.
\newblock Cost function dependent barren plateaus in shallow parametrized quantum circuits.
\newblock \emph{Nature communications}, 12:\penalty0 1791, 2021.
\newblock \doi{10.1038/s41467-021-21728-w}.

\bibitem[Marrero et~al.(2021)Marrero, Kieferov{\'a}, and Wiebe]{marrero2021entanglement}
Carlos~Ortiz Marrero, M{\'a}ria Kieferov{\'a}, and Nathan Wiebe.
\newblock Entanglement-induced barren plateaus.
\newblock \emph{PRX Quantum}, 2\penalty0 (4):\penalty0 040316, 2021.
\newblock \doi{10.1103/PRXQuantum.2.040316}.

\bibitem[Uvarov and Biamonte(2021)]{uvarov2021barren}
AV~Uvarov and Jacob~D Biamonte.
\newblock On barren plateaus and cost function locality in variational quantum algorithms.
\newblock \emph{Journal of Physics A: Mathematical and Theoretical}, 54\penalty0 (24):\penalty0 245301, 2021.
\newblock \doi{10.1088/1751-8121/abfac7}.

\bibitem[Holmes et~al.(2022)Holmes, Sharma, Cerezo, and Coles]{holmes2022connecting}
Zo{\"e} Holmes, Kunal Sharma, Marco Cerezo, and Patrick~J Coles.
\newblock Connecting ansatz expressibility to gradient magnitudes and barren plateaus.
\newblock \emph{PRX Quantum}, 3\penalty0 (1):\penalty0 010313, 2022.
\newblock \doi{10.1103/PRXQuantum.3.010313}.

\bibitem[Fontana et~al.(2023)Fontana, Herman, Chakrabarti, Kumar, Yalovetzky, Heredge, Sureshbabu, and Pistoia]{fontana2023adjoint}
Enrico Fontana, Dylan Herman, Shouvanik Chakrabarti, Niraj Kumar, Romina Yalovetzky, Jamie Heredge, Shree~Hari Sureshbabu, and Marco Pistoia.
\newblock The adjoint is all you need: Characterizing barren plateaus in quantum ans\"atze.
\newblock \emph{arXiv preprint arXiv:2309.07902}, 2023.
\newblock \doi{10.48550/arXiv.2309.07902}.

\bibitem[Ragone et~al.(2023)Ragone, Bakalov, Sauvage, Kemper, Marrero, Larocca, and Cerezo]{ragone2023unified}
Michael Ragone, Bojko~N Bakalov, Fr{\'e}d{\'e}ric Sauvage, Alexander~F Kemper, Carlos~Ortiz Marrero, Martin Larocca, and M~Cerezo.
\newblock A unified theory of barren plateaus for deep parametrized quantum circuits.
\newblock \emph{arXiv preprint arXiv:2309.09342}, 2023.
\newblock \doi{10.48550/arXiv.2309.09342}.

\bibitem[Moosavian et~al.(2019)Moosavian, Garrison, and Jordan]{moosavian2019site}
Ali~Hamed Moosavian, James~R Garrison, and Stephen~P Jordan.
\newblock Site-by-site quantum state preparation algorithm for preparing vacua of fermionic lattice field theories.
\newblock \emph{arXiv preprint arXiv:1911.03505}, 2019.
\newblock \doi{10.48550/arXiv.1911.03505}.

\bibitem[Hamed~Moosavian and Jordan(2018)]{hamed2018faster}
Ali Hamed~Moosavian and Stephen Jordan.
\newblock Faster quantum algorithm to simulate fermionic quantum field theory.
\newblock \emph{Physical Review A}, 98\penalty0 (1):\penalty0 012332, 2018.
\newblock \doi{10.1103/PhysRevA.98.012332}.

\bibitem[Jordan et~al.(2014)Jordan, Lee, and Preskill]{jordan2014quantum}
Stephen~P Jordan, Keith~SM Lee, and John Preskill.
\newblock Quantum algorithms for fermionic quantum field theories.
\newblock \emph{arXiv preprint arXiv:1404.7115}, 2014.
\newblock \doi{10.48550/arXiv.1404.7115}.

\bibitem[McClean et~al.(2017)McClean, Kimchi-Schwartz, Carter, and De~Jong]{mcclean2017hybrid}
Jarrod~R McClean, Mollie~E Kimchi-Schwartz, Jonathan Carter, and Wibe~A De~Jong.
\newblock Hybrid quantum-classical hierarchy for mitigation of decoherence and determination of excited states.
\newblock \emph{Physical Review A}, 95\penalty0 (4):\penalty0 042308, 2017.
\newblock \doi{10.1103/PhysRevA.95.042308}.

\bibitem[McClean et~al.(2020)McClean, Jiang, Rubin, Babbush, and Neven]{mcclean2020decoding}
Jarrod~R McClean, Zhang Jiang, Nicholas~C Rubin, Ryan Babbush, and Hartmut Neven.
\newblock Decoding quantum errors with subspace expansions.
\newblock \emph{Nature communications}, 11\penalty0 (1):\penalty0 636, 2020.
\newblock \doi{10.1038/s41467-020-14341-w}.

\bibitem[Yoshioka et~al.(2022{\natexlab{a}})Yoshioka, Hakoshima, Matsuzaki, Tokunaga, Suzuki, and Endo]{yoshioka2022generalized}
Nobuyuki Yoshioka, Hideaki Hakoshima, Yuichiro Matsuzaki, Yuuki Tokunaga, Yasunari Suzuki, and Suguru Endo.
\newblock Generalized quantum subspace expansion.
\newblock \emph{Physical Review Letters}, 129\penalty0 (2):\penalty0 020502, 2022{\natexlab{a}}.
\newblock \doi{10.1103/PhysRevLett.129.020502}.

\bibitem[Yoshioka et~al.(2022{\natexlab{b}})Yoshioka, Sato, Nakagawa, Ohnishi, and Mizukami]{yoshioka2022variational}
Nobuyuki Yoshioka, Takeshi Sato, Yuya~O Nakagawa, Yu-ya Ohnishi, and Wataru Mizukami.
\newblock Variational quantum simulation for periodic materials.
\newblock \emph{Physical Review Research}, 4\penalty0 (1):\penalty0 013052, 2022{\natexlab{b}}.
\newblock \doi{PhysRevResearch.4.013052}.

\bibitem[Kirby(2024)]{kirby2024analysis}
William Kirby.
\newblock Analysis of quantum {K}rylov algorithms with errors.
\newblock \emph{{Quantum}}, 8:\penalty0 1457, 2024.
\newblock ISSN 2521-327X.
\newblock \doi{10.22331/q-2024-08-29-1457}.

\bibitem[Kirby et~al.(2023)Kirby, Motta, and Mezzacapo]{kirby2023exact}
William Kirby, Mario Motta, and Antonio Mezzacapo.
\newblock Exact and efficient lanczos method on a quantum computer.
\newblock \emph{Quantum}, 7:\penalty0 1018, 2023.
\newblock \doi{10.22331/q-2023-05-23-1018}.

\bibitem[Baker(2021)]{baker2021lanczos}
Thomas~E Baker.
\newblock Lanczos recursion on a quantum computer for the {G}reen's function and ground state.
\newblock \emph{Physical Review A}, 103\penalty0 (3):\penalty0 032404, 2021.
\newblock \doi{10.1103/PhysRevA.103.032404}.

\bibitem[Dong et~al.(2022)Dong, Lin, and Tong]{dong2022ground}
Yulong Dong, Lin Lin, and Yu~Tong.
\newblock Ground-state preparation and energy estimation on early fault-tolerant quantum computers via quantum eigenvalue transformation of unitary matrices.
\newblock \emph{PRX quantum}, 3\penalty0 (4):\penalty0 040305, 2022.
\newblock \doi{10.1103/PRXQuantum.3.040305}.

\bibitem[Kane et~al.(2024)Kane, Gomes, and Kreshchuk]{kane2024nearly}
Christopher~F Kane, Niladri Gomes, and Michael Kreshchuk.
\newblock Nearly optimal state preparation for quantum simulations of lattice gauge theories.
\newblock \emph{Physical Review A}, 110\penalty0 (1):\penalty0 012455, 2024.

\bibitem[O'Leary et~al.(2024)O'Leary, Anderson, Jaksch, and Kiffner]{oleary2024partitioned}
Tom O'Leary, Lewis~W Anderson, Dieter Jaksch, and Martin Kiffner.
\newblock Partitioned quantum subspace expansion.
\newblock \emph{arXiv preprint arXiv:2403.08868}, 2024.
\newblock \doi{10.48550/arXiv.2403.08868}.

\bibitem[Epperly et~al.(2022)Epperly, Lin, and Nakatsukasa]{epperly2022theory}
Ethan~N Epperly, Lin Lin, and Yuji Nakatsukasa.
\newblock A theory of quantum subspace diagonalization.
\newblock \emph{SIAM Journal on Matrix Analysis and Applications}, 43\penalty0 (3):\penalty0 1263--1290, 2022.
\newblock \doi{10.1137/21M145954X}.

\bibitem[Parrish and McMahon(2019)]{parrish2019quantum}
Robert~M Parrish and Peter~L McMahon.
\newblock Quantum filter diagonalization: Quantum eigendecomposition without full quantum phase estimation.
\newblock \emph{arXiv preprint arXiv:1909.08925}, 2019.
\newblock \doi{10.48550/arXiv.1909.08925}.

\bibitem[Stair et~al.(2020)Stair, Huang, and Evangelista]{stair2020multireference}
Nicholas~H Stair, Renke Huang, and Francesco~A Evangelista.
\newblock A multireference quantum krylov algorithm for strongly correlated electrons.
\newblock \emph{Journal of chemical theory and computation}, 16\penalty0 (4):\penalty0 2236--2245, 2020.
\newblock \doi{10.1021/acs.jctc.9b01125}.

\bibitem[Cortes and Gray(2022)]{cortes2022quantum}
Cristian~L Cortes and Stephen~K Gray.
\newblock Quantum krylov subspace algorithms for ground-and excited-state energy estimation.
\newblock \emph{Physical Review A}, 105\penalty0 (2):\penalty0 022417, 2022.
\newblock \doi{10.1103/PhysRevA.105.022417}.

\bibitem[Stair et~al.(2023)Stair, Cortes, Parrish, Cohn, and Motta]{stair2023stochastic}
Nicholas~H Stair, Cristian~L Cortes, Robert~M Parrish, Jeffrey Cohn, and Mario Motta.
\newblock Stochastic quantum krylov protocol with double-factorized hamiltonians.
\newblock \emph{Physical Review A}, 107\penalty0 (3):\penalty0 032414, 2023.
\newblock \doi{10.1103/PhysRevA.107.032414}.

\bibitem[Motta et~al.(2020)Motta, Sun, Tan, O’Rourke, Ye, Minnich, Brandao, and Chan]{motta2020determining}
Mario Motta, Chong Sun, Adrian~TK Tan, Matthew~J O’Rourke, Erika Ye, Austin~J Minnich, Fernando~GSL Brandao, and Garnet Kin-Lic Chan.
\newblock Determining eigenstates and thermal states on a quantum computer using quantum imaginary time evolution.
\newblock \emph{Nature Physics}, 16\penalty0 (2):\penalty0 205--210, 2020.
\newblock \doi{10.1038/s41567-019-0704-4}.

\bibitem[Yeter-Aydeniz et~al.(2020)Yeter-Aydeniz, Pooser, and Siopsis]{yeter2020practical}
K{\"u}bra Yeter-Aydeniz, Raphael~C Pooser, and George Siopsis.
\newblock Practical quantum computation of chemical and nuclear energy levels using quantum imaginary time evolution and lanczos algorithms.
\newblock \emph{npj Quantum Information}, 6\penalty0 (1):\penalty0 63, 2020.
\newblock \doi{10.1038/s41534-020-00290-1}.

\bibitem[Takeshita et~al.(2020)Takeshita, Rubin, Jiang, Lee, Babbush, and McClean]{takeshita2020increasing}
Tyler Takeshita, Nicholas~C Rubin, Zhang Jiang, Eunseok Lee, Ryan Babbush, and Jarrod~R McClean.
\newblock Increasing the representation accuracy of quantum simulations of chemistry without extra quantum resources.
\newblock \emph{Physical Review X}, 10\penalty0 (1):\penalty0 011004, 2020.
\newblock \doi{10.1103/PhysRevX.10.011004}.

\bibitem[Colless et~al.(2018)Colless, Ramasesh, Dahlen, Blok, Kimchi-Schwartz, McClean, Carter, de~Jong, and Siddiqi]{colless2018computation}
James~I Colless, Vinay~V Ramasesh, Dar Dahlen, Machiel~S Blok, Mollie~E Kimchi-Schwartz, Jarrod~R McClean, Jonathan Carter, Wibe~A de~Jong, and Irfan Siddiqi.
\newblock Computation of molecular spectra on a quantum processor with an error-resilient algorithm.
\newblock \emph{Physical Review X}, 8\penalty0 (1):\penalty0 011021, 2018.
\newblock \doi{10.1103/PhysRevX.8.011021}.

\bibitem[Huggins et~al.(2020)Huggins, Lee, Baek, O’Gorman, and Whaley]{huggins2020non}
William~J Huggins, Joonho Lee, Unpil Baek, Bryan O’Gorman, and K~Birgitta Whaley.
\newblock A non-orthogonal variational quantum eigensolver.
\newblock \emph{New Journal of Physics}, 22\penalty0 (7):\penalty0 073009, 2020.
\newblock \doi{10.1088/1367-2630/ab867b}.

\bibitem[Barenco et~al.(1997)Barenco, Berthiaume, Deutsch, Ekert, Jozsa, and Macchiavello]{barenco1997stabilization}
Adriano Barenco, Andre Berthiaume, David Deutsch, Artur Ekert, Richard Jozsa, and Chiara Macchiavello.
\newblock Stabilization of quantum computations by symmetrization.
\newblock \emph{SIAM Journal on Computing}, 26\penalty0 (5):\penalty0 1541--1557, 1997.
\newblock \doi{10.1137/S0097539796302452}.

\bibitem[Buhrman et~al.(2001)Buhrman, Cleve, Watrous, and De~Wolf]{buhrman2001quantum}
Harry Buhrman, Richard Cleve, John Watrous, and Ronald De~Wolf.
\newblock Quantum fingerprinting.
\newblock \emph{Physical Review Letters}, 87\penalty0 (16):\penalty0 167902, 2001.
\newblock \doi{10.1103/PhysRevLett.87.167902}.

\bibitem[Lanczos(1950)]{lanczos1950iteration}
Cornelius Lanczos.
\newblock An iteration method for the solution of the eigenvalue problem of linear differential and integral operators.
\newblock \emph{Journal of Research of the National Bureau of Standards}, 50, 1950.
\newblock \doi{10.6028/jres.045.026}.

\bibitem[Park and Light(1986)]{park1986unitary}
Tae~Jun Park and JC~Light.
\newblock Unitary quantum time evolution by iterative lanczos reduction.
\newblock \emph{The Journal of chemical physics}, 85\penalty0 (10):\penalty0 5870--5876, 1986.
\newblock \doi{10.1063/1.451548}.

\bibitem[Cullum and Willoughby(2002)]{cullum2002lanczos}
Jane~K Cullum and Ralph~A Willoughby.
\newblock \emph{Lanczos algorithms for large symmetric eigenvalue computations: Vol. I: Theory}.
\newblock SIAM, 2002.
\newblock \doi{10.1137/1.9780898719192}.

\bibitem[Paige(1971)]{paige1971computation}
Christopher~Conway Paige.
\newblock \emph{The computation of eigenvalues and eigenvectors of very large sparse matrices}.
\newblock PhD thesis, University of London, 1971.

\bibitem[Saad(1980)]{saad1980rates}
Yousef Saad.
\newblock On the rates of convergence of the lanczos and the block-lanczos methods.
\newblock \emph{SIAM Journal on Numerical Analysis}, 17\penalty0 (5):\penalty0 687--706, 1980.
\newblock \doi{10.1137/0717059}.

\bibitem[Seki and Yunoki(2021)]{seki2021quantum}
Kazuhiro Seki and Seiji Yunoki.
\newblock Quantum power method by a superposition of time-evolved states.
\newblock \emph{PRX Quantum}, 2\penalty0 (1):\penalty0 010333, 2021.
\newblock \doi{10.1103/PRXQuantum.2.010333}.

\bibitem[Banks et~al.(1976)Banks, Susskind, and Kogut]{banks1976strong}
Tom Banks, Leonard Susskind, and John Kogut.
\newblock Strong-coupling calculations of lattice gauge theories:(1+ 1)-dimensional exercises.
\newblock \emph{Physical Review D}, 13\penalty0 (4):\penalty0 1043, 1976.
\newblock \doi{10.1103/PhysRevD.13.1043}.

\bibitem[Susskind(1977)]{susskind1977lattice}
Leonard Susskind.
\newblock Lattice fermions.
\newblock \emph{Physical Review D}, 16\penalty0 (10):\penalty0 3031, 1977.
\newblock \doi{10.1103/PhysRevD.16.3031}.

\bibitem[Hamer et~al.(1997)Hamer, Weihong, and Oitmaa]{hamer1997series}
CJ~Hamer, Zheng Weihong, and J~Oitmaa.
\newblock Series expansions for the massive model in hamiltonian lattice theory.
\newblock \emph{Physical Review D}, 56\penalty0 (1):\penalty0 55, 1997.
\newblock \doi{10.1103/PhysRevD.56.55}.

\bibitem[Hauke et~al.(2013)Hauke, Marcos, Dalmonte, and Zoller]{hauke2013quantum}
Philipp Hauke, David Marcos, Marcello Dalmonte, and Peter Zoller.
\newblock Quantum simulation of a lattice {S}chwinger model in a chain of trapped ions.
\newblock \emph{Physical Review X}, 3\penalty0 (4):\penalty0 041018, 2013.
\newblock \doi{10.1103/PhysRevX.3.041018}.

\bibitem[Tong et~al.(2022)Tong, Albert, McClean, Preskill, and Su]{tong2022provably}
Yu~Tong, Victor~V Albert, Jarrod~R McClean, John Preskill, and Yuan Su.
\newblock Provably accurate simulation of gauge theories and bosonic systems.
\newblock \emph{Quantum}, 6:\penalty0 816, 2022.
\newblock \doi{10.22331/q-2022-09-22-816}.

\bibitem[K{\"u}hn et~al.(2014)K{\"u}hn, Cirac, and Ba{\~n}uls]{kuhn2014quantum}
Stefan K{\"u}hn, J~Ignacio Cirac, and Mari-Carmen Ba{\~n}uls.
\newblock Quantum simulation of the {S}chwinger model: A study of feasibility.
\newblock \emph{Physical Review A}, 90\penalty0 (4):\penalty0 042305, 2014.
\newblock \doi{10.1103/PhysRevA.90.042305}.

\bibitem[Buyens et~al.(2017{\natexlab{b}})Buyens, Montangero, Haegeman, Verstraete, and Van~Acoleyen]{buyens2017finite}
Boye Buyens, Simone Montangero, Jutho Haegeman, Frank Verstraete, and Karel Van~Acoleyen.
\newblock Finite-representation approximation of lattice gauge theories at the continuum limit with tensor networks.
\newblock \emph{Physical Review D}, 95\penalty0 (9):\penalty0 094509, 2017{\natexlab{b}}.
\newblock \doi{10.1103/PhysRevD.95.094509}.

\bibitem[Funcke et~al.(2020)Funcke, Jansen, and K{\"u}hn]{funcke2020topological}
Lena Funcke, Karl Jansen, and Stefan K{\"u}hn.
\newblock Topological vacuum structure of the schwinger model with matrix product states.
\newblock \emph{Physical Review D}, 101\penalty0 (5):\penalty0 054507, 2020.
\newblock \doi{10.1103/PhysRevD.101.054507}.

\bibitem[Martyn et~al.(2021)Martyn, Rossi, Tan, and Chuang]{martyn2021grand}
John~M Martyn, Zane~M Rossi, Andrew~K Tan, and Isaac~L Chuang.
\newblock Grand unification of quantum algorithms.
\newblock \emph{PRX Quantum}, 2\penalty0 (4):\penalty0 040203, 2021.
\newblock \doi{10.1103/PRXQuantum.2.040203}.

\bibitem[Kikuchi et~al.(2023)Kikuchi, Mc~Keever, Coopmans, Lubasch, and Benedetti]{kikuchi2023realization}
Yuta Kikuchi, Conor Mc~Keever, Luuk Coopmans, Michael Lubasch, and Marcello Benedetti.
\newblock Realization of quantum signal processing on a noisy quantum computer.
\newblock \emph{npj Quantum Information}, 9\penalty0 (1):\penalty0 93, 2023.
\newblock \doi{10.1038/s41534-023-00762-0}.

\bibitem[Gily{\'e}n et~al.(2019)Gily{\'e}n, Su, Low, and Wiebe]{gilyen2019quantum}
Andr{\'a}s Gily{\'e}n, Yuan Su, Guang~Hao Low, and Nathan Wiebe.
\newblock Quantum singular value transformation and beyond: exponential improvements for quantum matrix arithmetics.
\newblock In \emph{Proceedings of the 51st Annual ACM SIGACT Symposium on Theory of Computing}, pages 193--204, 2019.
\newblock \doi{10.1145/3313276.3316366}.

\bibitem[Ross and Selinger(2016)]{ross2014optimal}
Neil~J. Ross and Peter Selinger.
\newblock Optimal ancilla-free {C}lifford+{T} approximation of z-rotations.
\newblock \emph{Quantum Info. Comput.}, 16\penalty0 (11–12):\penalty0 901–953, sep 2016.
\newblock ISSN 1533-7146.
\newblock \doi{10.5555/3179330.3179331}.

\bibitem[Litinski(2019)]{litinski2019magic}
Daniel Litinski.
\newblock Magic state distillation: Not as costly as you think.
\newblock \emph{Quantum}, 3:\penalty0 205, 2019.
\newblock \doi{10.22331/q-2019-12-02-205}.

\bibitem[Oliveira and Terhal(2008)]{oliveira2005complexity}
Roberto Oliveira and Barbara~M. Terhal.
\newblock The complexity of quantum spin systems on a two-dimensional square lattice.
\newblock \emph{Quantum Info. Comput.}, 8\penalty0 (10):\penalty0 900–924, nov 2008.
\newblock ISSN 1533-7146.
\newblock \doi{10.5555/2016985.2016987}.

\bibitem[Liu et~al.(2007)Liu, Christandl, and Verstraete]{liu2007quantum}
Yi-Kai Liu, Matthias Christandl, and Frank Verstraete.
\newblock Quantum computational complexity of the {N}-representability problem: {QMA} complete.
\newblock \emph{Physical review letters}, 98\penalty0 (11):\penalty0 110503, 2007.
\newblock \doi{10.1103/PhysRevLett.98.110503}.

\bibitem[Schuch and Verstraete(2009)]{schuch2009computational}
Norbert Schuch and Frank Verstraete.
\newblock Computational complexity of interacting electrons and fundamental limitations of density functional theory.
\newblock \emph{Nature physics}, 5\penalty0 (10):\penalty0 732--735, 2009.
\newblock \doi{10.1038/nphys1370}.

\bibitem[Bookatz(2014)]{bookatz2012qma}
Adam~D. Bookatz.
\newblock {QMA}-complete problems.
\newblock \emph{Quantum Info. Comput.}, 14\penalty0 (5 \& 6):\penalty0 361–383, apr 2014.
\newblock ISSN 1533-7146.
\newblock \doi{10.5555/2638661.2638662}.

\bibitem[Bespalova and Kyriienko(2021)]{bespalova2021hamiltonian}
Tatiana~A Bespalova and Oleksandr Kyriienko.
\newblock Hamiltonian operator approximation for energy measurement and ground-state preparation.
\newblock \emph{PRX Quantum}, 2\penalty0 (3):\penalty0 030318, 2021.
\newblock \doi{10.1103/PRXQuantum.2.030318}.

\bibitem[Shen et~al.(2023)Shen, Klymko, Sud, Williams-Young, de~Jong, and Tubman]{shen2023real}
Yizhi Shen, Katherine Klymko, James Sud, David~B Williams-Young, Wibe~A de~Jong, and Norm~M Tubman.
\newblock Real-time krylov theory for quantum computing algorithms.
\newblock \emph{Quantum}, 7:\penalty0 1066, 2023.
\newblock \doi{10.22331/q-2023-07-25-1066}.

\bibitem[Zhang et~al.(2023)Zhang, Wang, Xu, and Li]{zhang2023measurement}
Zongkang Zhang, Anbang Wang, Xiaosi Xu, and Ying Li.
\newblock Measurement-efficient quantum krylov subspace diagonalisation.
\newblock \emph{arXiv preprint arXiv:2301.13353}, 2023.
\newblock \doi{10.48550/arXiv.2301.13353}.

\bibitem[Berry et~al.(2019)Berry, Gidney, Motta, McClean, and Babbush]{berry2019qubitization}
Dominic~W Berry, Craig Gidney, Mario Motta, Jarrod~R McClean, and Ryan Babbush.
\newblock Qubitization of arbitrary basis quantum chemistry leveraging sparsity and low rank factorization.
\newblock \emph{Quantum}, 3:\penalty0 208, 2019.
\newblock \doi{10.22331/q-2019-12-02-208}.

\bibitem[ric()]{richards2015university}
\url{https://doi.org/10.5281/zenodo.22558}.

\bibitem[Anderson et~al.(2022)Anderson, Kiffner, Barkoutsos, Tavernelli, Crain, and Jaksch]{anderson2022coarse}
Lewis~W Anderson, Martin Kiffner, Panagiotis~Kl Barkoutsos, Ivano Tavernelli, Jason Crain, and Dieter Jaksch.
\newblock Coarse-grained intermolecular interactions on quantum processors.
\newblock \emph{Physical Review A}, 105\penalty0 (6):\penalty0 062409, 2022.
\newblock \doi{10.1103/PhysRevA.105.062409}.

\bibitem[Ganis(1959)]{ganis1959notes}
Sam~E. Ganis.
\newblock Notes on the {F}ibonacci sequence.
\newblock \emph{The American Mathematical Monthly}, 66\penalty0 (2):\penalty0 129--130, 1959.
\newblock ISSN 00029890, 19300972.
\newblock \doi{10.2307/2310016}.

\bibitem[He et~al.(2017)He, Luo, Zhang, Wang, and Wang]{he2017decompositions}
Yong He, Ming-Xing Luo, E~Zhang, Hong-Ke Wang, and Xiao-Feng Wang.
\newblock Decompositions of n-qubit {T}offoli gates with linear circuit complexity.
\newblock \emph{International Journal of Theoretical Physics}, 56:\penalty0 2350--2361, 2017.
\newblock \doi{10.1007/s10773-017-3389-4}.

\bibitem[Vatan and Williams(2004)]{vatan2004optimal}
Farrokh Vatan and Colin Williams.
\newblock Optimal quantum circuits for general two-qubit gates.
\newblock \emph{Physical Review A}, 69\penalty0 (3):\penalty0 032315, 2004.
\newblock \doi{10.1103/PhysRevA.69.032315}.

\bibitem[Smolin and DiVincenzo(1996)]{smolin1996five}
John~A Smolin and David~P DiVincenzo.
\newblock Five two-bit quantum gates are sufficient to implement the quantum fredkin gate.
\newblock \emph{Physical Review A}, 53\penalty0 (4):\penalty0 2855, 1996.
\newblock \doi{10.1103/PhysRevA.53.2855}.

\bibitem[Baker et~al.(2019)Baker, Duckering, Hoover, and Chong]{baker2019decomposing}
Jonathan~M Baker, Casey Duckering, Alexander Hoover, and Frederic~T Chong.
\newblock Decomposing quantum generalized toffoli with an arbitrary number of ancilla.
\newblock \emph{arXiv preprint arXiv:1904.01671}, 2019.
\newblock \doi{10.48550/arXiv.1904.01671}.

\bibitem[Zindorf and Bose(2024)]{zindorf2024efficient}
Ben Zindorf and Sougato Bose.
\newblock Efficient implementation of multi-controlled quantum gates.
\newblock \emph{arXiv preprint arXiv:2404.02279}, 2024.
\newblock \doi{10.48550/arXiv.2404.02279}.

\bibitem[Moses et~al.(2023)Moses, Baldwin, Allman, Ancona, Ascarrunz, Barnes, Bartolotta, Bjork, Blanchard, Bohn, et~al.]{moses2023race}
Steven~A Moses, Charles~H Baldwin, Michael~S Allman, R~Ancona, L~Ascarrunz, C~Barnes, J~Bartolotta, B~Bjork, P~Blanchard, M~Bohn, et~al.
\newblock A race-track trapped-ion quantum processor.
\newblock \emph{Physical Review X}, 13\penalty0 (4):\penalty0 041052, 2023.
\newblock \doi{10.1103/PhysRevX.13.041052}.

\bibitem[h2d(2023)]{h2datasheet}
Quantinuum system model h2: Emulator data sheet, version 1.1.
\newblock \url{www.quantinuum.com/hardware/h2}, 2023.

\bibitem[Chen et~al.(2023)Chen, Nielsen, Ebert, Inlek, Wright, Chaplin, Maksymov, P{\'a}ez, Poudel, Maunz, et~al.]{chen2023benchmarking}
Jwo-Sy Chen, Erik Nielsen, Matthew Ebert, Volkan Inlek, Kenneth Wright, Vandiver Chaplin, Andrii Maksymov, Eduardo P{\'a}ez, Amrit Poudel, Peter Maunz, et~al.
\newblock Benchmarking a trapped-ion quantum computer with 29 algorithmic qubits.
\newblock \emph{arXiv preprint arXiv:2308.05071}, 2023.
\newblock \doi{10.48550/arXiv.2308.05071}.

\bibitem[ion(2022)]{ionqforte}
Ionq forte.
\newblock \url{www.ionq.com/quantum-systems/forte}, 2022.

\bibitem[ibm()]{ibmeagle}
IBM Kyiv processor, accessed 16/09/2023 via IBM Quantum experience.

\bibitem[Morvan et~al.(2024)Morvan, Villalonga, Mi, Mandra, Bengtsson, Klimov, Chen, Hong, Erickson, Drozdov, et~al.]{morvan2024phase}
Alexis Morvan, B~Villalonga, X~Mi, S~Mandra, A~Bengtsson, PV~Klimov, Z~Chen, S~Hong, C~Erickson, IK~Drozdov, et~al.
\newblock Phase transitions in random circuit sampling.
\newblock \emph{Nature}, 634\penalty0 (8033):\penalty0 328--333, 2024.
\newblock \doi{10.1038/s41586-024-07998-6}.

\end{thebibliography}
\onecolumn
\newpage

\appendix

\section{Numerical experiments and fitting}\label{app: fitting}
\subsection{Noiseless case} \label{app: noiseless fitting}

These results are shown in Fig.~\ref{fig: results vs order fit}(a). We see that for each value of $N$, the fractional energy error achieved by the algorithm decreases exponentially with Krylov basis size. As seen in Fig.\ref{fig: results vs order fit}(b), for the system sizes considered, our chosen initial state has reasonably large overlap with the true ground state. This overlap decreases linearly with system size $N$.

To estimate the basis size required to achieve a given accuracy we perform a fit of the form $\log\Delta E/E_\textrm{int} = \chi D+\lambda$ for each value of $N$, where $\Delta E/E_\textrm{int}$ is the fractional energy error, and $D$ is the Krylov basis order. $\chi$ and $\lambda$ are $N$ dependent parameters that correspond to the gradient and intercept of the lines in Fig.~\ref{fig: results vs order fit}~(a). For each value of $N$, we fit to data-points up to $D=10$ ($D=4$ in the case of $N=4$ only) with the fitting parameters $\chi$ and $\lambda$ shown in Figs.~\ref{fig: results vs order fit}~(c) and (d) respectively. Using these fitting parameters, we are able to estimate the Krylov basis size required to compute the ground state energy within a desired precision, there is no noise due to measurement or numerical precision (and therefore no risk of instability in the generalised eigenvalue problem).

\begin{figure*}
    \centering
    \includegraphics[width=\linewidth]{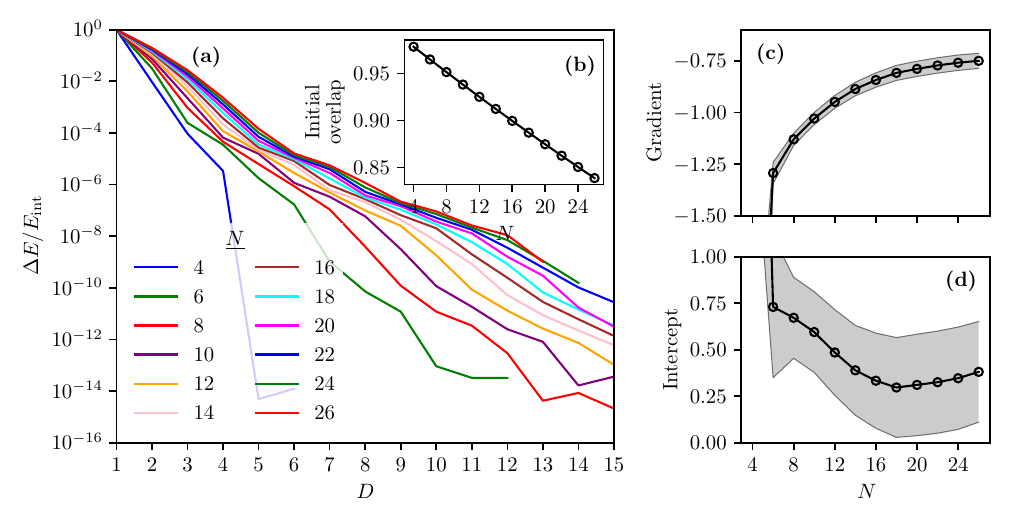}
    \caption[Fractional energy error as function of Krylov order and fitting parameters as function of lattice size in the case of zero measurement error]{\textbf{(a)} Fractional energy error as a function of maximum Krylov order for different value of lattice size $N$ for $\mu=1.5$, $x=5$. No measurement noise has been applied. \textbf{(b)} Overlap between initial state $\ket{\psi_0}$ and exact ground state. \textbf{(c)} \& \textbf{(d)} Gradient and intercept for straight lines (on log-linear plot) fit to data in (a). Filled area indicates standard errors on estimates.}
    \label{fig: results vs order fit}
    \vspace{1em}
    \includegraphics[width=\linewidth]{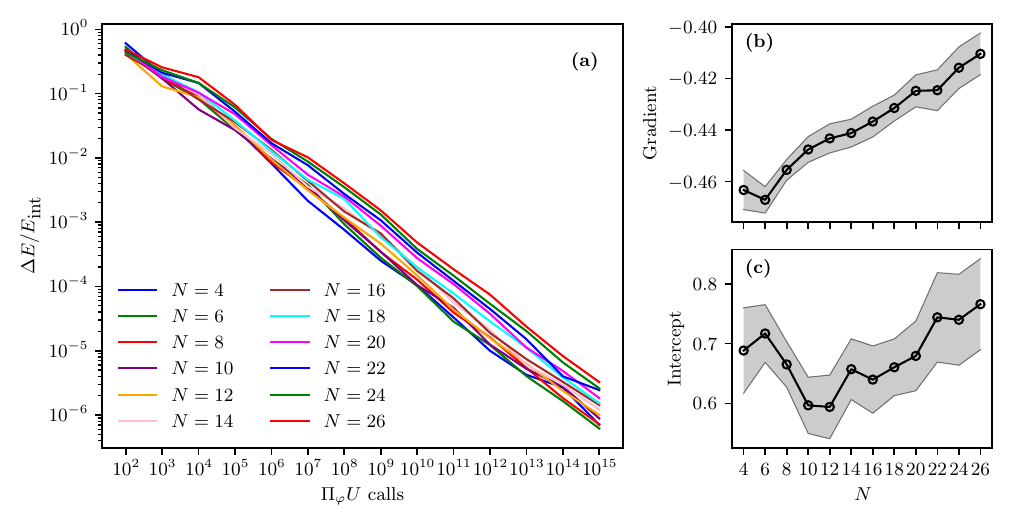}
    \caption[Fractional energy error as function of calls to qubitization procedure and fitting parameters as function of lattice size]{Median fractional energy error achieved in Fig.~\ref{fig: results vs calls} for optimal choices of maximum Krylov order as a function of calls to $\Pi_\varphi U$. \textbf{(b)} \& \textbf{(c)} Gradient and intercept for straight lines (on log-log plot) fit to data in (a). Filled area indicates standard errors on estimates.}
    \label{fig: results vs calls fit}
\end{figure*}

\subsection{Using partitioned QSE with measurement noise}\label{app: noise fitting}

The best median error achieved by PQSE as a function of $\Pi_\varphi U$ calls is shown in Fig.~\ref{fig: results vs calls fit} (a) for each value of $N$. Unlike for the noiseless case in which we saw an approximately log-linear relationship between the error and basis size, for a given $N$ we observe a strong log-log dependence between the fractional energy error and the circuit cost (number of calls to $\Pi_\varphi U$). For each value of $N$, we fit to this data using a function of the form $\log\Delta E/E_\textrm{int} = \chi\log(\#\Pi_\varphi U\text{ calls}) + \lambda$. Fitting parameters $\chi$ and $\lambda$ are shown as a function of $N$ in Figs.~\ref{fig: results vs calls fit} (b) and (c). Using these fitting parameters allows us to estimate the number of calls required to achieve a specified accuracy for a given value of $N$. 

\subsection{Using thresholded QSE with measurement noise}\label{app: more qse simulations}

\begin{figure*}
    \centering
    \includegraphics[width=\linewidth]{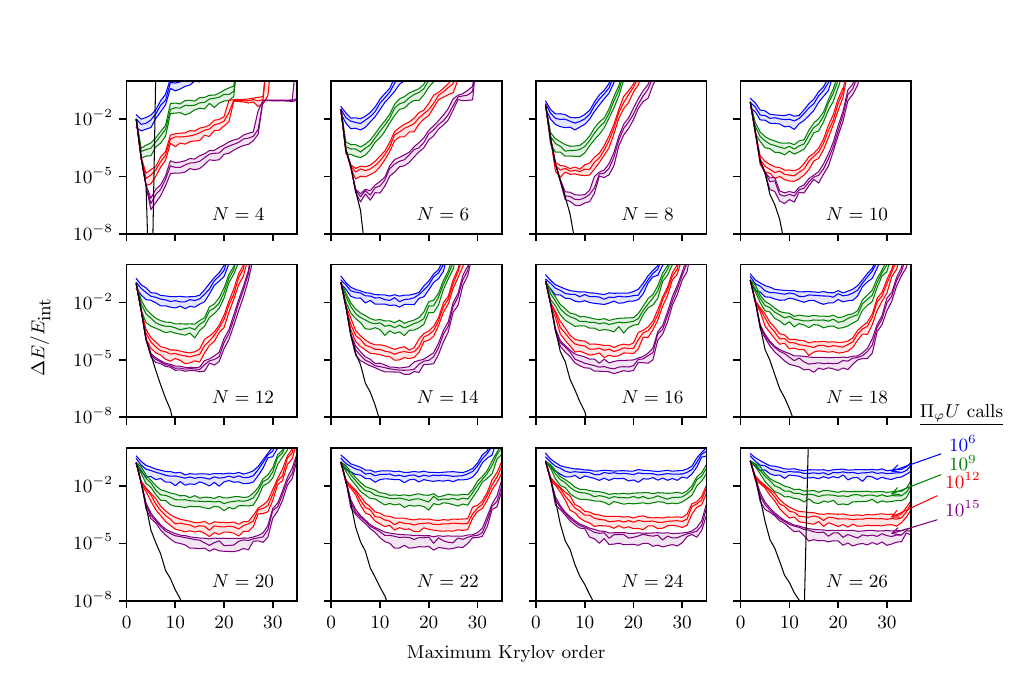}
    \caption[Fractional energy error using TQSE procedure with finite number of measurements for range of different lattice sizes.]{Fractional energy error using TQSE procedure with finite number of measurements for range of different lattice sizes $N$. Values are shown as a function of the Krylov basis size for the TQSE procedure and as a function of number of calls to the block encoding operator $\Pi_\varphi U$ (which affect the number of measurements as defined in the main text). Lines and shaded area corresponds to the median as well as upper and lower quartiles of energy error. The black line (labelled as infinite calls to $\Pi_\varphi U$) corresponds to no measurement noise.}
    \label{fig: results vs calls thqse}
    \vspace{10pt}
    \includegraphics[width=0.5\linewidth]{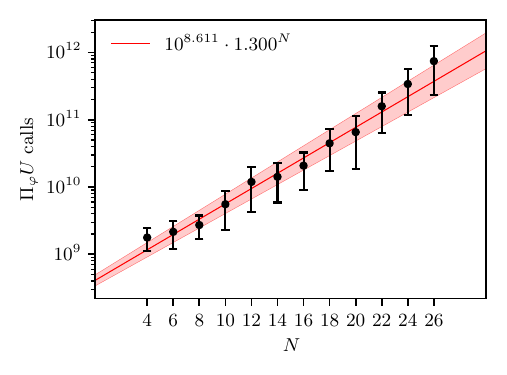}
    \caption[Estimated number of call to qubitization procedure with TQSE to achieve fractional energy error of $10^{-4}$]{Estimated number of call to qubitization procedure with TQSE to achieve fractional energy error of $10^{-4}$ calculated using data from Fig.~\ref{fig: results vs calls thqse}. Red line corresponds to exponential fit (straight line on log-log axes) with filled area indicating uncertainty in fitting and extrapolation.}
    \label{fig: results vs calls extrapolation thqse} 
\end{figure*}

Here we present the results of using thresholded QSE (TQSE) algorithm as given by Ref.~\cite{epperly2022theory} for the same systems as in Appendix~\ref{app: noise fitting}. These are shown in Fig.~\ref{fig: results vs calls thqse} as a direct comparison with Fig.~\ref{fig: results vs calls} in which partitioned QSE (PQSE)~\cite{oleary2024partitioned} was used. Noise was applied in the same way as in the main text and the threshold parameter was chosen to be equal to $||\mathbf{\Delta_S}||$ which has been found, in practice to be a good choice of thresholding parameter~\cite{epperly2022theory}. For input states used the expectation of the Hamiltonian was approximately equal to system size $N$ (as is the true ground state energy), therefore we rescaled expectation values used in $\mathbf{S}$ and $\mathbf{H}$ as $\bra{\psi_0}H^k\ket{\psi_0} \longrightarrow \bra{\psi_0}H^k\ket{\psi_0}/N^k$, which ensured elements of $\mathbf{\Delta_S}$ and $\mathbf{\Delta_H}$ were of similar absolute values (as well as the same fractions of corresponding elements in $\mathbf{S}$ and $\mathbf{H}$). 

For $N \gtrsim 24$ We can see that the TQSE performed over an order of magnitude worse than PQSE for the same number of calls to $U$. Furthermore, fitting and extrapolation to predict the number of calls to $U$ to achieve a given measurement error shows more severe exponential scaling. Comparing the fit of Fig.~\ref{fig: results vs calls extrapolation thqse} to that of Fig.~\ref{fig: noisy results vs calls extrapolation} shows that the number of calls to $U$ required to achieve fractional measurement error of $10^{-4}$ is $\sim1.300^N$ for TQSE, compared to $\sim1.143^N$ for PQSE.

\section{Preparing $\ket{G}$ and $U$ using linear combination of unitaries}\label{app: preparing block encoding}

In this section, we will recap the procedure that Kirby et al. described within Appendix C. of Ref.~\cite{kirby2023exact} to prepare $U$ and $\ket{G}$ for a local Hamiltonian. Later we will modify this procedure to reduce the resource cost for our Hamiltonian.

\subsection{Implementing $U$}\label{app: implementing U}
We need to implement $U$ given by \eqref{eqn: U}, for some choice of $\{\ket{i}\}$. Let 
\begin{equation}
    \ket{i} = \ket{\vec x, \vec z},
\end{equation}
where $(\vec x, \vec z)$ is the binary representation of $n$ qubit Pauli operator $\hat P_i$, defined as follows. For now we will assume that all $\{\alpha_i\}$ are positive in the definition of the Hamiltonian \eqref{eqn: gen hamiltonian} and that no $\hat P_i$ have a negative sign. Kirby et al. consider applying the required signs to the Pauli operator for their Hamiltonian separately, we will show how to do this for an arbitrary model later.

\begin{definition}[Binary representation of Pauli operator] \label{def: bin rep of pauli}
    Given a Pauli operator $\hat P \in \{\I,X,Y,Z\}^{\otimes n}$, the binary representation of $\hat P$ is a pair of $n$ length binary vectors $(\vec x, \vec z)$ such that a value 1 in position $j$ of $\vec x$ means that $\hat P$ contains $X$ acting on qubit $j$, and similarly for $\vec{z}$ and $Z$. This means $\hat P$ can be written as
    \begin{equation}
        \hat P(\vec x, \vec z) = \bigotimes_{j=0}^{n-1} i^{x_jz_j}X^{x_j}Z^{z_j}.
    \end{equation}
    Here, $i$ is the imaginary unit which occurs when $\vec x$ and $\vec z$ contain 1 in the same position. Both $X$ and $Z$ acting on a qubit, along with the phase, gives $Y$ operator since $ZX=iY$.
\end{definition} 
Using this representation, $U$ can be implemented as
\begin{equation}
U=\sum_{\vec x, \vec z} \ketbra{\vec x, \vec z}{\vec x, \vec z}\otimes\hat P(\vec x, \vec z)=\prod_{j=0}^{n-1}i^{x_jz_j}C_{x_j}X_j\cdot C_{z_j}Z_j,
\end{equation}
which requires three layers of two qubit gates. Firstly, a layer of controlled gates that applies $Z$ to each qubit $j$, controlled on the value of auxiliary (qubit $z_j$), a second layer doing the same for $X$, and then a layer that applies the two qubit controlled $S$ gate
\begin{equation}\label{eqn: cs}
    \begin{pmatrix}
    1&0&0&0\\
    0&1&0&0\\
    0&0&1&0\\
    0&0&0&i
    \end{pmatrix}
\end{equation}
between each pair of qubits $(x_j,z_j)$. 

\begin{definition}[Hamming weight of a Pauli operator] \label{def: hamming weight of pauli} Given a Pauli operator $\hat{P}$ with corresponding binary vectors $(\vec{x},\vec{z})$ as defined in Def.~\ref{def: bin rep of pauli}. The Hamming weight of the operator $b(\hat{P})$ is equal to the total number of non-zero (i.e. one) elements of $\vec{x}$ and $\vec{z}$.
\end{definition}

The Hamming weight of each operator within the Hamiltonian will play a significant role in determining the gate cost of the qubitisation procedure.

\subsection{Preparing $\ket{G}$}\label{app: preparing G}

We must now prepare the block encoding state
\begin{equation}
\ket{G} = \sum_{i=0}^{T-1}\sqrt{|\alpha_i|}\ket{(\vec x, \vec z)_i}_a,
\end{equation}
which requires us to prepare a potentially arbitrary real superposition of basis states. This can be achieved using multi-controlled partial CNOT ($\CnNOT{n}(\vartheta)$) gates as multi-controlled versions of the two qubit partial swap ($\PSWAP(\vartheta)$) gates. The action of the singly controlled partial NOT gate is
\begin{equation}
    \CNOT(\vartheta) = 
    \begin{pmatrix}
        1&0&0&0\\
        0&1&0&0\\
        0&0&\cos\vartheta & \sin\vartheta\\
        0&0&\sin\vartheta & -\cos\vartheta
    \end{pmatrix},
\end{equation}
and the uncontrolled action of the PSWAP gate on two qubits is
\begin{equation}
        \PSWAP(\vartheta)=
    \begin{pmatrix}
        1 &0&0&0\\
        0&\cos\vartheta & -\sin \vartheta &0\\
        0&\sin\vartheta &\cos\vartheta &0\\
        0&0&0&1
    \end{pmatrix},
\end{equation}
for some rotation angle $\vartheta$. 
We note two things, if we have a state containing basis states of Hamming weight (as defined in Def.~\ref{def: hamming weight of pauli}) less than or equal to $n$, then a $\CnNOT{n}(\vartheta)$ gate can rotate an $n$ weight basis state (with 1 corresponding to every control qubit) into an $n+1$ weight basis state (with the additional Hamming weight coming from the target qubit) and will have no effect on lower weight states. Similarly, a $\CnPSWAP{n}(\vartheta)$ will rotate between different states of Hamming weight $n+1$ (those with 1s in the controls and a 1 and a 0 in the two target qubits) and have no effect on lower weight states.

The following procedure can then be used to generate an arbitrary superposition of basis states.
\begin{enumerate}
    \item Starting with zero state, apply $X$ gate to flip a single qubit.
    \item Use $\PSWAP(\vartheta)$ gates, to distribute this single 1 between other qubits to generate all Hamming weight 1 states required.
    \item Use a single $\CNOT(\vartheta)$ gate to rotate one of the Hamming weight 1 states into a Hamming weight 2 state.
    \item Use $\CPSWAP(\vartheta)$ gates to distribute this Hamming weight 2 state into all Hamming weight 2 states required (Hamming weight 1 states will be unaffected).
    \item Repeat steps 3 and 4, increasing the number of controls and the Hamming weight by 1 each time, until all states up to the maximum Hamming weight have been produced.
\end{enumerate}
The choice of $\vartheta$ for each gate must be chosen according to the values of $\alpha_i$ required, noting that larger amplitudes than $\sqrt{\alpha_i}$ will be required from some rotations for basis states that will have $\CnNOT{n}$ gates applied to access the higher weight basis states.

\subsection{Applying signs to $\hat P_i$'s and preparing $\ket{\tilde{G}}$}\label{sec: block encoding minus signs}

To prepare the state 
\begin{equation}
\ket{\tilde{G}} = \sum_{i=0}^{T-1}\sgn(\alpha_i)\sqrt{|\alpha_i|}\ket{(\vec x, \vec z)_i},
\end{equation}
we can use the same procedure and rotation gates as for $\ket{G}$. The only change required is to flip the signs of some of the rotation angles for some of the partial SWAP and partial CNOT gates. By flipping the sign of the rotation angle $\vartheta \rightarrow -\vartheta$ for these gates, we induce a $-1$ phase shift in the resulting basis state compared to the one used in its generation. This approach ensures the incorporation of the requisite negative signs. It is crucial to monitor the sign evolution of successive bitstrings throughout this non-Markovian process. The sign of a basis state being rotated out may itself be positive or negative, contingent upon the preceding rotations and potential bitflips utilised in its generation.

\section{Gate costs for block encoding Schwinger model}\label{app: resource cost schwinger}

\subsection{Sketch of proofs for gate costs}\label{subsec: sketch proof costs}

We present some mathematical lemmas that will be useful for our costing in Appendix~\ref{subsec: useful maths}. To compute the resource costs that we desire, we begin by calculating the Pauli decomposition for each term within the Schwinger model Hamiltonian. This is done in Appendix~\ref{app: schwinger pauli decomp}. Since the gauge fields are bosonic fields these Pauli decompositions take similar forms to the operators studied in other works involving digital simulation of bosonic systems~\cite{anderson2022coarse}. For each term in the Pauli decomposition we compute the Hamming weight of the bit string used to encode the operator within the LCU qubitization procedure.

\subsubsection{Calculating the gate cost of $U$}

Calculating an upper bound for the cost of $U$ consists of two parts. The first set of two qubit gates to apply the controlled X, controlled Z and controlled S ($i$ phase) gates required to apply Pauli operators $\hat{P}_i$. For the CX and CZ gates, there is one of these for each qubit within the corresponding $\vec{x}$ and $\vec{z}$ auxiliary register. The controlled S gates are each controlled on qubit within the $\vec{x}$ register, targeting the corresponding qubit within the $z$ register. Thus the total cost of this step is just $\sim3mN$ two qubit gates. A full proof of this cost is given in Theorem~\ref{thm: cost U}

We also give, in Theorem~\ref{thm: cost U with phases}, the cost of $U$ if, as described by Kirby et al.~\cite{kirby2023exact}, we apply the phases of $\alpha_i$ as part of it, rather than by using an asymmetric block encoding and applying phases within $\tilde{G}$. Applying phases within $U$ gives a gate cost that is $\mathcal{O}(N^2\log(N))$ so we achieve a significant saving by applying phases within $\tilde{G}$.

\subsubsection{Calculating the gate cost of $G$ and $\tilde{G}$}

As described by Kirby et al.~\cite{kirby2023exact}, $\ket{G}$ can be generated by applying partial rotations between different computational basis states of the $\ket{G}$ register where the different states rotated into encode the binary encoding of each Pauli decomposition with the Hamiltonian. This cost is calculated in Appendix~\ref{app: block encoding g} by computing sums for the number of terms with a given Hamming weight, using the previously computed Pauli decompositions, and therefore the number of multicontrolled SWAP and partial CNOT gates required. Having computed a decomposition of the multi-qubit gates into T-gates, CNOTs and $R_z$ gates (Lemma~\ref{lem: CnNOT CpSWAP decomp}), we can then calculate the total number of non-Clifford (T and rotation gates) along with two qubit CNOT gates. 

We are able to reduce the cost of $G$ compared to the method described by Kirby et al.~\cite{kirby2023exact} by exploiting the symmetry within the Hamiltonian. Noticing that gauge field and interaction terms within the Hamiltonian are the same for each gauge-link site, rather than rotating between the binary states required for all lattice sites, we perform the costly rotations into basis states $\ket{\vec{x},\vec{z}}$ for just a single link site and then perform a less costly partial swap of all terms into the other sites. This method is described in Appendix~\ref{sec: copy trans inv terms} and the total gate cost for $G$, using this method is presented in Theorem~\ref{thm: cost G}.

To convert the single qubit $R_z$ gate counts into the more relevant T-gates, we make use of the approximate decomposition of $R_z$ given by Ross et al.~\cite{ross2014optimal}. This decomposition requires a number of T-gates per rotation gate that scales logarithmically with the required rotation precision (see Proposition \ref{prop: optimal sk for rz}). Assuming that required rotation error scales inversely with the average values of the amplitudes of each Pauli decomposition (and therefore the number of terms in the Hamiltonian) allows us to estimate the number of T-gates required for the rotations within $G$. For further details see Appendix~\ref{app: rz to t}.

If we use $m$ qubits for each gauge link field (such that they have maximum local occupancy of $\Lambda=2^{m-1}$, the number of gates required scales as $\mathcal{O}(m2^m+Nm)$. Recalling from Remark~\ref{rem: m log N} that the maximum gauge field value is equal to the number of lattice sites, the maximum value of $\Lambda$ required is $N$ and this gate cost is $\mathcal{O}(\log(N)N)$. Without using our method to reduce the cost through by performing this partial swap, the gate cost of $G$ would be $\mathcal{O}(\log(N)N^2)$.

We note in Corollary~\ref{cor: cost G tilde} that $\tilde{G}$ can be implemented with the same rotation gates as $G$ except flipping the sign of some of the angles. Thus, the cost for $\tilde{G}$ is the same as the cost for $G$.

\subsubsection{Total resource cost of algorithm}\label{subsec: total resource cost}

Now that we have calculated the gate costs for $G$ and $\tilde{G}$, it is easy to calculate the gate costs of corresponding rotations $\Pi_\varphi$ and $\tilde{\Pi}_\varphi$. This is given in Theorem~\ref{thm: cost R} and is twice the cost of $G$ with an additional number $\mathcal{O}(mN)$ number of CNOT and T-gates. Combining these with the cost of $U$ allows us to calculate the gate cost to perform a single measurement of $\bra{\psi_0}H^k\ket{\psi_0}$. The cost of this, as described in Theorem~\ref{thm: total scaling}, is $\mathcal{O}(k\log(N)N)$.

\subsection{Some useful mathematics}\label{subsec: useful maths}

Before we evaluate the quantum resource required for the block encoding, we prove a few useful lemmas.

\begin{lemma}\label{lem: sum binomial}
Given positive integer $b$,
\begin{equation}
    \sum_{w_y=0}^{\left\lfloor\frac{b}{2}\right\rfloor} {{b-w_y} \choose w_y} = F(b+1),
\end{equation}
where $F(n)$ is the $n$th Fibonacci number.
\end{lemma}
\begin{proof}
    See Ref.~\cite{ganis1959notes}.
\end{proof}

\begin{lemma}\label{lem: sum F}
    Defining $F(n)$ to be the $n$th Fibonacci number, then for any positive integer $m$,
    \begin{equation}
        \sum_{b=2}^{m+2}F(b+1) = \frac{1}{\sqrt{5}}\left(\phi^{m+5} - (-\phi)^{-m-5}\right) - 3.
    \end{equation}
    where $\phi$ is the golden ratio.
\end{lemma}
\begin{proof}
    From the definition of the Fibonacci numbers
    \begin{equation}
    \begin{split}
        F(m+3) & = F(m+5) - F(m+4)\\
        F(m+2) & = F(m+4) - F(m+3)\\
        & \vdots \\
        F(4) & = F(6) - F(5)\\
        F(3) & = F(5) - F(4).
    \end{split}
    \end{equation}
    Summing the L.H.S. and R.H.S. of the above yields
    \begin{equation}
        \sum_{b=2}^{m+2}F(b+1) = F(m+5) - F(4).
    \end{equation}
    Using Binet's formula and $F(4) = 3$ yields the result.
\end{proof}
\begin{lemma}\label{lem: sum bF}
    For any integer positive integer $m$,
    \begin{equation}
    \begin{split}
        \sum_{b=2}^{m+2}bF(b+1) 
        = \frac{m+2}{\sqrt{5}}\left(\phi^{m+5} - (-\phi)^{-m-5}\right) 
         - \frac{1}{\sqrt{5}}\left(\phi^{m+6} - (-\phi)^{-m-6}\right) 
         + 2.
    \end{split}
    \end{equation}
\end{lemma}
\begin{proof}
    From the definition of the Fibonacci numbers
    \begin{equation}
    \begin{split}
        (m+2)F(m+3) = & (m+2)F(m+5)  - (m+2)F(m+4)\\
        (m+1)F(m+2) = & (m+1)F(m+4) - (m+1)F(m+3)\\
        & \vdots \\
        3F(4) = & 3F(6) - 3F(5)\\
        2F(3) = & 2F(5) - 2F(4).
    \end{split}
    \end{equation}
    Summing the L.H.S. and R.H.S. of the above equations yields
    \begin{align}
    \begin{split}
    \sum_{b=2}^{m+2} bF(b+1)
    &= (m+2)F(m+5) - \sum_{b=4}^{m+3}F(b+1) - 2F(4) \\
    &= (m+2)F(m+5) - F(m+6) + F(3),   
    \end{split}
    \end{align}
    where we have rewritten the sum in the second line using Lemma~\ref{lem: sum F}.  Using Binet's formula and $F(3)=2$ yields the result.
\end{proof}

\begin{corollary}\label{cor: alpha beta sums}
$\;$
\begin{equation}
\begin{split}
        \sum_{b=2}^{m+2}F(b+1)(\alpha b + \beta) =& \frac{\alpha}{\sqrt{5}}\left(\phi^{m+5} - (-\phi)^{-m-5}\right)m - \frac{\alpha}{\sqrt{5}}\left(\phi^{m+6} - (-\phi)^{-m-6}\right) \\
        & + \frac{(2\alpha+\beta)}{\sqrt{5}}\left(\phi^{m+5} - (-\phi)^{-m-5}\right) + (2\alpha - 3\beta).
\end{split}
\end{equation}
\end{corollary}

\begin{proof}
    This follows from Lemmas~\ref{lem: sum F}~and~\ref{lem: sum bF}.
\end{proof}
\begin{proposition}\label{prop: alpha beta binom sum}
    \begin{equation}
        \sum_{b=0}^m {m \choose b}(\alpha b + \beta) = \alpha m 2^{m-1} + \beta 2^m.
    \end{equation}
\end{proposition}
\begin{proof}
This follows from standard results for the summation of binomial coefficients.
\end{proof}

\subsection{Pauli decomposition of the Hamiltonian}\label{app: schwinger pauli decomp}

In order to calculate the resources needed for the qubitization procedure, we begin by calculating the Pauli operators needed to construct the Hamiltonian.

\begin{definition}[Weight of a Pauli operator]\label{def: pauli weight}
    Given an operator $\hat P$ which can be written as single tensor product of the form $\{\I,X,Y,Z\}^{\otimes m}$, we define the weight $W(\hat P)$ of $\hat P$ as the total number of non-identity operators within $\hat P$.
\end{definition}

\begin{definition}[X,Y and Z weight of a Pauli operator]\label{def xys pauli weight}
    For an operator $\hat P$ as above, we define the X weight of an operator as the number of X operators in its Pauli decomposition. Similarly for the Y weight and Z weight. If $W(\hat P)$ is the total Pauli weight of an operator, then $W_x(\hat P)+W_y(\hat P)+W_z(\hat P)=W(\hat P)$, where $W_{x,y,z}(\hat P)$ are the X,Y and Z weights respectively.
\end{definition}

\begin{remark}\label{rem: hamming weight bin rep}
    Given an operator $\hat P$ with X,Y,Z Pauli weights $W_x(\hat P), W_y(\hat P), W_z(\hat P)$. The Hamming weight of the binary representation vector $(\vec x, \vec y)$ of $\hat P$ as defined in Def.~\ref{def: hamming weight of pauli} is equal to $b(\hat P)=W_x(\hat P)+2W_y(\hat P)+W_z(\hat P)$.
\end{remark}

To analyse the resource requirements of the qubitization procedure, we will need to decompose each term in the Hamiltonian into a basis of Pauli operators and calculate the number of constituent operators of a given weight.

\subsubsection{Gauge field terms}

\begin{lemma}\label{lem: pauli decomp field energy}
    Consider the truncated gauge-field operator $L^2 = \sum_{l=0}^{\Lambda-1} \left(l-\Lambda/2\right)^2 \ketbra{\underline{l}}{\underline{l}}$ where $\Lambda$ is a power of 2. If we write the operator in the Pauli basis $\{\I,X,Y,Z\}^{\otimes m}$, where $m=\log_2(\Lambda)$, then for every integer $0 \leq w < m$, the number of multi-qubit Pauli operators of weight $w$ is less than or equal to the binomial coefficient ${m \choose w}$.
\end{lemma}

\begin{proof}
    We note the single qubit Pauli operators
    \begin{equation}\label{eqn: single qubit ops}
    \begin{split}
    \ketbra{0}{0} & = \frac{1}{2}(\I + Z),\\
    \ketbra{1}{0} & = \frac{1}{2}(X - iY),\\
    \end{split}
    \qquad
    \begin{split}
    \ketbra{1}{1} & = \frac{1}{2}(\I - Z),\\
    \ketbra{0}{1} & = \frac{1}{2}(X + iY).\\
    \end{split}
    \end{equation}
    
    Using \eqref{eqn: single qubit ops}, we can rewrite this operator as
    \begin{equation}\label{eqn: pauli decomp field}
    \begin{split}
        &\sum_{l=0}^{\Lambda-1} \left(l-\Lambda/2\right)^2 \ketbra{\underline{l}}{\underline{l}} = \sum_{l=0}^{\Lambda-1} \left(l-\Lambda/2\right)^2 \bigotimes_{j=0}^{m-1} (\I+ (-1)^{\bin(l)_j}Z),
    \end{split}
    \end{equation}
    where $\bin(l)_j \in\{0,1\}$ denotes the $j$th bit in the binary representation of integer $l$. For a given $l$, expand the number the $m$ qubit tensor product yields all $m$ qubit operators in $\{\I, Z\}^{\otimes m}$ with coefficients $\pm 1$, the number of unique terms containing a total of $W$ Pauli $Z$ operators is given by the binomial coefficient ${m \choose w}$.
    Every value of $l$ yields operators in $\{\I, Z\}^{\otimes m}$, just with a different set of prefactors. Thus, the total number of unique Pauli operators of $W$ will be at most ${m \choose w}$. It may be the case that, depending on the values of $\Lambda$ and therefore $(l-\Lambda/2)$, some of these term cancel, however the binomial coefficient provides and upper limit on the number of terms.
\end{proof}

\begin{remark}\label{rem: field energy only z}
    We note that the only Pauli operators within the above are $Z$ operators, there are no $X$ or $Y$ operators.
\end{remark}

\subsubsection{Interaction terms}

\begin{lemma}
    Consider the truncated raising operator $e^{-i\hat\theta}=\sum_{l=0}^\Lambda\ketbra{\underline{l+1}}{\underline{l}}$ where $\Lambda$ is a power of 2 and cyclic numbering is used such that $\ket{\underline{\Lambda}} = \ket{\underline0}$. If we write the raising operator in the Pauli basis $\{\I,X,Y,Z\}^{\otimes m}$, where $m=\log_2(\Lambda)$, then for every integer $0 \leq w \leq m$, the number of multi-qubit Pauli operators of weight $w$ present is equal to $2^w$.
\end{lemma}
\begin{proof}
    As noted in Theorem 1 of Ref.~\cite{anderson2022coarse}, the Pauli decomposition of a bosonic ladder operator can be written as 
    \begin{equation}
    \begin{split}
        \ketbra{\underline{l+1}}{\underline{l}} =\bigotimes_{j=k(l)+1}^{m}\left(\I + (-1)^{\bin(l)_j}Z\right) \otimes (X-iY)\otimes(X+iY)^{\otimes k(l)},
    \end{split}
    \end{equation}
    where $\bin(l)_j \in\{0,1\}$ denotes the $j$th bit in the binary representation of integer $l$ and 
    \begin{equation}
        k(l) := \max \{j \mid \bin(l)_j=0\}
    \end{equation} 
    is the position of the right-most zero within that representation. We say $k(l)=m$, if $l$ contains no zeros (i.e., $l=\Lambda-1)$. Therefore, we can write the full raising operator as
    \begin{equation}
    \begin{split}
        \sum_{l=0}^{\Lambda-1}\ketbra{\underline{l+1}}{\underline{l}} =  \sum_{l=0}^{\Lambda-1} \bigotimes_{j=k(l)+1}^{m}\left(\I + (-1)^{\bin(l)_j}Z\right) \otimes (X+iY)\otimes(X-iY)^{\otimes k(l)},
    \end{split}
    \end{equation}
    which, by rewriting the summations, can be re-expressed as
    \begin{equation}\label{eqn: expanded raising}
    \begin{split}
        &\sum_{\{l|k(l)=k'\}}\ketbra{\underline{l+1}} {\underline{l}} = \sum_{k'=0}^m \left[ \sum_{\{l|k(n)=k'\}} \bigotimes_{j=k'+1}^{m} \left(\I + (-1)^{\bin(l)_j}Z\right)\right] \otimes (X+iY)\otimes(X-iY)^{\otimes k'}.
    \end{split}
    \end{equation}
    If we consider first the term within the square braces, we note that
    \begin{equation}\label{eqn: sum z products}
        \sum_{\{l|k(l)=k'\}} \bigotimes_{j=k'+1}^{m} \left(\I + (-1)^{\bin(l)_j}Z\right) = (2\I)^{\otimes (m-k'-1)}.
    \end{equation}
    To see this is the case, we use induction. Firstly, for $m=k'+1$, we note that there are only four terms (acting on a single qubit) within the summation, these correspond to $n=\Lambda$ and $n=\Lambda+1=0$. The summation is thus equal to $(\I+Z) + (\I-Z) = 2\I$. When $m=k'+2$, there are two terms within the summation, each acting on two qubits. These are equal to 
    \begin{equation}
    \begin{split}
    &(\I+Z)\otimes(\I+Z) +(\I+Z)\otimes(\I-Z) + (\I-Z)\otimes(\I+Z) + (\I-Z)\otimes(\I-Z)\\
    &\;=\big[(\I+Z) + (\I-Z)\big]\otimes\big[(\I+Z) + (\I-Z)\big]\\
    &\;=2\I\otimes2\I.
    \end{split}
    \end{equation}
    From this, we can see how \eqref{eqn: sum z products} holds by induction. Therefore,
    \begin{equation}\label{eqn: final gauge raising expansion}
    \begin{split}
    \sum_{\{l|k(l)=k'\}}\ketbra{\underline{l+1}}{\underline{l}} = \sum_{k'=0}^m &(2\I)^{\otimes(m-k'-1)}\otimes (X+iY)\otimes(X-iY)^{\otimes k'}.
    \end{split}
    \end{equation}
    Looking at the $(X+iY)\otimes(X-iY)^{\otimes k'}$ term, expanding these brackets will yield $2^{k'+1}$ terms in $\{X,Y\}^{\otimes k'+1}$, which all have weight $k'+1$. There will be $(m+1)$ unique values of $k'=0,1,\dots,m$, each of which contribute $2^{k'+1}$ terms of weight $k'+1$, hence the result.
\end{proof}

\begin{remark}
From the above proof, we can also see that if we consider all multi-qubit Pauli operators of a given $w$ within the above expansion, the number of operators that contain a total of $w_y$ Y operators (or equivalently X operators) is ${w \choose w_y}$. There are no $Z$ operators remaining in the expansion.
\end{remark}

\begin{corollary}\label{cor: weight interaction term}
    Consider a single interaction term (for a certain $n$) given by \eqref{eqn: schwinger interaction} using the truncated raising and lowering operators $e^{-i\hat\theta}=\sum_{l=0}^\Lambda\ketbra{\underline{l+1}}{\underline{l}}$ and $e^{i\hat\theta}=\sum_{l=0}^\Lambda\ketbra{\underline{l}}{\underline{l+1}}$. Decomposing the interaction term into Pauli operators onto a qubit register of size $(m+2)$, such that $m=\log_2(\Lambda)$ qubits are used for the gauge field space, and the remaining two qubits for the two spin degrees of freedom, then for every integer $w$ satisfying $2\leq w \leq m+2$, there will be $2^{w-1}$ Pauli operators of weight $w$. These operators will always contain an even number $w_y$ of Pauli Y operators, for a given $w$ the number of operators with $w_y$ Pauli Y operators will be ${w\choose w_y}$.
\end{corollary}
\begin{proof}
    Decomposing the spin raising and lower operators within a single interaction term, we get
    \begin{equation}\label{eqn: pauli decomp interaction}
    \begin{split}
        \sigma^+e^{i\theta}\sigma^- + \hc =& \frac{1}{4}\big[(X+iY)e^{i\theta}(X-iY) +\hc\big]\\
        =& \frac{1}{4}\big[Xe^{i\theta}X +iXe^{i\theta}Y -iYe^{i\theta}X + Ye^{i\theta}Y + \hc \big].
    \end{split}        
    \end{equation}
    We can see from \eqref{eqn: final gauge raising expansion}, that half of the terms within the Pauli expansion of $e^{i\theta}$ are Hermitian and half are anti-Hermitian. Therefore, exactly half the terms within the expansion $Xe^{i\theta}X+iXe^{i\theta}Y - iYe^{i\theta}X - Ye^{i\theta}Y$ will be Hermitian and anti-Hermitian. Thus, in adding the Hermitian conjugate to get the full interaction term in \eqref{eqn: pauli decomp interaction}, half of the terms (the anti-hermitian ones) will cancel, and the other half (the Hermitian ones) will be repeated. We note that anti-hermitian terms, and therefore those that cancel, always contain an odd number of Pauli Y operators, since in the definition of the raising operators these contain the imaginary prefactors.
    
    Counting up the number of terms in \eqref{eqn: pauli decomp interaction}, multiplying by the weights of the operator in the Pauli decomposition of $e^{i\theta}$ and taking into account the cancellation of half the terms gives the required result.
\end{proof}

\subsubsection{Spin terms}

\begin{remark}\label{rem: weight spin terms}
    It is trivial to see that each the Pauli decomposition for an onsite spin term $(-1)^n\left[\mu+\mu\sigma^3(n)\right]/2$ consists of a single $Z$ operator only.
\end{remark}

\subsection{Generating the block encoding state $\ket{G}$}\label{app: block encoding g}

The circuit cost of generating $\ket{G}$ using the method described in Appendix~\ref{app: preparing block encoding} is entirely determined by the number of operators and the weight of each operator within the Pauli decomposition of the Hamiltonian. Here we will estimate the total number of gates needed to generate $G$ for the 1+1D Schwinger Hamiltonian. We will start by calculating the number of vectors of each Hamming weight required to encode the Pauli decomposition of the Hamiltonian. We will consider gauge field terms, spin terms and interaction terms separately.

We will separate $\ket{G}$ into parts corresponding to each term in the Hamiltonian:
\begin{equation}\label{eqn: split G}
    \ket{G} = \sum_{n=1}^{N-1}\ket{G^\text{field}(n)} + \sum_{n=1}^{N}\ket{G^\text{spin}(n)} + \sum_{n=1}^{N-1}\ket{G^\text{int}(n)},
\end{equation}
where $\ket{G^\text{field}(n)}$ is an unnormalised state containing the basis states corresponding to the binary representation of the Pauli decomposition, with appropriate amplitudes, for the Hamiltonian term for the energy of gauge field $n$. The definitions of $\ket{G^\text{int}(n)}$ and $\ket{G^\text{spin}(n)}$ follow similarly.

\subsubsection{Binary representation of gauge field terms $\ket{G^\text{field}(n)}$}

Consider a single term corresponding to the gauge field energy for a given site $n$ in \eqref{eqn: schwinger ham pauli}, for a given integer $0\leq w_z\leq m$, there will be ${m \choose w_z}$ Pauli operators of Pauli weight $W(\hat P) = W_z(\hat P) = w_z$. This follows from Lemma~\ref{lem: pauli decomp field energy} and Remark~\ref{rem: field energy only z}. Thus, for all $0<w_z\leq m$, the gauge field term for a given lattice site $n$ in the Hamiltonian will contribute a total of ${m \choose w_z}$ Pauli operators of Pauli weight $w_z$, which from Remark \ref{rem: hamming weight bin rep}, will require binary representation vector of Hamming weight $b=w_z$. Therefore, we write the number of terms corresponding to Hamming weight $b$ due to the fields as
\begin{equation}\label{eqn: num b for field terms}
N^\text{field}_b=
\begin{cases}
    {m \choose b} \quad &\text{for } 0\leq b \leq m\\
    0&\text{otherwise}.
\end{cases}
\end{equation}

\subsubsection{Binary representation of spin terms $\ket{G^\text{spin}(n)}$}

Next, consider a single term corresponding to the onsite energy of a single spin $n$. The Pauli decomposition of this operator will involve a single qubit Z operator with Pauli weight $W(\hat P) = W_z(\hat P)=1$. Thus the $N$ onsite spin energy terms will lead to $N$ operators with binary representation vector of hamming weight one. Therefore, we write the number of terms corresponding to Hamming weight $b$ due to the spins as
\begin{equation}\label{eqn: num b for spin terms}
N^\text{spin}_b=
\begin{cases}
    1 \quad &\text{for } b=1\\
    0&\text{otherwise}.
\end{cases}
\end{equation}

\subsubsection{Binary representation of interaction terms $\ket{G^\text{int}(n)}$}

Finally, consider a single term corresponding to an interaction term for a given $n$. We define $m':=m+2$. From Corollary~\ref{cor: weight interaction term}, for a given $2\leq w \leq m'$, there will be ${w \choose w_y}$ bitstrings of Pauli weight $w_y\leq w$ if $w_y$ is even and zero otherwise. The Hamming weight of the associated binary representations will be $b=w+w_y$, which can therefore can take any value between $2$ and $2m'$. The number of terms with Hamming weight $b$ will be given by the following summation
\begin{equation}
\begin{split}
    N_b^\text{int} &=\sum_{\substack{2\leq w \leq m'\\ 0\leq w_y \leq m' \\ \text{s.t. } b=w+w_y \\ w_y\text{ even}}}{w \choose w_y} = \sum_{\substack{w_y=\max(0,b-m') \\ w_y\text{ even}}}^{\lfloor b/2\rfloor}{b-w_y \choose w_y},
\end{split}
\end{equation}
where we arrive at the final line by substituting $w=b-w_y$ and simplifying the ranges on the summations appropriately. We can find an upper bound for this as follows:
\begin{align}
     N_b^\text{int} &=\sum_{\substack{w_y=\max(0,b-m') \\ w_y\text{ even}}}^{\lfloor b/2\rfloor}{b-w_y \choose w_y} \leq \sum_{\substack{w_y=0 \\ w_y\text{ even}}}^{\lfloor b/2\rfloor}{b-w_y \choose w_y} \leq F(b+1),
\end{align}
where we have used the result for the summation from Lemma~\ref{lem: sum binomial}. Therefore, 
\begin{equation}\label{eqn: num interaction terms b}
     N_b^\text{int} \leq
     \begin{cases}
         F(b+1) & 2<b\leq 2m'\\
         \qquad 0 & \text{otherwise}.
     \end{cases}
\end{equation}

\subsubsection{Gate cost of associated $\CnPSWAP{n}$s and $\CnNOT{n}$s}
Now that we have the number of bitstrings for each Hamming weight, we can begin to calculate the gate cost of generating $\ket{G}=\sum_i\sqrt{\alpha_i}\ket{i}$. To do so, we first break down the cost of the multiqubit $\CnNOT{n}(\vartheta)$ and $\CnPSWAP{n}(\theta)$ in the following two remarks, before multiplying these by the total number that must be applied according to the number of binary representations required with each Hamming weight.

\begin{remark}[$n$ qubit Toffoli gate with one ancilla~\cite{he2017decompositions}]\label{rem: he decomp 1} A single $\CnNOT{n-1}$, where $n\geq 3$, can be implemented using $32n-96$ T-gates, $24n-72$ $\CNOT$ gates and one recyclable ancilla qubit. A $\CnZ{n-1}$ gate can be implemented with the same cost plus two additional Hadamard gates.
\end{remark}

\begin{remark}[$n$ qubit Toffoli gate with mulitple ancillae~\cite{he2017decompositions}]\label{rem: he decomp 2} A single $\CnNOT{n-1}$, where $n\geq 3$, can be implemented using $4n-8$ T-gates, $4n-7$ $\CNOT$ gates and $n-1$ ancilla qubits. A $\CnZ{n-1}$ gate can be implemented with the same cost plus two additional Hadamard gates.
\end{remark}

\begin{remark}[Creating multi-controlled $\PSWAP$ and partial $\CNOT$ gates]\label{rem: cpswap and cnot decomp} Following the logic of Appendix C3 of Ref.~\cite{kirby2023exact}, $\CnPSWAP{n-1}(\vartheta)$ and $\CnNOT{n}(\vartheta)$ gates can be implemented using the following gate sequences:
\begin{equation}
\begin{split}
    &\textnormal{C}^{n-1}_{\{j\}}\textnormal{PSWAP}_{k,l}(\vartheta) = e^\frac{i\vartheta \hat\mu_{k,l}}{2}\cdot\textnormal{C}^{n-1}_{\{j\}}\text{Z}_k \cdot e^\frac{-i\vartheta \hat\mu_{k,l}}{2}\cdot\textnormal{C}^{n-1}_{\{j\}}\text{Z}_k,
\end{split}
\end{equation}
where $\textnormal{C}^{n-1}_{\{j\}}$ denotes controlled operator on a set of $n-1$ qubits $\{j\}$ and the operator $\hat\mu_{k,l}$, which generates the two qubit $\PSWAP_{k,l}(\vartheta)$ gate, is defined by
\begin{equation}
    \hat{\mu}_{k,l}=
        \begin{pmatrix}
    0&0&0&0\\
    0&0&i&0\\
    0&-i&0&0\\
    0&0&0&0
    \end{pmatrix},
\end{equation}
and
\begin{equation}
\begin{split}
    \textnormal{C}^{n-1}_{\{j\}}\textnormal{CNOT}_{k,l}(\vartheta) = e^\frac{i\vartheta \hat\nu_{k,l}}{2}\cdot\textnormal{C}^{n-1}_{\{j\}}\text{Z}_k \cdot e^\frac{-i\vartheta \hat\nu_{k,l}}{2}\cdot\textnormal{C}^{n-1}_{\{j\}}\text{Z}_k\cdot\textnormal{C}^{n-1}_{\{j\}}\text{S}^\dagger_k,
\end{split}
\end{equation}
where $\nu_{k,l}$, which generates a two qubit $\CNOT(\vartheta)_{k,l}$ gate, is defined by
\begin{equation}
    \hat{\nu}_{k,l}=
        \begin{pmatrix}
    0&0&0&0\\
    0&0&0&0\\
    0&0&0&1\\
    0&0&1&0
    \end{pmatrix}.
\end{equation}
As discussed by Kirby et al.~\cite{kirby2023exact}, the application of the multi-controlled $S^\dagger$ gate can be included in the decomposition of the multi-controlled $Z$ gate so that it contributes none of it's two-qubit or non-Clifford gates.
\end{remark}

\begin{remark}\label{rem: SO(4) decomp}
    The two qubit gates $\CNOT(\vartheta)$ and $\PSWAP(\vartheta)$ can be constructed with two $\CNOT$s, six single qubit rotations (i.e. $R_z$ gates) and single qubit Clifford gates.
\end{remark}

This follows from the construction by Vatan and Williams (\cite{vatan2004optimal}, Fig. 2) of a general SO(4) gate and counting the number of non-Clifford and CNOT gates required.

\begin{lemma}[Implementing partial swap and partial CNOT gates.]\label{lem: CnNOT CpSWAP decomp}
    A $\CnPSWAP{n-1}(\vartheta)$ gate and a $\CnNOT{n}(\vartheta)$ can be implemented using $8n-16$ T-gates, $8n-14$ $\CNOT$ gates, $12$ $R_Z$ gates and single qubit Clifford gates, provided we have access to $n-1$ ancilla qubits.
\end{lemma}
\begin{proof}
    This follows from Remarks~\ref{rem: he decomp 2}, \ref{rem: cpswap and cnot decomp} and~\ref{rem: SO(4) decomp}.
\end{proof}

\subsubsection{Counting the number of $\CnPSWAP{n}$s, $\CnNOT{n}$s and $\CnZ{n}$s}\label{sec: counting CnPSWAPS etc}

Following the procedure as laid out in Sec.~\ref{app: preparing G}, we must apply a $\CnNOT{n}(\vartheta)$ once for each value of possible Hamming weight (each time we encounter Step 3) up to the maximum possible values $b_\text{max} = 2m+4$. 
We note that a single $\CnPSWAP{n}(\vartheta)$ gate can only rotate between states that have same Hamming weight and differ in the position of only a single one bit. Thus, if we just apply a single $\CnNOT{n}(\vartheta)$ each time we apply Step 3, say to access a Hamming weight $n+1$ state from a Hamming weight $n$ state, then it is likely that more than one of the 1s within this new state do not match up with the position of one of the 1s in one of the weight $n+1$ states we need to rotate into with the $\CnPSWAP{n}(\vartheta)$ gate. This would require us to apply more than one $\CnPSWAP{n}(\vartheta)$ per unique value of $(\vec{x},\vec{z})$ gate of Hamming weight $n+1$. However, if we examine the form of the Pauli decompositions for the spin (which is trivial), field~\eqref{eqn: pauli decomp field} and interaction terms~\eqref{eqn: final gauge raising expansion} for single lattice site $n$, we can exploit the structure of the binary representation strings to ensure we only need to apply at most one $\CnPSWAP{n}(\vartheta)$ for each unique value of $(\vec{x},\vec{z})$, at the expense of a small number of $\CnNOT{n}(\vartheta)$ gates.\\

Additionally, we will reduce the overall cost of the algorithm by only using the procedure of~\cite{kirby2023exact} (described in Section~\ref{app: preparing G} using $\CnNOT{n}(\vartheta)$ and $\CnPSWAP{n}(\vartheta)$) gates to generate the terms of $\ket{G}$ corresponding to gauge field and interaction energy of a single gauge link (and all fermion only terms). We will then generate the terms related to the other $N-2$ gauge sites using a procedure described in Section~\ref{sec: copy trans inv terms}.\\

For the gauge field energy of a single site, the binary representation vectors (considering only the $m$ bits corresponding to the gauge field qubits of each auxiliary register) take all possible values of the form  $\vec x = 0^m, \vec z\in\{0,1\}^m$, where $0^m$ corresponds to a zero bit repeated $m$ times. Thus, for each unique Hamming weight $0< b \leq m$, there is at least one binary representation vector present that differs from another of weight $b-1$ by just a single bit. Thus, for the gauge field terms, one must apply $\CnNOT{b-1}(\vartheta)$ gate for each value of $1\leq b \leq m$, where we have multiplied by the total number of gauge field registers, to generate the appropriate Hamming weights. This is then combined with $\left[{m \choose b} -1\right]$ $\CnPSWAP{b-2}(\vartheta)$ gates for each $1\leq b \leq m$ to access all bitstrings of that Hamming weight. Here, the binomial coefficient arises from the number of weight $b$ strings given by~\eqref{eqn: num b for field terms}. 

For the interaction terms of a single gauge link site, we can see from corollary~\ref{cor: weight interaction term} that the binary representation of the interaction terms have the following form
\begin{equation}
    \vec{x} = 1^w\cdot0^{m'-w}, \quad\vec{z} \in \pi(1^{w_y}\cdot0^{w-w_y})\cdot0^{m'-w},
\end{equation}
for $2\leq w_y \leq w \leq m$, with $w_y$ even, where for notational simplicity the first two bits correspond to the two spin sites and the remaining $m$ bits correspond to the gauge field. Here $\pi(1^{w_y}\cdot0^{w-w_y})$ denotes the set of all permutations of a bitstring containing $w_y$ ones and $w-w_y$ zeros. From this representation it is clear that for every weight $b>1$ bitstring there is at least one a) bitstring of weight $b-1$ that is the same except it contains one fewer one in the $\vec{x}$ register or b) bitstring of weight $b-2$ that is the same except it contains two fewer ones in the $\vec{z}$ register. Thus, every value of $b$ can be accessed by applying a $\CnNOT{b-1}(\vartheta)$ to one of the $b-1$ states, or a $\CnNOT{b-2}(\vartheta)$ to a weight $b-2$ state followed by a $\CnNOT{b-1}(\vartheta)$. Therefore, all interaction bitstrings for each Hamming weight $2\leq b \leq 2m'$ can be generated by applying $N-1$ total $\CnNOT{b-1}(\vartheta)$ gates and less than or equal to $F(b+1) -1$ total $\CnPSWAP{b-2}(\vartheta)$ gates. The latter upper bound comes from the remaining number of terms of weight $b$ from~\eqref{eqn: num interaction terms b}.\\

Finally, all the onsite spin terms are of weight one and can thus be generated from an initial weight one state by $N-1$ total $\PSWAP(\vartheta)$ gates. These weight one states can be used as the source for the $b=1$ states for the interaction and gauge field terms.\\

The total number of $\CnNOT{b-1}(\vartheta)$, $\CnPSWAP{b-2}(\vartheta)$ for the three parts of the Hamiltonian are summarised in Table~\ref{tab: gate counts}.

\begin{table*}
\centering
\begin{tabular} {@{}ccccc@{}}
\toprule
Hamiltonian term & Hamming weights & \# $\CnNOT{b-1}(\vartheta)$ & \# $\CnPSWAP{b-2}(\vartheta)$ \\
\hline
Gauge fields & $0 < b \leq m$ & $1$ & ${m\choose b} -1$\\
Interactions & $2\leq b \leq m+2$ & $1$ & $F(b+1) -1$\\
Spins & $b=2$ & $0$ & $(N-1)$\\
\bottomrule
\end{tabular}
\caption[Number of multi-controlled gates for each Hamming weight $b$ needed to generate state $(N-1)\ket{G^\text{field}(1)}+(N-1)\ket{G^\text{int}(1)}+\sum_{i=1}^{N}\ket{G^\text{spin}(n)}$ for each term in the Schwinger model Hamiltonian.]{Number of multi-controlled gates for each Hamming weight $b$ needed to generate state $(N-1)\ket{G^\text{field}(1)}+(N-1)\ket{G^\text{int}(1)}+\sum_{i=1}^{N}\ket{G^\text{spin}(n)}$ for each term in the Hamiltonian. See Section~\ref{sec: counting CnPSWAPS etc} for details.\\}
\label{tab: gate counts}
\begin{tabular} {@{}cccccc@{}}
\toprule
Hamiltonian term & \# $\CnNOT{}(\vartheta)$ & \# $X$ & \# $\textnormal{CSWAP}$ & \# $\CnNOT{m}$\\
\hline
Gauge fields + & \multirow{2}{*}{$N-2$} & \multirow{2}{*}{$m(N-2)$} & \multirow{2}{*}{$(m+4)(N-2)$} & \multirow{2}{*}{$N-2$}\\
Interactions&&&&\\
\bottomrule
\end{tabular}
\caption[Gate counts to partially swap gauge field and interaction terms for a gauge link site into the remaining $N-2$ sites implementing the transformation $(N-1)\ket{G^\text{field}(1)}+(N-1)\ket{G^\text{int}(1)}+\sum_{i=1}^{N}\ket{G^\text{spin}(n)}\longrightarrow \ket{G}$ for Schwinger model.]{Gate counts to partially swap gauge field and interaction terms for a gauge link site into the remaining $N-2$ sites implementing the transformation $(N-1)\ket{G^\text{field}(1)}+(N-1)\ket{G^\text{int}(1)}+\sum_{i=1}^{N}\ket{G^\text{spin}(n)}\longrightarrow \ket{G}$. See Section~\ref{sec: copy trans inv terms} for details.}
\label{tab: gate costs translation}
\end{table*} 

\subsubsection{Reducing gate cost using translational symmetry}\label{sec: copy trans inv terms}

Naively, we could generate $\ket{G}$ by applying the multi-controlled gates within Section~\ref{sec: counting CnPSWAPS etc} a total of $(N-1)$  times for the gauge field and interaction terms (once for each gauge link in the model) and $N$ times for the spin terms. This would lead to a very large number of gates controlled on lots of qubits. However, by noticing that the gauge field and interaction terms have the same form for all $n=1,2,\dots,N$, we can significantly reduce the cost in preparing $\ket{G}$ by partially swapping these terms for the $n=1$ site into all other lattice sites.

Consider following the procedure in the previous section to generate terms corresponding to just the gauge field and interaction terms for the first site. At this point the state will be proportional to $\ket{G^\text{field}(1)} + \ket{G^\text{int}(1)}$. Written in terms of the computation basis, each term will be of the form
\begin{equation}
\overunderbraces{
    &&& \br{1}{\text{field }n=1} 
    &&&& \br{1}{\text{field }n=2} 
    &&&} 
{
    &\ket{\alpha,\beta} &\otimes& \ket{\psi} &\otimes& \ket{\alpha',\beta'} &\otimes& \ket{0^m} &\otimes& \ket{0,0} &\otimes\dots
}  
{
    &\br{1}{x_0,z_0 \text{ site } n=1}
    &&&&\br{1}{x_0,z_0 \text{ site } n=2} 
    &&&& \br{1}{x_0,z_0 \text{ site } n=3}
},
\end{equation}
where $\alpha, \beta,\alpha'\beta' \in \{0,1\}$ and $\ket{\psi}$ represents one of the $m$ binary states corresponding to the gauge field register. 
If we include an ancilla qubit to which we apply a rotation gate we have
\begin{equation}
    \underbrace{\big(\cos(\vartheta)\ket{0} + \sin(\vartheta)\ket{1}\big)}_{\text{ancilla}} \otimes \ket{\alpha,\beta} \otimes \ket{\psi} \otimes \ket{\alpha',\beta'} \otimes \ket{0^m} \otimes \ket{0,0} \otimes\dots.
\end{equation}
Next, we apply $m$ CSWAP (Fredkin) gates, controlled on the ancilla qubit, to swap each qubit within the $m$ qubit $n=1$ field register, with the corresponding qubit in the $n=2$ field qubit. This yields
\begin{equation}
\begin{split}
    &\cos(\vartheta)\big(\ket{0} \otimes \ket{\alpha,\beta} \otimes \ket{\psi} \otimes \ket{\alpha',\beta'} \otimes \ket{0^m} \otimes \ket{0,0} \otimes\dots\big) \\
    & + \sin(\vartheta)\big(\ket{1} \otimes \ket{\alpha,\beta} \otimes \ket{0^m} \otimes \ket{\alpha',\beta'} \otimes \ket{\psi} \otimes \ket{0,0} \otimes\dots\big).
\end{split}
\end{equation}
We have now swapped the first two field registers, conditioned on the ancilla qubit. We remove the ancilla qubit, by applying an $X$ gate to each qubit in the $n=2$ field register, followed by an $\CnNOT{m}$ gate conditioned on the $n=2$ field register and targeting the ancilla qubit, followed by an $X$ gate to each qubit in the in the $n=2$ field register again, and finally an $X$ gate to the ancilla qubit. This yields
\begin{equation}
\begin{split}
    \ket{0}\otimes\big(& \cos(\vartheta) \ket{\alpha,\beta} \otimes \ket{\psi} \otimes \ket{\alpha',\beta'} \otimes \ket{0^m} \otimes \ket{0,0} \otimes \dots + \sin(\vartheta) \ket{\alpha,\beta} \otimes \ket{0^m} \otimes \ket{\alpha',\beta'} \otimes \ket{\psi} \otimes \ket{0,0} \otimes\dots\big)
\end{split}
\end{equation}
where we now have an arbitrary superposition (determined by $\vartheta$) over states with the gauge field terms swapped, with the auxiliary qubit unentangled and returned to its initial state. For initial states corresponding to the terms corresponding to the $n=1$ gauge field energy~\eqref{eqn: pauli decomp field}, $\alpha=\beta=\alpha'=\beta'=0$ (as there are no spin terms) and so we have effectively partially swapped the required binary encoding of the gauge field terms from $n=1$ to $n=2$. We can ensure the interaction terms (in which some of $\alpha,\beta,\alpha'\beta'$ are non zero) are correctly swapped as well by doing the following. Apply two Fredkin gates, controlled on the least significant $x$ qubit in the $n=2$ field register, swapping the registers corresponding to the $n=2$ spin register with the $n=3$ spin register. From~\eqref{eqn: expanded raising}, we see that interaction terms will always contain a one within this control qubit, thus for $\ket{\psi}$ corresponding to an interaction term, we get
\begin{equation}
\begin{split}
    &\cos(\vartheta) \ket{\alpha,\beta} \otimes \ket{\psi} \otimes \ket{\alpha',\beta'} \otimes \ket{0^m} \otimes \ket{0,0} \otimes \dots + \sin(\vartheta) \ket{\alpha,\beta} \otimes \ket{0^m} \otimes \ket{0,0} \otimes \ket{\psi} \otimes \ket{\alpha',\beta'} \otimes\dots,
\end{split}
\end{equation}
where we omit the ancilla qubit which is no longer used. Finally, applying another two Fredkin gates, controlled in the same way to swap the $n=1$ and $n=2$ spin registers, we get:
\begin{equation}
\begin{split}
    &\cos(\vartheta) \ket{\alpha,\beta} \otimes \ket{\psi} \otimes \ket{\alpha',\beta'} \otimes \ket{0^m} \otimes \ket{0,0} \otimes \dots + \sin(\vartheta) \ket{0,0} \otimes \ket{0^m} \otimes \ket{\alpha,\beta} \otimes \ket{\psi} \otimes \ket{\alpha',\beta'}\otimes\dots.
\end{split}
\end{equation}
This is the form of the output for both interaction and gauge energy terms (for the latter, since all the spin bits are zero, these Fredkin gates have no effect). Thus, we have taken the bitstrings corresponding to the interaction and gauge field terms for $n=1$ and used them to generate the $n=2$ gauge field and interaction terms, avoiding the need to reapply the costly multicontrolled procedure described in Section~\ref{sec: counting CnPSWAPS etc}. i.e., following the notation in \eqref{eqn: split G}, we have performed

\begin{equation}
\begin{split}
    &\frac{1}{\mathcal{N}}\left(\ket{G^\text{field}(1)}+\ket{G^\text{int}(1)}\right)\longrightarrow \frac{1}{\mathcal{N}}\left(\cos{\vartheta}\ket{G^\text{field}(1)}+\ket{G^\text{int}(1)}\right)+\frac{1}{\mathcal{N}}\left(\sin(\vartheta)\ket{G^\text{field}(2)}+\ket{G^\text{int}(2)}\right),
\end{split}
\end{equation}
where $\mathcal{N}$ is the norm of unnormalised state $\ket{G^\text{field}(1)}+\ket{G^\text{int}(1)}$. We can keep repeating this, noticing that the spin only terms are unaffected, and choosing $\vartheta$ each time such that we end up with the transformation 
\begin{equation}
\begin{split}
    &(N-1)\left(\ket{G^\text{field}(1)} + \ket{G^\text{field}(1)}\right) + \sum_{n=1}^N\ket{G^\text{spin}(n)} \\
    &\quad\longrightarrow \sum_{n=1}^{N-1}\ket{G^\text{field}(n)}+ \sum_{n=1}^{N-1}\ket{G^\text{int}(n)} + \sum_{n=1}^N\ket{G^\text{spin}(n)} = \ket{G}.
\end{split}
\end{equation}

The total gate cost for this is $(N-2)$ $R_x(\vartheta)$ gates, $2m(N-2)$ $X$ gates, $(m+4)(N-2)$ CSWAP gates and $(N-2)$ $\CnNOT{m}$ gates along with an ancilla qubit that can be reused. 

\subsubsection{Total gate cost for $\ket{G}$}
\begin{theorem}[Improved gate cost for $G$]\label{thm: cost G}
For the single flavour 1+1D Schwinger model with $N$ lattice sites and gauge fields represented by $m\geq2$ qubits, the block encoding state $\ket{G}$ can be prepared using the following number of gates. 

\begin{alignat}{2}
    &\# \text{T gates} &&  \leq \; 56 -16m
     + 8m\fib{m+5}{-m-5} -8\fib{m+6}{-m-6} \nonumber\\
    & && \phantom{\leq} -16\cdot2^m+ 4N + 4m2^m + 8Nm\\
    &\# \text{CNOT gates} &&  \leq \; 54 -22m
     -2\fib{m+5}{-m-5} -8\fib{m+6}{-m-6} + 8m\fib{m+5}{-m-5} \nonumber\\
    & && \phantom{\leq} -14\cdot2^m+ 25N + 4m2^m + 11Nm\\
    &\# R_z \text{ gates} &&  \leq \; -48
     +12\fib{m+5}{-m-5} +12\cdot2^m+ 6N
\end{alignat}

where $\phi$ is the golden ratio. In addition, we require an unspecified number of single qubit Clifford gates as well as $2N + 2m(N-1)$ qubits to store $\ket{G}$ and $m$ ancilla qubits for the computation.
\end{theorem}
\begin{proof}
The full procedure to generate $\ket{G}$ consists of two parts. First, generate the state
\begin{equation}\label{eqn: g step 1 gates}
    (N-1)\left(\ket{G^\text{field}(1)} + \ket{G^\text{field}(1)}\right) + \sum_{n=1}^N\ket{G^\text{spin}(n)},
\end{equation}
by applying the gates listed in Table~\ref{tab: gate counts} as described in Section~\ref{sec: counting CnPSWAPS etc} once for the gauge field and interaction terms, and $N-1$ times for the spin terms. 

The total number of T-gate, CNOT gates and $R_z$ gates for gauge field and interaction terms for this procedure can be written as
\begin{equation}\label{eqn: total gate cost for gen op}
\begin{split}
    &\sum_{b=1}^m{m\choose b}(\alpha b +\beta) + \sum_{b=2}^{m+2}F(b+1)(\alpha b + \beta),
\end{split}
\end{equation}
where for $\alpha=8, \beta=-16$ for T-gates, $\alpha=8, \beta=-14$ for $\CNOT$ gates and $\alpha=0, \beta=12$ for $R_z$ gates according to the decomposition of $\CnNOT{b-1}(\vartheta)$ and $\CnPSWAP{b-2}(\vartheta)$ gates of Proposition~\ref{lem: CnNOT CpSWAP decomp}. In addition, $2(N-1)$ $\CNOT$s and $6(N-1)$ $R_z$ gates are required for the spin terms.\\

Secondly, we must apply the gates listed in Table~\ref{tab: gate costs translation} as described in Section~\ref{sec: copy trans inv terms} to implement the transformation into $\ket{G}$. \\

The cost of a single $\CNOT(\vartheta)$ gate will be six $R_z$ gates and two $\CNOT$s. A single Fredkin gates ($\text{CSWAP}$) can be implemented with two $\CNOT$s along with a Toffoli ($\text{CCNOT}$) gate~\cite{smolin1996five} for a total cost of four T gates and five $\CNOT$s (using remark~\ref{rem: he decomp 2} for the Toffoli decomposition). Each of the $\CnNOT{m}$ gates requires $4m-4$ and T gates $4m-3$ $\CNOT$s from Remark~\ref{rem: he decomp 2}. Thus the total cost for this step is
\begin{align}
    &8Nm+4N-16m-8 \quad \text{T gates},\label{eqn: g step 2 t gates}\\
    &11Nm + 25N -22m-4 \quad \CNOT \text{ gates}\label{eqn: g step 2 cnots}\\
    &6N-12 \quad R_z \text{ gates}\label{eqn: g step 2 rzs}.
\end{align}
The required expressions result from using Corollary \ref{cor: alpha beta sums} and Proposition \ref{prop: alpha beta binom sum} to upper bound \eqref{eqn: g step 1 gates} with the appropriate values of $\alpha, \beta$ and adding the result to the relevant equation of  \eqref{eqn: g step 2 t gates}, \eqref{eqn: g step 2 cnots} and \eqref{eqn: g step 2 rzs}.
\end{proof}
\begin{remark}
If we choose $m=\left\lceil\log(N)+1\right\rceil$ to correspond to the maximum gauge field occupancy of Remark~\ref{rem: m log N}, then the gate cost of $\ket{G}$ is $\mathcal{O}(N\log N)$ and $(2N-1)\log N - \log(N)$ qubits are required.
\end{remark}

\subsection{T gate cost of rotation gates}\label{app: rz to t}

For error corrected quantum computing, we typically do not have access to to an arbitrary gate set but rather a restircted subset from which arbitrary operations must be performed approximately. according to the Solovay Kitaev theorem however, the gate cost of doing so scales favourably with approximation error. In our case, we need to convert the arbitrary $R_z$ rotations into Clifford + T-gates. To do so, we use the following Lemma.
\begin{proposition}[Clifford+T approximations of z-rotations~\cite{ross2014optimal}]\label{prop: optimal sk for rz}
    There exists an efficient algorithm, which is known, for approximately decomposing an $R_z$ gate up to rotation error $\epsilon$ which requires $3\log{\epsilon^{-1}} + \mathcal{O}(\log\log\epsilon^{-1})$ T gates.
\end{proposition}

We use this to make the following approximation for the number of T gates required to implement a single $R_z$ gate in our algorithm.
\begin{theorem}
    Assume the following
    \begin{enumerate}
        \item The average fractional error required on values of $\alpha_i$ is given by a constant value $\bar\epsilon_\alpha$.
        \item The average value of $\sqrt{|\alpha_i|}$ is approximately equal to the reciprocal of the number of Pauli terms within the Hamiltonian. 
        \item Recalling that the values of $\sqrt{|\alpha_i|}$ are generated by successive rotations using partial swap gates, for an average approximation error $\epsilon$ per $R_z$ gate, the average error on resulting value of given $\sqrt{|\alpha_i|}$ is less than approximately $d\epsilon$, where $d$ is the maximum number of $R_z$ gates affecting the value of $\sqrt{|\alpha_i|}$.
    \end{enumerate}
    Then, the average number of T gates required per $R_z$ is approximately less than or equal to
    \begin{equation}
    \begin{split}
        3\log\left((N-1)2^m + N\cdot2^{m+1} + N\right) + 3\log(12m+48) + 3\log({\bar\epsilon_\alpha}^{-1}),
    \end{split}
    \end{equation}
    up to $\mathcal{O}(\log\log)$ terms.
\end{theorem}
\begin{proof}

Given that the average value of $\sqrt{|\alpha_i|}$ equal to the reciprocal of the number of Pauli terms within the Hamiltonian. The absolute error required is approximated by
\begin{equation}
    \epsilon_\text{tot} = \bar\epsilon_\alpha\left[(N-1)2^m+(2^{m-1})(N-1)+N\right]^{-1},
\end{equation}
where we have used the total number of Pauli strings resulting from Corollary~\ref{cor: weight interaction term} and Remarks~\ref{rem: weight field terms},~\ref{rem: weight spin terms}.

From Lemma~\ref{lem: CnNOT CpSWAP decomp}, each multi-qubit rotation gate for the generation of $\ket{G}$ requires 12$R_z$ gates. The maximum number of rotations affecting a single value of $\alpha_i$ is equal to the Hamming weight of the corresponding basis state $\ket{(\vec{x},\vec{z})_i}$. The maximum Hamming weight present in the Hamiltonian is equal to $2m+4$ from the interaction terms of \eqref{eqn: pauli decomp interaction}. Thus, to achieve average absolute error on $\sqrt{\alpha_i}$ that is less than or equal to $\epsilon_\text{tot}$, we require $\epsilon_{R_Z}\approx(12m-48)^{-1}\epsilon_\text{tot}$.

From the the approximate values of $\epsilon_\text{tot}$ and $\epsilon_{R_Z}$ above and using Proposition~\ref{prop: optimal sk for rz} for the optimal decomposition of $R_Z$ gates we get the result.
\end{proof}

\begin{corollary}
    Assuming $\bar\epsilon_\alpha$ is constant, if $m=\lceil\log(N)\rceil$ then this is $\mathcal{O}(\log(N))$ and the total T-gate cost for generating $\ket{G}$, including rotations, is $\mathcal{O}(N\log(N))$. This does not affect the asymptotic scaling of this part of the algorithm.
\end{corollary}

For gate cost calculations presented in the main text, we assume that $\bar{\epsilon}_\alpha$ terms are neglible.

\subsection{Applying the block encoding unitary $U$}
\subsubsection{Total gate cost for $U$}

\begin{theorem}\label{thm: cost U}
    For the single flavour 1+1D Schwinger model with $N$ lattice sites and gauge fields represented by $m\geq2$ qubits, the block encoding state $\ket{G}$ can be prepared using the following number of gates. 
\begin{alignat}{2}
    &\# \text{T gates} &&  = 6N + 6mN - 6m\\
    &\# \text{CNOT gates} &&  = 8N + 8mN - 8m
\end{alignat}
We require an unspecified number of single qubit Clifford gates and $m$ ancilla qubits for the computation.
\end{theorem}
\begin{proof}
As described in Ref.~\cite{kirby2023exact}, to implement the Pauli operators in $U$, we require a total of $N_G$, the number of qubits in $\ket{G}$, two qubit $\CNOT$, $\CnZ{}$ and $\text{Controlled }S$ (see~\eqref{eqn: cs}) gates. The decompositions of the first two are trivial, for the latter we note that a controlled $S$ gate performing~\eqref{eqn: cs} can be implemented as:
\begin{equation}
\begin{split}
\Qcircuit @C=0.8em @R=0.8em {
& \gate{T} & \ctrl{1} & \qw & \ctrl{1} & \qw & \qw \\
& \qw & \targ & \gate{T^\dagger} & \targ & \gate{T} & \qw
}
\\
\end{split}
\end{equation}

Therefore, given that $N_G=2N+2m(N-1)$, the number of gates required to apply the controlled Pauli operations of $U$ is
\begin{align}
    & 3\cdot(2N+2m(N-1))\quad \text{T gates},\\
    & 4\cdot(2N+2m(N-1)) \quad \CNOT\text{ gates},
\end{align}
and zero $R_z$ gates.
\end{proof}

\begin{remark}
If $m=\left\lceil\log(N)+1\right\rceil$ as in Remark~\ref{rem: m log N}, then the gate cost of $U$ is $\mathcal{O}(N\log N)$.
\end{remark}

\subsubsection{Applying phases}\label{subsec: thm: cost of U with phases}

If, rather than introducing the signs of $\alpha_i$ within $\ket{\tilde{G}}$, we follow a similar method as Appendix C of Ref.~\cite{kirby2023exact} and add the signs within $U$, we require the operation

\begin{equation}
    U = \sum_i \sgn(\alpha_i)\ketbra{i}{i}\otimes\hat{P}_i.
\end{equation}

To understand how this would be implemented, we note that we can apply a $-1$ to an auxiliary state of Hamming weight $n$ by applying a $\CnZ{n-1}$ gate to all auxiliary qubits containing 1 in the basis state of interest. This will also flip the sign of higher Hamming weight states that contain 1s in the same positions as the Hamming weight $n$ state. Thus, the desired phases can be achieved by applying $Z$ gates to qubits in Hamming weight one states that require a $-1$ phase, then applying $\textnormal{C}Z$ gates to  Hamming weight two states that need their phases changed, flipping back the phase of any that may have been unnecessarily changed by the application of the previous single qubit phase gates, then doing the same for weight three states with $\text{CC}Z$ gates and so on until all states have the required phase applied to them. The overall gate cost is $\mathcal{O}(\log(N)N^2)$ as given in the following theorem.

\begin{theorem}[Applying phases within $U$]\label{thm: cost U with phases}
If, we add the required phases for the Block encoding of the single flavour 1+1D Schwinger model using a method analagous to that described in Appendix C of Ref.~\cite{kirby2023exact}, the gate cost of $U$ is

\begin{alignat}{2}
    &\# \text{T gates} &&  \leq \; (N-1)\bigg[32 +6 m - 4\fib{m+6}{-m-6} + 4m\fib{m+5}{-m-5} -8\cdot2^m +2m2^m\bigg] + 6N\\
    &\# \text{CNOT gates} &&  \leq \; (N-1) \bigg[29 + 8m - 4\fib{m+6}{-m-6} +4m\fib{m+5}{-m-5} - 7\cdot2^m + 2m2^m\bigg] + 8N.
\end{alignat}

\end{theorem}
\begin{proof}
    In addition to the controlled Pauli and phase gates of Theorem~\ref{thm: cost U}, we must generate the required minus signs by applying multicontrolled Z gates. In the worst case, every single bitstring of weight $b$ requires a different phase from the total phase of the $b-1$ weight strings differing by just a single bit. Thus, each one must have a $\CnZ{b}$ gate applied to achieve the correct phase. Thus, we require at most $(N-1){m \choose b}$ total $\CnZ{b-1}$ gates for each $1\leq b \leq m$ for gauge field terms and $(N-1)F(b+1)$ total $\CnZ{b-1}$ for each $2\leq b \leq m+2$ gates for interaction terms. These values correspond to the $N_b^\text{field}$ and $N_b^\text{int}$ in~\eqref{eqn: num b for field terms},~\eqref{eqn: num interaction terms b}.

    As in Theorem~\ref{thm: cost G}, the gate cost of these terms is given by \eqref{eqn: total gate cost for gen op}, this time multiplying by $(N-1)$, where for $\CnZ{b-1}$ gates $\alpha=4,\beta=-8$ for T gates and $\alpha=4,\beta=-7$ for CNOT gates. No rotation gates are needed.  The total cost follows from Lemma \ref{cor: alpha beta sums} and Proposition \ref{prop: alpha beta binom sum} \eqref{eqn: g step 1 gates}.
\end{proof}

\subsection{Total gate cost for $\Pi_\varphi$ and $\tilde{\Pi}_\varphi$}

\begin{theorem}\label{thm: cost R}
        For the single flavour 1+1D Schwinger model with $N$ lattice sites and gauge fields represented by $m\geq2$ qubits, the rotation operator $\Pi_\varphi$ for the QSVT procedure of Theorem \ref{thrm: qsvt} can be implemented with two times the cost of $G$ from Theorem~\ref{thm: cost G} plus
        $128mN + 128N - 128m - 192$ T gates, $96mN + 96N - 96m - 144$ $\CNOT$ gates and an $R_z$ gate.
\end{theorem}

\begin{proof}
As described in Remark~\ref{rem: pi g}, the rotation operator can be implemented with two $G$ operators, two $\CnNOT{N_G}$ gates and an $R_z$ gate. From Remark~\ref{rem: he decomp 1} the cost of a $\CnNOT{N_G}$ gate is 
$32N_G-96$ T gates and $24N_G-72$ $\CNOT$ gates along with one ancilla qubit which can be reused from the $G$ computation. Using $N_G = 2m(N-1) + N$ and combining the costs of $G$s, the $\CnNOT{N_G}$s and rotation yields the result.
\end{proof}

\begin{remark}
If $m=\left\lceil\log(N)+1\right\rceil$ as in Remark~\ref{rem: m log N}, then the gate cost of $\Pi_{\varphi_j}$ is $\mathcal{O}(N\log N)$. It may be possible to reduce the multicontrolled gates by using additional ancilla~\cite{baker2019decomposing} or by reducing the rotation to the subspace of the maximum achievable Hamming weight as in~\cite{kirby2023exact}. Here we sketch the steps taken to compute the quantum resource cost.
\end{remark}

\begin{corollary}\label{cor: cost G tilde}
    Since $\tilde{G}$ can be generated from the same gates as $G$ with some rotation angles differing in the phase, the cost of  $\tilde{\Pi_\varphi}$ is the same as the cost of $\Pi$ given in Theorem~\ref{thm: cost R}.
\end{corollary}

\subsection{Total gate cost of the full algorithm}

\begin{theorem}\label{thm: total scaling}
    A single measurement needed for the estimation of matrix element $\langle H\rangle$ can be calculated by using two applications of the gates listed in Theorem~\ref{thm: cost G} and $k$ applications of the gates in Theorems~\ref{thm: cost U} and~\ref{thm: cost R}. Then assuming $m=\log{N}$, the gate cost of this is $\mathcal{O}(kN\log(N))$.
\end{theorem}
\begin{proof}
    This follows from Lemma~\ref{cor: qsvt for poly} and the preceding gates costs for $U$, $\ket{G}$ and $\ket{\tilde{G}}$.
\end{proof}

\section{Numerical experiments in the strong coupling regime} \label{app: strong coupling}

In the majority of the main text, we studied a system with Hamiltonian parameters bare mass $\mu=1.5$ and coupling parameter $x=0.5$. The stronger coupling regime, in which the ground start is further from the N\'eel state, is more challenging for classical methods. In the following, we present additional numerical experiments, repeating the procedure described in Sec.~\ref{sec: total resources} for a system with $\mu=1.5$ and $x=5$ using PQSE.

In Fig.~\ref{fig: results vs order fit strong coupling}, we present the fractional energy error for this system, analagous to Fig.~\ref{fig: results vs order fit} as a function of maximum Krylov order $D$ assuming no measurement noise. This is analogous to Fig.~\ref{fig: results vs order fit} with $x=5$, rather than $x=0.5$. We see that the overlap between the initial state and the true ground state, decreasing more rapidly with $N$ than for the more weakly coupled system, as a result, a larger Krylov basis dimension $D$ is required to reach the same fractional energy error. In Fig.~\ref{fig: noiseless results vs calls extrapolation strong coupling}, we show the required Krylov basis size $D(N)$, to achieve a fixed fractional energy error of $10^{-4}$ using these results, again assuming no measurement noise. This is analogous to Fig.~\ref{fig: noiseless results vs calls extrapolation}.

In Fig.~\ref{fig: results vs calls strong coupling}, we present the fractional energy error for this system, analogous to Fig.~\ref{fig: results vs calls} in the main text. In Fig.~\ref{fig: noisy results vs calls extrapolation strong coupling}, we present the expected total number of calls required to the block encoding operator $\Pi_\varphi U$, analogous to Fig.~\ref{fig: noisy results vs calls extrapolation}. Finally, in Fig.~\ref{fig: total t gates strong coupling}, we present the total number of T-gates required to achieve fractional ground state energy errors of $1\times10^{-2}$, $1\times10^-{4}$ and $1\times10^{-6}$ arising from the estimated number of block encoding calls and the estimated gate cost for the block encoding procedure, analogous to Fig.~\ref{fig: total t gates}.

From these results, we see that the number of calls to the block encoding procedure, and therefore the number of gates required to achieve the same error for the same system size, is considerably larger than for the $\mu=1.5$, $x=0.5$ system.

\begin{figure*}
    \centering
    \includegraphics[width=\linewidth]{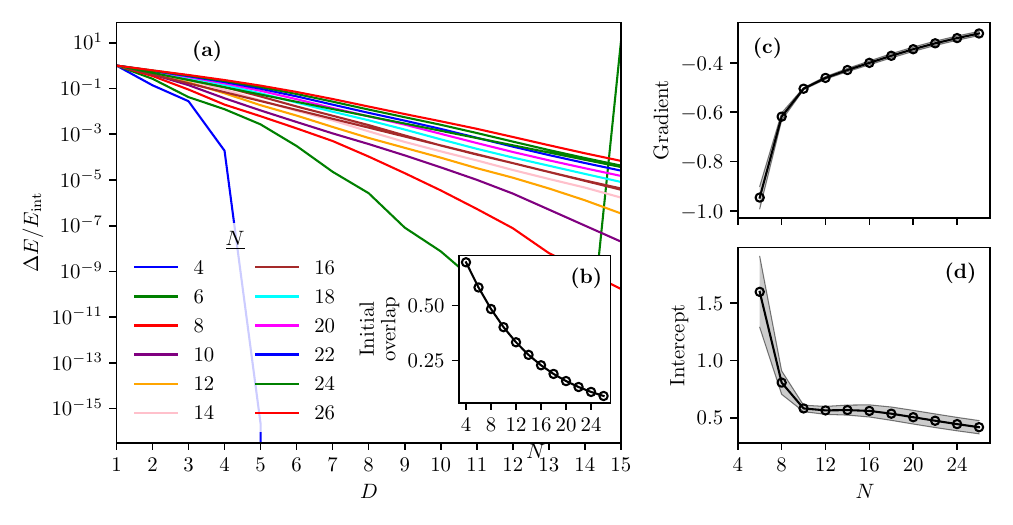}
    \caption[Fractional energy error as function of Krylov order and fitting parameters as function of lattice size in the case of zero measurement error]{\textbf{(a)} Fractional energy error as a function of maximum Krylov order for different value of lattice size $N$ for $\mu=1.5$, $x=5$. No measurement noise has been applied. \textbf{(b)} Overlap between initial state $\ket{\psi_0}$ and exact ground state. \textbf{(c)} \& \textbf{(d)} Gradient and intercept for straight lines (on log-linear plot) fit to data in (a). Filled area indicates standard errors on estimates.}
    \label{fig: results vs order fit strong coupling}
\end{figure*}

\begin{figure}
    \centering
    \includegraphics[width=0.5\linewidth]{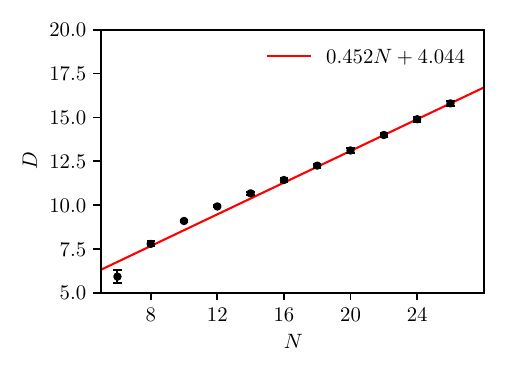}
    \caption[Estimated Krylov order $D$ needed for QSE to achieve fractional energy error of $\Delta E/E_\text{int} = 10^{-4}$ in the case of no measurement noise]{Estimated Krylov order $D$ needed for QSE to achieve fractional energy error of $\Delta E/E_\text{int} = 10^{-4}$ in the case of no measurement noise calculated using data from Appendix~\ref{app: noiseless fitting} for a system with $\mu=1.5$, $x=5$. Error bars correspond to standard errors on the fits to generate these points. Red line corresponds to linear fit to final two data points.}
    \label{fig: noiseless results vs calls extrapolation strong coupling}
\end{figure}

\begin{figure*}
    \centering
    \includegraphics[width=\linewidth]{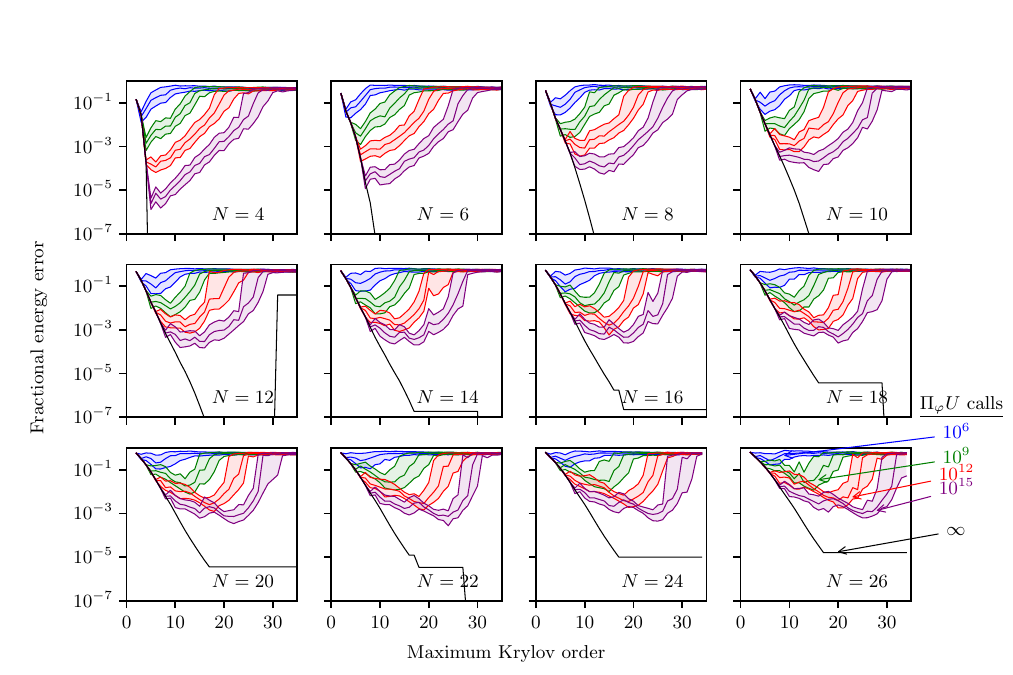}
    \caption[Fractional energy error with PQSE procedure with finite number of measurements for range of different lattice sizes.]{Fractional energy error with PQSE procedure with finite number of measurements for range of different lattice sizes $N$ for $\mu=1.5$, $x=5$. Values are shown as a function of the maximum allowed Krylov basis for the PQSE procedure and as a function of number of calls to the block encoding operator $\Pi_\varphi U$ (which affect the number of measurements as defined in the main text). Lines and shaded area corresponds to the median as well as upper and lower quartiles of energy error. The black line (labelled as infinite calls to $\Pi_\varphi U$) corresponds to no measurement noise.}
    \label{fig: results vs calls strong coupling}
\end{figure*}

\begin{figure}
    \centering
    \includegraphics[width=0.5\linewidth]{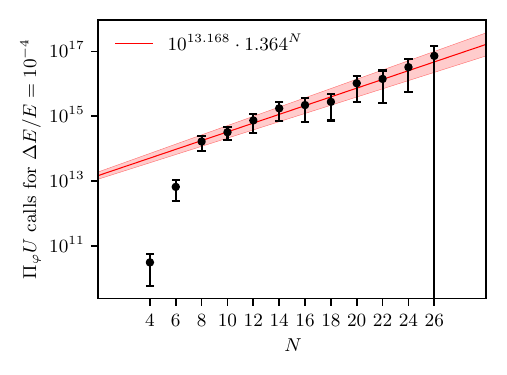}
    \caption[Estimated number of call to qubitization procedure using PQSE to achieve fractional energy error of $10^{-4}$]{Estimated number of call to qubitization procedure using PQSE to achieve fractional energy error of $\Delta E/E_\textrm{int} = 10^{-4}$ for $\mu=1.5$, $x=5$. Error bars correspond to the standard error resulting from fits used to generate the data points. Red line corresponds to exponential fit (straight line on log-linear axes) with filled area indicating uncertainty in fitting and extrapolation.}
    \label{fig: noisy results vs calls extrapolation strong coupling}
\end{figure}

\begin{figure}
    \centering
    \includegraphics[width=0.5\linewidth]{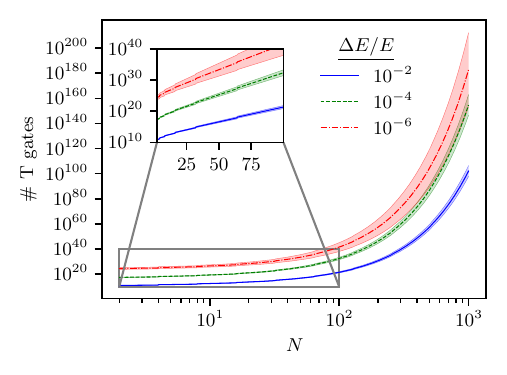}
    \caption{Total number of T-gates required for the whole PQSE procedure for the single flavour Schwinger model with $\mu=1.5$, $x=5$ as a function of lattice size $N$ for different fractional energy errors $\Delta E/E_\textrm{int}$.}
    \label{fig: total t gates strong coupling}
\end{figure}

\section{Comparing resource requirements for all-to-all to linear-nearest-neighbor qubit connectivity}

Throughout the rest of this article, we have assumed an all-to-all qubit connectivity. In this section, we briefly discuss how a limited qubit connectivity, namely linear-nearest-neighbour (LNN) will affect the gate costs associated with the algorithm. For this we consider the following Lemma which, to the author's knowledge, is the tightest upper bound for the resource cost of multi-controlled CNOT gates for a linear nearest-neighbour topology.

\begin{lemma}[$n$ qubit Toffoli gate with linear-nearest-neighbour qubit connectivity, Theorem 8 Ref.~\cite{zindorf2024efficient}]\label{lem: nearest neighbour cnot cost}
    Consider a $k$ qubit subregister with linear-nearest-neighbour qubit connectivity. A $\CnNOT{n-1}$ gate on these qubits, where $n>6$ and $k>n+1$ can be implemented with $16n-32$ T-gates, $8k+14n-44$ CNOT gates as well as single qubit Clifford gates.
\end{lemma}\label{lem: zindorf decomp}

To the authors knowledge, this is the most optimal decomposition of a multi-qubit Toffoli gate considering a restricted qubit connectivity. We recall that the best known decomposition~\cite{he2017decompositions} when allowing for all-to-all connectivity and multiple ancillas was $4n-8$ T-gates and $4n-7$ CNOT gates (Remark~\ref{rem: he decomp 2}). 
To leading order, the number of T-gates required for a multi-controlled Toffoli gate for LNN is four times larger than that required for all-to-all connectivity. Thus, the T-gate cost of the algorithm will be at most four times larger for a LNN topology. The exact value may be considerably less as many T-gates within the computation are for other non-Clifford gates (for example the single qubit rotation gates) and will depend on what ordering is chosen fot the physical qubits (which will affect which gates need to be implemented according to the decompositions in Refs.~\cite{he2017decompositions} or~\cite{zindorf2024efficient}.
Since the CNOT cost in Lemma~\ref{lem: zindorf decomp} depends on the number of qubits the multi-qubit gate acts over, computing how a LNN qubit topology affects the resource cost will be more complex. However, since the CNOT cost for all-to-all qubit connectivity at system sizes of interest is already well beyond what is achievable non fault-tolerant devices, we do not consider how this cost will increase for a ressticted qubit topology.

\section{Hardware information for current generation quantum processors}\label{app: hardware specs}

Below we give details for current generation quantum processors used for analysis in Sec.~\ref{sec: numerical experiments}. Gate times and coherence times for four processors, where they could be found or calculated, are given in Table~\ref{tab: hardware}.

For data on the Quantinuum H2-1 processor, two-qubit gate time was estimated in the following way. We take benchmark data for shot times and number of operations for a range of circuits from Table I of Ref.~\cite{moses2023race}. We take the shot time for the circuit, subtract the quoted initial state preparation time (32$\mu$s) and divide this by the number of two-qubit gates in the circuit (assuming two-qubit gates dominate the execution time over single qubit and SPAM gates). We take the mean of these values to estimate an average two qubit gate time, excluding the `Transport 1Q RB, $l = 64$' circuit which contains no two-qubit gates. Coherence time is calculated by taking the quoted quadratic dephasing rate ($0.043\cdot2\pi$rad/s) from Ref.~\cite{h2datasheet}, which is quoted as being applied during qubit transport and idling. We inverted this and multiplied by the average fraction of time spent in transport from Table I of~\cite{moses2023race} to give an estimate of the typical dephasing time.

For data on the IonQ Forte processor, two-qubit gate times were estimated in the following way. Using runtime benchmark given in the supplemental data for Ref.~\cite{chen2023benchmarking}, for each type of algorithm listed (each with varying problem sizes), the total single circuit time was divided by the number of native two-qubit gates (assuming that two-qubit gates dominated execution time). The mean two-qubit gate time was then calculated over the algorithm types. Estimated coherence times are taken from Ref.~\cite{ionqforte}.

Using the total CNOT counts given in Fig.~\ref{fig: k=1 gates}(d), we get a rough estimate of the time taken to apply a single block encoding step $\Pi_{\varphi_j}U$ by multiplying the CNOT count by the two-qubit gate time for each processor. The values calculated are shown in Table~\ref{tab: hardware runtimes} both as an absolute time, as well as a multiple of the for T1 and T2 times. This makes three assumptions, (1) when compiled to the particular hardware gate set, the number of native two-qubit gates is roughly equal to the number of CNOTs we have counted, (2) total runtime is dominated by two-qubit gate operations, and (3) CNOT gates are implemented in serial. In reality, many of these CNOTs could be run in parallel or even removed by recompilation; so this is likely to be an overestimate. We can instead assume that all gates are maximally parallelised such that two-qubit gates are applied in layers of $N_\text{qubits}/2$ acting at the same time to give a lower bound on the run time. These values are shown in Table~\ref{tab: hardware runtimes 2}.
\begin{sidewaystable*}
\centering
\begin{tabular} {@{}cccccc@{}}
\toprule
    Processor & Type &\# Qubits & T1 & T2 & two qubit gate time \\
    \hline
    IBM Eagle r3~\cite{ibmeagle} & Superconducting & 127 & 275 $\mu$s & 117 $\mu$s & 636 ns \\
    Google Sycamore~\cite{morvan2024phase} & Superconducting & 70 & 22.9 $\mu$s & 15.5 $\mu$s& 20 ns \\
    Quantinuum H2-1~\cite{moses2023race} & Ion trap & 32 & ? & 1410 s & 6.46 ms\\
    IonQ Forte~\cite{chen2023benchmarking} & Ion trap & 32 & 10--100 s & $\sim$1s& 931 $\mu$s\\
\bottomrule
\end{tabular}
\caption{Hardware information for range of superconducting and ion trap quantum processors. T1 and T2 correspond to qubit relaxation and coherence times and two qubit gate time refers to native two-qubit gate execution times. The Quantinuum H2-1 and IonQ Forte processors, T1, T2 and two qubit gate times were calculated as described in main text of this Appendix. T1 time for the Quantinuum H2-1 could not be found. Values for other processors are median values taken from the corresponding references.\\}
\label{tab: hardware}
\centering
\begin{tabular}{ccccccccccccc}
\toprule
\multirow{2}{*}{Processor} & \multicolumn{3}{c}{Block encoding runtime (s)} & \multicolumn{3}{c}{Runtime as fraction of T1} & \multicolumn{3}{c}{Runtime as fraction of T2}\\
 & $N$=100 & $N$=1,000 & $N$=10,000 & $N$=100 & $N$=1,000 & $N$=10,000 & $N$=100 & $N$=1,000 & $N$=10,000 & \\
\hline
Eagle r3 & $6.62 \times 10^{-3}$ & $7.49 \times 10^{-1}$ & $1.06 \times 10^{1}$ & $2.41 \times 10^{1}$ & $2.72 \times 10^{3}$ & $3.86 \times 10^{4}$ & $5.66 \times 10^{1}$ & $6.40 \times 10^{3}$ & $9.08 \times 10^{4}$ \\
Sycamore & $2.08 \times 10^{-4}$ & $2.36 \times 10^{-2}$ & $3.34 \times 10^{-1}$ & $9.09 \times 10^{0}$ & $1.03 \times 10^{3}$ & $1.46 \times 10^{4}$ & $1.34 \times 10^{1}$ & $1.52 \times 10^{3}$ & $2.15 \times 10^{4}$ \\
H2-1 & $6.73 \times 10^{1}$ & $7.61 \times 10^{3}$ & $1.08 \times 10^{5}$ & - & - & - & $4.77 \times 10^{-2}$ & $5.40 \times 10^{0}$ & $7.65 \times 10^{1}$ \\
Forte & $9.69 \times 10^{0}$ & $1.10 \times 10^{3}$ & $1.55 \times 10^{4}$ & $9.69 \times 10^{0}$ & $1.10 \times 10^{2}$ & $1.55 \times 10^{3}$ & $9.69 \times 10^{0}$ & $1.10 \times 10^{3}$ & $1.55 \times 10^{4}$ \\
\bottomrule
\end{tabular}
\caption{Estimated block encoding runtime for different quantum processors assuming that two-qubit gates are executed in series for different system sizes. Runtimes are given in seconds as well as proportions of the T1 (where known) and T2 times for each processor using data from Table~\ref{tab: hardware}.\\}
\label{tab: hardware runtimes}
\begin{tabular}{ccccccccccccc}
\toprule
\multirow{2}{*}{Processor} & \multicolumn{3}{c}{Block encoding runtime (s)}  & \multicolumn{3}{c}{Runtime as fraction of T1} & \multicolumn{3}{c}{Runtime as fraction of T2}\\
& $N$=100 & $N$=1,000 & $N$=10,000 & $N$=100 & $N$=1,000 & $N$=10,000 & $N$=100 & $N$=1,000 & $N$=10,000 & \\
\hline
Eagle r3 & $6.64 \times 10^{-6}$ & $5.01 \times 10^{-5}$ & $5.33 \times 10^{-5}$ & $2.42 \times 10^{-2}$ & $1.82 \times 10^{-1}$ & $1.94 \times 10^{-1}$ & $5.68 \times 10^{-2}$ & $4.28 \times 10^{-1}$ & $4.55 \times 10^{-1}$ \\
Sycamore & $2.09 \times 10^{-7}$ & $1.58 \times 10^{-6}$ & $ 1.68 \times 10^{-6}$ & $9.12 \times 10^{-3}$ & $6.88 \times 10^{2}$ & $7.32 \times 10^{-2}$ & $1.35 \times 10^{-2}$ & $1.02 \times 10^{-1}$ & $1.08 \times 10^{-1}$ \\
H2-1 & $6.75 \times 10^{-2}$ & $5.09 \times 10^{-1}$ & $5.41 \times 10^{-1}$ & - & - & - & $4.79 \times 10^{-5}$ & $3.61 \times 10^{-4}$ & $3.84 \times 10^{-4}$ \\
Forte & $9.73 \times 10^{-3}$ & $7.34 \times 10^{-2}$ & $7.80 \times 10^{-2}$ & $9.73 \times 10^{-4}$ & $7.34 \times 10^{-3}$ & $7.80 \times 10^{-3}$ & $9.73 \times 10^{-3}$ & $7.34 \times 10^{-2}$ & $7.80 \times 10^{-2}$ \\
\bottomrule
\end{tabular}
\caption{Estimated block encoding run time for different quantum processors assuming that  two-qubit gates are maximally parallelised. Run times are given as proportions of the T1 (where known) and T2 times for each processor using data from Table~\ref{tab: hardware}.}
\label{tab: hardware runtimes 2}
\end{sidewaystable*}
\end{document}